\documentclass[12pt]{article}
\usepackage[utf8]{inputenc}
\usepackage{aliascnt}
\usepackage{amsmath}
\usepackage{amsthm}
\usepackage{amssymb}
\usepackage{accents}
\usepackage{natbib}
\usepackage{mathrsfs}
\usepackage[colorlinks,citecolor=blue,linkcolor=red]{hyperref}
\usepackage[nameinlink,capitalize]{cleveref}
\usepackage{enumerate}
\usepackage{bbm}
\usepackage{color}
\usepackage{xargs}
\usepackage{blkarray}
\usepackage{enumitem}
\usepackage{stmaryrd}
\usepackage[blocks]{authblk}
\usepackage[ruled,linesnumbered]{algorithm2e}
\usepackage{hyperref}
\usepackage{graphicx,psfrag,epsf}
\usepackage{enumerate}
\usepackage{natbib}
\usepackage{url} 
\usepackage{lipsum,afterpage}
\usepackage{subcaption}
\usepackage{multirow}
\makeatletter
\def\algbackskip{\hskip-\ALG@thistlm}
\makeatother
\makeatletter
\def\fourvdots{\vbox{\baselineskip1\p@ \lineskiplimit\z@
  \hbox{.}\hbox{.}\hbox{.}\hbox{.}}}
\makeatother

\newcommand{\rank}[1]{\operatorname{rank}\p{#1}}
\newcommand{\p}[1]{\left(#1\right)}

\newcommand{\blind}{1}

\newcommand{\pscal}[2]{\langle #1, #2  \rangle}
\newcommandx{\norm}[2][2=]{\| #1 \|_{#2}}
\newcommandx{\nint}[1]{\llbracket #1 \rrbracket}
\newcommandx{\maxi}[2][2=]{\underset{#2}{\operatorname{max}}\left\{#1\right\}}
\newcommandx{\mini}[2][2=]{\underset{#2}{\operatorname{min}}\left\{#1\right\}}
\newcommandx{\argmaxi}[2][2=]{\underset{#2}{\operatorname{argmax}}\left\{#1\right\}}
\newcommandx{\argmini}[2][2=]{\underset{#2}{\operatorname{argmin}}\left\{#1\right\}}
\newcommandx{\diff}[1]{\nabla\mathcal{L}(#1; Y,\Omega)}
\newcommandx{\fu}[1]{\mathsf{f}_U(#1)}

\newcommandx{\rk}[1]{\operatorname{rk}(#1)}
\newcommandx{\prob}[1]{\mathbb{P}\left(#1 \right) }
\newcommandx{\e}{\mathrm{e}}
\newcommandx{\PE}[1][1=]{\mathbb{E}\left[#1 \right] }
\newcommandx{\pe}[1][1=]{\mathbb{E}\left[#1 \right] }
\newcommandx{\set}[2]{\{ #1 \, ; \, 2\}}
\newcommandx{\TX}{X^0}
\newcommandx{\TS}{S^0}
\newcommandx{\TL}{L^0}
\newcommandx{\Talpha}{\alpha^0}
\newcommandx{\smax}{\sigma_{+}}
\newcommandx{\smin}{\sigma_{-}}
\newcommandx{\umax}{{u}}
\newcommandx{\R}{\mathbb{R}}
\newcommandx{\N}{\mathbb{N}}
\newcommandx{\Yset}{\mathbb{Y}}
\newcommandx{\Xset}{\mathbb{X}}
\newcommandx{\Ycurl}{\mathcal{Y}}
\newcommandx{\lone}{\|\cdot \|_{1}}
\newcommandx{\linf}{\| \cdot \|_{\infty}}
\newcommandx{\lop}{\| \cdot \|}
\newcommandx{\lnuc}{\| \cdot \|_{*}}
\newcommandx{\am}{\mathbf{a}}
\newcommandx{\um}{\mathbf{u}}
\newcommand{\xupdownarrow}[1]{%
  {\left\updownarrow\vbox to #1{}\right.\kern-\nulldelimiterspace}
}
\def\rmd{\mathrm{d}}

\def\wrt{with respect to}
\def\approxQL{\mathcal{A}}
\def\rset{\mathbb{R}}
\def\eqsp{\;}
\def\eqdef{:=}

\def\argmin{\operatorname{argmin}}
\usepackage[textsize=footnotesize]{todonotes} 

\begin{document}
\newcommand{\ubar}[1]{\underaccent{\bar}{#1}}
\newtheorem{theorem}{Theorem}
\newtheorem{assumption}{\textbf{H}\hspace{-0.15cm}}
\Crefname{assumption}{\textbf{H}}{\textbf{H}\hspace{-0.5cm}}
\crefname{assumption}{\textbf{H}}{\textbf{H}\hspace{-0.5cm}}

\newtheorem{lemma}{Lemma}
\Crefname{lemma}{Lemma}{Lemma}
\crefname{lemma}{lemma}{lemma}

\newtheorem{corollary}{Corollary}
\Crefname{corollary}{Corollary}{Corollary}
\crefname{corollary}{corollary}{corollary}
\newtheorem{remark}{Remark}
\Crefname{remark}{Remark}{Remark}
\crefname{remark}{remark}{remark}

\newtheorem{proposition}{Proposition}
\Crefname{proposition}{Proposition}{Proposition}
\crefname{proposition}{proposition}{proposition}

\newtheorem{definition}{Definition}
\Crefname{definition}{Definition}{Definition}
\crefname{definition}{definition}{definition}

\newtheorem{example}{Example}
\Crefname{example}{Example}{Example}
\crefname{example}{example}{example}

\def\spacingset#1{\renewcommand{\baselinestretch}%
{#1}\small\normalsize} \spacingset{1}


  \title{\bf Main effects and interactions in mixed and incomplete data frames}
  \author[1,2]{Genevi\`{e}ve Robin}
  \author[3,4]{Olga Klopp}
  \author[1,2]{Julie Josse}
\author[1,2,5]{\'{E}ric Moulines, \thanks{This work has been funded by the DataScience Inititiative (Ecole Polytechnique) and the  Russian Academic Excellence Project '5-100'}}  
\author[6,7]{Robert Tibshirani}
\affil[1]{Centre de Math\'ematiques Appliquées, \'{E}cole Polytechnique, France }
\affil[2]{Projet XPOP, INRIA}
\affil[3] {ESSEC Business School}
\affil[4] {CREST, ENSAE}
\affil[5]{National Research University Higher School of Economics, Russia}
\affil[6]{Department of Statistics, Stanford University}
\affil[7]{Department of Biomedical Data Science, Stanford University}

  \maketitle

\bigskip
\begin{abstract}
A mixed data frame (MDF) is a table collecting categorical, numerical and count observations. The use of MDF is widespread in statistics and the applications are numerous from  abundance data in ecology to recommender systems. In many cases, an MDF exhibits simultaneously \textit{main effects}, such as row, column or group effects and \textit{interactions}, for which a low-rank model has often been suggested. Although the literature on low-rank approximations is very substantial, with few exceptions, existing methods do not allow to incorporate main effects and interactions while providing statistical guarantees. The present work fills this gap. 

We propose an estimation method which allows to recover simultaneously the main effects and the interactions. We show that our method is near optimal under conditions which are met in our targeted applications. We also propose an optimization algorithm which provably converges to an optimal solution. Numerical experiments reveal that our method, \texttt{mimi}, performs well when the main effects are sparse and the interaction matrix has low-rank. We also show that \texttt{mimi} compares favorably to existing methods, in particular when the main effects are significantly large compared to the interactions, and when the proportion of missing entries is large. The method is available as an R package on the Comprehensive R Archive Network.\end{abstract}

\noindent%
\vfill

\newpage
\spacingset{1.45} 

\section{Introduction}

Mixed data frames (MDF) (see \cite{Pages2015,glrm}) are tables collecting categorical, numerical and count data.
In most applications, each row is an example or a subject and each column is a feature or an attribute. A distinctive characteristic of MDF is that column entries may be of different types and most often many entries are missing. MDF appear in numerous applications including patient records in health care (survival values at different time points, quantitative and categorical clinical features like blood pressure, gender, disease stage, see, \textit{e.g.}, \cite{health-mixed}), survey data \cite[Chapters 5 and 6]{survey}, abundance tables in ecology \citep{legendre4corner}, and recommendation systems \citep{agarwal2011}.

\subsection{Main effects and interactions}
In all these applications, data analysis is often made in the light of additional information, such as sites and species traits in ecology, or users and items characteristics in recommendation systems. This caused the introduction of the two central concepts of interest in this article: \textit{main effects} and \textit{interactions}. This terminology is classically used to distinguish between effects of covariates on the observations which are independent of the other covariates (main effects), and effects of covariates on the observations which depend on the value of one or more other covariates (interactions). For example, in health care, a treatment might extend survival for all patients -- this is a main effect -- or extend survival for young patients but shorten it for older patients -- this is an interaction.

Many statistical models have been developed to analyze such types of data. Abundance tables counting species across environments are for instance classically analyzed using the log-linear model \citep[Chapter 4]{Agresti13}. This model decomposes the logarithms of the expected abundances into the sum of species (rows) and environment (columns) effects, plus a low-rank interaction term.
Other examples include multilevel models \citep{gelman_multi} to analyze hierarchically structured data where examples (patients, students, etc.) are nested within groups (hospitals, schools, etc.). 

\subsection{Generalized low-rank models}
At the same time, low-rank models, which embed rows and columns into low-dimensional spaces, have been widely used for exploratory analysis of MDF \citep{Kiers1991, Pages2015, glrm}.
Despite the abundance of results in low-rank matrix estimation (see \cite{review-low-rank} for a literature survey), to the best of our knowledge, most of the existing methods for MDF analysis do not provide a statistically sound way to account for \textit{main effects} in the data. In most applications, estimation of main effects in MDF has been done heuristically as a preprocessing step \citep{softImpute, glrm,Landgraf15}.
\cite{fithian2013scalable} incorporate row and column covariates in their model, but mainly focus on optimization procedures and did not provide statistical guarantees concerning the main effects. \cite{Mao} propose a procedure to estimate jointly main effects and a low-rank structure -- which can be interpreted as interactions --, but the procedure is based on a least squares loss, and is therefore not suitable to mixed data types.

On the other hand, several approaches to model non-Gaussian, and particularly discrete data are available in the matrix completion literature, but they do not consider main effects. \cite{Davenport2012} introduced one-bit matrix completion, where the observations are binary such as yes/no answers, and provide nearly optimal upper and lower bounds on the mean square error of estimation. One-bit matrix completion was also studied in \cite{Cai:2013:MCM:2567709.2627673}. In \cite{KloppLafond2015}, the authors introduce multinomial matrix completion, where the observations are allowed to take more than two values, such as ratings in recommendation systems, and propose a minimax optimal estimator. Unbounded non-Gaussian observations have also been studied before.
For instance, \cite{Cao2016} extended the approach of \cite{Davenport2012} to Poisson matrix completion, and \cite{Gunasekar:2014:EFM:3044805.3045106} and \cite{Lafond2015} both studied exponential family matrix completion. 

\subsection{Contributions}
In the present paper we propose a new framework for incomplete and mixed data which allows to account for main effects and interactions. Before introducing a general model for MDF with sparse main effects and low-rank interactions, we start in \Cref{ex} with a concrete example from survey data analysis. Then, we propose in \Cref{estim} an estimation procedure based on the minimization of a doubly penalized negative quasi log-likelihood. We also propose a block coordinate gradient descent algorithm to compute our estimator, and prove its convergence result. In \Cref{up-bounds} we discuss the statistical guarantees of our procedure and provide upper bounds on the estimation errors of the sparse and low-rank components. To assess the tightness of our convergence rates, in \Cref{low-bound}, we derive lower bounds and show that, in a number of situations, our upper bounds are near optimal. In \Cref{examples}, we specialize our results to three examples of interest in applications. 

To support our theoretical claims, numerical results are presented in \Cref{sec:numerical-results}. In \Cref{subsec:simulated-data}, we provide the results of our experiments that show that our method "mimi" (\textbf{m}ain effects and \textbf{i}nteractions in \textbf{m}ixed and \textbf{i}ncomplete data frames) performs well when the main effects are sparse and the interactions are low-rank. In case of model mis-specification, mimi gives similar results to a two-step procedure where main effects and interactions are estimated separately. Then, in \Cref{subsec:mixed-data}, we compare mimi to existing methods for mixed data imputation. Our experiments reveal that mimi compares favorably to competitors, in particular, when the main effects are significantly large compared to the interactions, and when the proportion of missing entries is large. Finally, in \Cref{subsec:american-community-survey}, we illustrate the method with the analysis of a census data set. 
\if0\blind
{The method is implemented in the R \citep{R} package \href{https://CRAN.R-project.org/package=mimi}{\texttt{mimi}} available on the Comprehensive R Archive Network; the proofs and additional experiments are postponed to the supplementary material.
}\fi
\if1\blind
{The method is implemented in the R \citep{R} package available on the Comprehensive R Archive Network; the proofs and additional experiments are postponed to the supplementary material.
}\fi

\paragraph{Notation}
We denote the Frobenius norm on $\mathbb{R}^{m_1\times m_2}$ by $\norm{\cdot}[F]$, the operator norm by $\norm{\cdot}$, the nuclear norm by $\norm{A}[*]$  and the sup norm $\norm{\cdot}[\infty]$. $\norm{\cdot}[2]$ is the usual Euclidean norm, $\norm{\cdot}[0]$ the number of non zero coefficients, and $\norm{\cdot}[\infty]$ the infinity norm. For $n\in\N$, denote $\nint{n} = \{1,\ldots,n\}$.  We denote the support of $\alpha\in \mathbb{R}^N$ by $\operatorname{supp}(\alpha) = \{k\in\nint{N}, \alpha_k\neq 0\}$. For $I\subseteq\nint{1}{m_1}$, we denote $\mathbbm{1}_{I}$, defined by $\mathbbm{1}_{I}(i) = 1$ if $i\in I$ and $0$ otherwise, the indicator of set $I_h$.

\section{General model and examples}
\label{ex}
\subsection{American Community Survey}
\label{sec:acs-intro}
Before introducing our general model, we start by giving a concrete example.
The American Community Survey\footnote{https://www.census.gov/programs-surveys/acs/about.html} (ACS) provides detailed information about the American people on a yearly basis. Surveyed households are asked to answer 150 questions about their employment, income, housing, etc. As shown in \Cref{tab:ACS}, this results in a highly heterogeneous and incomplete data collection.
\begin{table}[ht]
\footnotesize
\centering
\begin{tabular}{rrrrrrrrrrrrrrrrrrrrrr}
  \hline
ID & Nb. people & Electricity bill (\$) & Food Stamps & Family Employment Status & Allocation \\
\hline
  1 &   2  & 160 &   No &   Married couple, neither employed &   Yes \\
  2 &   1  & 390 &   No &   \textbf{NA}  &   No \\
  3 &   4  & \textbf{NA} &   No &   Married couple, husband employed &   No \\
  4 &   2  & 260 &   No &   Married couple, neither employed &   No \\
  5 &   2  & 100 &   No &   Married couple, husband employed &   No \\
  6 &   2  & 130 &   No & \textbf{NA} &   No \\
   \hline
\end{tabular}
\caption{American Community Survey: Excerpt of the 2016 public use microsample data.}
\label{tab:ACS}
\end{table}

Here, the Family Employment Status (FES) variable categorizes the surveyed population in groups, depending on whether the household contains a couple or a single person, and whether the householders are employed. In an exploratory data analysis perspective, a question of interest is: does the household category influence the value of the other variables? For example income, food stamps allocation, etc. Furthermore, as we do not expect the group effects to be sufficient to explain the observations, can we also model residuals, or \textit{interactions}?

Denote $Y=(Y_{ij})$ the data frame containing the households in rows and the questions in columns. If the $j$-th column is continuous (electricity bill for instance), one might model the group effects and interactions as follows:
\[
 \mathbb{E}[Y_{ij}] = \Talpha_{c(i)j} + \TL_{ij},
\]
where $c(i)$ indicates the group to which individual $i$ belongs, and $\Talpha_{c(i)j}$ and $\TL_{ij}$ are fixed group effects and interactions respectively. 
This corresponds to the so-called \textit{multilevel regression} framework \citep{gelman_multi}. If the $j$-th column is binary (food stamps allocation for instance), one might model
$$ \mathbb{P}(Y_{ij} = \text{``Yes"}) = \frac{\mathrm{e}^{\TX_{ij}}}{1+\e^{\TX_{ij}}},\quad \TX_{ij} = \Talpha_{c(i)j} + \TL_{ij},$$
corresponding to a logistic regression framework.

The goal is then, from the mixed and incomplete data frame $Y$, to estimate simultaneously the vector of group effects $\Talpha$ and the matrix of interactions $\TL$. We propose a method assuming the vector of main effects $\Talpha$ is sparse and the matrix of interactions $\TL$ has low-rank. The sparsity assumption means that groups affect a small number of variables. On the other hand, the low-rank assumption means the population can be represented by a few archetypical individuals and summary features \cite[Section 5.4]{glrm}, 
which interact in a multiplicative manner. In fact, if $\TL$ is of rank $r$, then it can be decomposed as the sum of $r$ rank-$1$ matrices as follows:
$$\TL = \sum_{k=1}^r u_kv_k^{\top},$$
where $u_k$ (resp. $v_k$) is a vector of $\mathbb{R}^{m_1}$ (resp. $\mathbb{R}^{m_2}$). Thus, using the above example, we obtain
$$ \mathbb{E}[Y_{ij}] = \Talpha_{c(i)j} + \sum_{k=1}^r u_{ik}v_{jk},$$
where the last term $\sum_{k=1}^r u_{ik}v_{jk}$ can be interpreted as the sum of multiplicative interaction terms between latent individual types and features.

\subsection{General model}
\label{subsec:gen-model}
We now introduce a new framework generalizing the above example to other types of data and main effects. Consider an MDF $Y = (Y_{ij})$ of size $m_1\times m_2$. The entries in each column $j\in\nint{m_2}$  belong to an observation space, denoted $\Yset_j$. For example, for numerical data, the observation space is $\Yset_j = \R$, and for count data, $\Yset_j = \N$ is the set of natural integers. For binary data, the observation space is $\Yset_j = \{0,1\}$. In the entire paper, we assume that the random variables $(Y_{ij})$ are independent and that for each $(i,j)\in\nint{m_1}\times\nint{m_2}$, $Y_{ij} \in \Yset_j$ and $\PE[|Y_{ij}|]<\infty$. Furthermore, we will assume that $Y_{ij}$ is sub-exponential with scale $\gamma$ and variance $\sigma^2$: for all $(i,j)\in\nint{m_1}\times\nint{m_2}$ and $|z|<\gamma$,
$\PE[\e^{z(Y_{ij} - \PE[Y_{ij}])} ]\leq  \e^{\sigma^2z^2/2}.$\\

In our estimation procedure, we will use a data-fitting term based on heterogeneous exponential family quasi-likelihoods.
Let $(\Yset, \Ycurl,\mu)$ be a measurable space, $h: \Yset \to \R_+$, and  $g: \R \to \R$ be functions. Denote by $\operatorname{Exp}^{(h,g)}= \{ f^{(h,g)}_x \, : \, x \in \R \}$  the canonical exponential family. Here, $h$ is the base function, $g$ is the link function, and $f^{(h,g)}_x$ is the density  \wrt\ the base measure $\mu$ given by
\begin{equation}
\label{exp-fam}
f^{(h,g)}_{x} (y) = h(y)\exp\left(y x  - g(x)\right),
\end{equation}
for $y\in\Yset$. For simplicity, we assume $\int h(y)\exp(yx)\mu(\rmd y)<\infty$ for all $x\in \R$.\\

The exponential family is a flexible framework for different data types.
For example, for numerical data, we set $g(x)=x^2 \sigma^2/2$ and $ h(y)= (2 \pi \sigma^2)^{-1/2} \exp(-y^2/\sigma^2)$.
In this case, $\operatorname{Exp}^{(h,g)}$ is the family of  Gaussian distributions with mean $\sigma^2x$ and variance $\sigma^2$.
For count data, we set $g(x)= \exp(a x)$ and $h(y)= 1/y!$, where $a \in \R$. In this case, $\operatorname{Exp}^{(h,g)}$ is the family of Poisson distributions with intensity $\exp(a x)$.
For binary data, $g(x)= \log(1+\exp(x))$ and $h(y)= 1$. Here, $\operatorname{Exp}^{(h,g)}$  is the family of Bernoulli distributions with success probability $1/(1+\exp(-x))$.\\

In our estimation procedure, we choose a collection $\{ (g_j, h_j),~j \in \nint{m_2} \}$ of link functions and base functions corresponding to the observation spaces $\{(\Yset_j, \Ycurl_j, \mu_j) \, , \, j \in \nint{m_2}\}$. For each $(i,j) \in  \nint{m_1} \times \nint{m_2}$, we denote by $X^0_{ij}$ the value of the parameter minimizing the divergence between the distribution of $Y_{ij}$ and the exponential family $\operatorname{Exp}^{(h_j,g_j)}$, $j \in \nint{m_2}$:
\begin{equation}
\label{exp-model}
\TX_{ij} = \operatorname{argmin}_{x\in\R} \left\lbrace - \PE[ Y_{ij}]x + g_j(x)\right\rbrace.
\end{equation}
To model main effects and interactions we assume the matrix of parameters $\TX = (\TX_{ij})\in\mathbb{R}^{m_1\times m_2}$ can be decomposed as the sum of sparse main effects and low-rank interactions:
\begin{equation}
\label{eq:X_decomp}
\TX = \sum_{k=1}^N\Talpha_kU^k+\TL.
\end{equation}
Here, $\mathcal{U}=(U^1,\ldots,U^N)$ is a fixed dictionary of $m_1\times m_2$ matrices, $\Talpha$ is a sparse vector with unknown support $\mathcal{I} = \{k\in\nint{N}; \Talpha_k\neq 0\}$ and $\TL$ is an $m_1\times m_2$ matrix with low-rank.
The decomposition introduced in \eqref{eq:X_decomp} is a general model combining regression on a dictionary and low-rank design.

\subsection{Low-rank plus sparse matrix decomposition}
Such decompositions have been studied before in the literature. In particular, a large body of work has tackled the problem of reconstructing a sparse and a low-rank terms exactly from the observation of their sum. \cite{chandra} derived identifiability conditions under which exact reconstruction is possible when the sparse component is \textit{entry-wise sparse}; the same model was also studied in \cite{Hsu_robustmatrix}. \cite{Candes:2011:RPC} proved a similar results for entry-wise sparsity, when the location of the non-zero entries are chosen uniformly at random. \cite{Xu:2010:RPV} extended the model to study \textit{column-wise sparsity}. \cite{mardani2013} studied an even broader framework with general sparsity pattern and determined conditions under which exact recovery is possible. \\
In the present paper, we consider the problem of estimating a (general) sparse component and a low-rank term from noisy and incomplete observation of their sum, when the noise is heterogeneous and in the exponential family. Because of this noisy setting, we can not recover the two components exactly. Thus, we do not require strong identifiability conditions as those derived in \citep{chandra,Hsu_robustmatrix,Candes:2011:RPC,Xu:2010:RPV,mardani2013}. However, since decomposition \eqref{eq:X_decomp} may not be unique, we restrict our model to the following class of possible decompositions, to which our estimator will be the closest. From all possible decompositions $(\alpha, L)$, consider $(\alpha^1,L^1)$ such that
 \begin{equation}
 \label{eq1}
 (\alpha^1,L^1)\in \operatorname{argmin}_{X^0=\sum\alpha_kU^k + L} \{\left\|\alpha\right\|_0 + \operatorname{rank} L\}.
 \end{equation}
 Let $s^1=\left\|\alpha^1\right\|_0 + \operatorname{rank} L^1$. Finally let
 \begin{equation}
 \label{eq2}
 (\Talpha,\TL)\in \operatorname{argmin}_{\substack{X^0=\sum\alpha_kU^k + L\\\left\|\alpha\right\|_0 + \operatorname{rank} L=s^1}} \left\|\alpha\right\|_0 .
 \end{equation}
The decomposition satisfying \eqref{eq1} and \eqref{eq2} may also not be unique. Assume that there exists a pair $(\alpha^{\star},L^{\star}) \neq(\Talpha,\TL)$ satisfying \eqref{eq1} and \eqref{eq2}. Then,
\begin{equation*}
\begin{aligned}
&\left\|L^{\star}-L^0 \right\|_F =\left\|\sum_k\alpha_k^{\star}U^k-\sum_k\alpha^0_kU^k \right\|_F&&\leq 2a\left\|\Talpha\right\|_0\max_k\left\|U_k\right\|_2=R,
\end{aligned}
\end{equation*}
with $a$ an upper bound on $\norm{\Talpha}[\infty]$. This implies that for all such possible decompositions $(\alpha^{\star}, L^{\star})$ we have that $L^{\star}$ and $\sum_k\alpha_k^{\star}U^k$ are in the small balls of radius $R$ and centered at $\TL$ and $\sum_k\Talpha_kU^k$ respectively. Our statistical guarantees in \Cref{main-results} show that our estimators of $\TL$ and $\sum_k\Talpha_kU^k$ are in balls of radius at least $R$, and also centered at $\TL$ and $\sum_k\Talpha_kU^k$. Moreover, we also show that this error bound is minimax optimal in several situations. To summarize, in our model the decomposition may not be unique, but all the possible decompositions are in a neighborhood of radius smaller than the optimal convergence rate.

\subsection{Examples}
We now provide three examples of dictionaries which can be used to model classical main effects.

\begin{example}
\normalfont
{\textbf{Group effects}}
\label{ex:groups}
We assume the $m_1$ individuals are divided into $H$ groups. For $h\in\nint{H}$ denote by $I_h \subset \nint{m_1}$ the $h$-th group containing $n_h$ individuals. The size of the dictionary is $N=H m_2$ and its elements are, for all $(h,q)\in\nint{H}\times\nint{m_2}$, $U_{h,q} = (\mathbbm{1}_{I_h}(i) \mathbbm{1}_{\{q\}}(j))_{(i,j)\in\nint{m_1}\times\nint{m_2}}$.
This example corresponds to the model discussed in \Cref{sec:acs-intro}; we develop it further in \Cref{sec:numerical-results} with simulations and a survey data analysis.
\end{example}

\begin{example}
\normalfont
{\textbf{Row and column effects}}
\label{ex:row-column}
(see \textit{e.g.} \citep[Chapter 4]{Agresti13}) Another classical model is the log-linear model for count data analysis. Here, $Y$ is a matrix of counts. Assuming a Poisson model, the parameter matrix $\TX$, which satisfies $\PE[Y_{ij}] = \exp(\TX_{ij})$ for all $(i,j)\in\nint{m_1}\times \nint{m_2}$, is assumed to be decomposed as follows:
\begin{equation}
\label{log-lin}
\TX_{ij} = (\alpha_r^0)_i+(\alpha_c^0)_j+\TL_{ij},
\end{equation}
where $\alpha_r^0\in\mathbb{R}^{m_1}$, $\alpha_c^0\in\mathbb{R}^{m_2}$ and $\TL\in \mathbb{R}^{m_1\times m_2}$ is low-rank.
 This model is often used to analyze abundance tables of species across environments (see, \textit{e.g.}, \cite{10.7717/peerj.2885}). In this case the low-rank structure of $\TL$ reflects the presence of groups of similar species and environments. Model \eqref{log-lin} can be re-written in our framework as
\begin{equation*}
\TX = \sum_{k=1}^{N}\Talpha_k U_k+\TL,
\end{equation*}
with $\Talpha = (\alpha_r^0, \alpha_c^0)$, $N = m_1 + m_2$ and where for $i\in\nint{m_1}$ and $j\in\nint{m_2}$ we have
$U_{i} = (\mathbbm{1}_{\{i\}}(k))_{(k,l)\in\nint{m_1}\times \nint{m_2}}$  and $U_{m_1+j} = (\mathbbm{1}_{\{j\}}(l))_{(k,l)\in\nint{m_1}\times \nint{m_2}}$.
\end{example}

\begin{example}
\label{ex:corruptions}
\normalfont
{\bf Corruptions}
Our framework also embeds the well-known robust matrix completion problem \citep{Hsu_robustmatrix, Candes:2011:RPC, Klopp2017} which is of interest, for instance, in recommendation systems.
In this application, malicious users coexist with normal users, and introduce spurious perturbations.
Thus, in robust matrix completion, we observe noisy and incomplete realizations of a low-rank matrix $\TL$ 
of fixed rank and containing zeros at the locations of malicious users, 
perturbed by corruptions. 
The sparse component corresponding to corruptions is denoted $\sum_{(i,j)\in \mathcal{I}}\Talpha_kU_{i,j}$, where the $U_{i,j}$, $(i,j)\in\nint{m_1}\times \nint{m_2}$, are the matrices of the canonical basis of $\mathbb{R}^{m_1\times m_2}$ $U_{i,j} = (\mathbbm{1}_{\{i\}}(k) \mathbbm{1}_{\{j\}}(l))_{(k,l)\in\nint{m_1}\times\nint{m_1}}$ and
$\mathcal{I}$ is the set of indices of corrupted entries.  
Thus, the non-zero components of $\Talpha$ correspond to the locations where the malicious users introduced the corruptions.
 For this example, the particular case of quadratic link functions $g_j(x) = x^2/2$ was studied in \cite{Klopp2017}. We generalize these results in two directions: we consider mixed data types and general main effects.
\end{example}

\subsection{Missing values}
Finally, we consider a setting with missing observations. Let $\Omega = (\Omega_{ij})$ be an observation mask with $\Omega_{ij} = 1$ if $Y_{ij}$ is observed and $\Omega_{ij} = 0$ otherwise. We assume that $\Omega$ and $Y$ are independent, i.e. a Missing Completely At Random (MCAR) scenario \citep{Little02}: $(\Omega_{ij})$ are independent Bernoulli random variables with probabilities $\pi_{ij}$, $(i,j) \in \nint{m_1}\times \nint{m_2}$.  Furthermore for all $(i,j) \in \nint{m_1}\times \nint{m_2}$, we assume there exists $p>0$ allowed to vary with $m_1$ and $m_2$, such that
\begin{equation}
\label{eq:mcar}
\pi_{ij}\geq p.
\end{equation}
For $j\in\nint{m_2}$, denote by $\pi_{.j} = \sum_{i=1}^{m_1}\pi_{ij}$, $j\in \nint{m_2}$ the probability of observing an element in the $j$-th column. Similarly, for $i\in\nint{m_1}$, denote by $\pi_{i.} = \sum_{j=1}^{m_2}\pi_{ij}$ the probability of observing an element in the $i$-th row. We define the following upper bound:
\begin{equation}
\label{eq:beta}
\max_{i,j}(\pi_{i.},\pi_{.j})\leq \beta.
\end{equation}

\section{Estimation procedure}
\label{estim}
Consider the data-fitting term defined by the heterogeneous exponential family negative quasi log-likelihood
\begin{equation}
\label{eq:neg-log-lik}
\mathcal{L}(X;Y,\Omega) = \sum_{i=1}^{m_1}\sum_{j=1}^{m_2}\Omega_{ij}\left\{-Y_{ij}X_{ij}+g_j(X_{ij}) \right\},
\end{equation}
and define the function
\begin{equation}
\label{eq:function-l}
f(\alpha,L)=\mathcal{L}(\fu{\alpha}+L;Y,\Omega),
\end{equation}
where for $\alpha\in\R^N$,
$\fu{\alpha} = \sum_{k=1}^N\alpha_kU_k$. We assume  $\norm{\Talpha}[\infty]\leq a$ and $\norm{\TL}[\infty]\leq a$ where $a > 0$ is a known upper bound.
We use the nuclear norm $\norm{\cdot}[*]$ (the sum of singular values) and $\ell_1$ norm $\norm{\cdot}[1] $ penalties as convex relaxations of the rank and sparsity constraints respectively:
\begin{align}
\label{eq:estimator}
&(\hat{\alpha},\hat{L})  \in  \operatorname{argmin}_{(\alpha,L)} F(\alpha, L)\\
& \text{s. t. } \norm{\alpha}[\infty]\leq a, \norm{L}[\infty]\leq a,
\end{align}
\begin{equation}
\label{eq:function-F}
F(\alpha,L)=f(\alpha,L) + \lambda_1\norm{L}[*] + \lambda_2\norm{\alpha}[1],
\end{equation}
with $\lambda_1>0$ and $\lambda_2>0$. In the sequel, for all $(\hat{\alpha},\hat{L})$ in the set of solutions, we denote by $\hat{X} = \fu{\hat\alpha}  + \hat{L}$.

\subsection{Block coordinate gradient descent (BCGD)}
\label{subsec:algo}
To solve \eqref{eq:estimator} we develop a block coordinate gradient descent algorithm where the two components $\alpha$ and $L$ are updated alternatively in an iterative procedure. At every iteration, we compute a (strictly convex) quadratic approximation of the data fitting term and apply block coordinate gradient descent to generate a search direction.  
This BCGD algorithm is a special instance of the coordinate gradient descent method for non-smooth separable minimization developed in \cite{tseng:yun:2009}.

 Note that the upper bound on $\norm{\alpha}[\infty]$ and $\norm{L}[\infty]$ is required to derive the statistical guarantees and, for simplicity, we did not implement it in practice. That is, we solve the following relaxed problem:
\begin{align}
\label{eq:estimator2}
(\hat{\alpha},\hat{L})  \in & \operatorname{argmin}_{(\alpha,L)} F(\alpha, L).
\end{align}
\paragraph{Quadratic approximation.} For any $(\alpha, L)\in\mathbb{R}^N\times\mathbb{R}^{m_1\times m_2}$ and for any direction $(d_{\alpha}, d_L)\in\mathbb{R}^N\times\mathbb{R}^{m_1\times m_2}$, consider the following local approximation of the data fitting term
\begin{equation}
\label{eq:quad-approx0}
f(\alpha+d_{\alpha}, L+d_L)= f(\alpha,L)+  \approxQL(\fu{\alpha}+L,d_{\alpha},d_L) + o(\norm{d_{\alpha}}[2]^2+\norm{d_L}[F]^2) \eqsp,
\end{equation}
where we have set
\begin{multline}
\label{eq:quad-approx}
\approxQL(X,d_{\alpha},d_L)= -2\sum_{i=1}^{m_1} \sum_{j=1}^{m_2} w_{ij}[X_{ij}] Z_{ij}[X_{ij}]  (\fu{d_{\alpha}}_{ij} + {d_L}_{ij}) \\+ \sum_{i=1}^{m_1} \sum_{j=1}^{m_2} w_{ij}[X_{ij}](\fu{d_{\alpha}}_{ij} + {d_L}_{ij})^2
 + \nu\norm{d_{\alpha}}[2]^2+\nu\norm{d_L}[F]^2.
\end{multline}
In \eqref{eq:quad-approx}, $\nu>0$ is a positive constant and for $x \in \rset$ and $(i,j) \in \nint{m_1} \times \nint{m_2}$,
\begin{equation}
\label{eq:definition-w-Z}
w_{ij}[x] = \Omega_{ij} g_j''(x)/2 \eqsp, \quad Z_{ij}[x] = (Y_{ij} - g_j'(x))/g_j''(x) \eqsp.
\end{equation}
Note that the approximation \eqref{eq:quad-approx} is simply a Taylor expansion of $\mathcal{L}$ around $X$, with an additional quadratic term $\nu\norm{d_{\alpha}}[2]^2+\nu\norm{d_L}[F]^2$ ensuring its strong convexity. Denote by $(\alpha^{[t]},L^{[t]})$ the fit of the parameter at iteration $t$ and set $X^{[t]}=\fu{\alpha^{[t]}}+L^{[t]}$. We update $\alpha$ and $L$ alternatively as follows.
\paragraph{$\alpha$-Update.} We first solve
\begin{equation}
\label{eq:search-alpha}
d_{\alpha}^{[t]} \in  \argmin_{d \in \rset^N} \left\{ \approxQL(X^{[t]},d,0) + \lambda_2 \norm{\alpha^{[t]}+d}[1] 
\right\} \eqsp.
\end{equation}
Problem \eqref{eq:search-alpha} may be rewritten as a weighted Lasso problem:
\begin{equation*}
\label{eq:search-alpha-1}
\argmin_{\alpha \in \rset^d} \sum_{i=1}^{m_1} \sum_{j=1}^{m_2} w_{ij}[X_{ij}^{[t]}](Z_{ij}^{[t]} - [\fu{\alpha}]_{ij})^2+\nu\norm{\alpha^{[t]}-\alpha}[2]^2 + \lambda_2 \norm{\alpha}[1]\eqsp,
\end{equation*}
where for $i,j \in \nint{m_1} \times \nint{m_2}$ we have set $Z_{ij}^{[t]} \eqdef Z_{ij}[X_{ij}^{[t]}] + \fu{\alpha^{[t]}}$.
Efficient numerical solutions to this problem are available (see, e.g., \cite{friedman2010regularization}). 
To update $\alpha^{[t]}$, we select a step size with an Armijo line search. The procedure goes as follows. We choose $\tau_{\text{init}} > 0$ and we let $\tau^{[t]}_{\alpha}$ be the largest element of $\{\tau_{\text{init}} \beta^j \}_{j=0}^\infty$ satisfying
\begin{equation*}
f(\alpha^{[t]}+\tau_{\alpha}^{[t]} d^{[t]}, L^{[t]})+\lambda_2\norm{\alpha^{[t]}+\tau_{\alpha}^{[t]} d^{[t]}}[1] 
\leq f(\alpha^{[t]}, L^{[t]})+\lambda_2\norm{\alpha^{[t]}}[1] + \tau^{[t]}_{\alpha} \zeta \Gamma^{[t]}_{\alpha},
\end{equation*}
where $0 < \beta < 1$, $0 < \zeta < 1$, $0 \leq \theta < 1$, and
\begin{multline*}
\Gamma^{[t]}_{\alpha} \eqdef -2 \sum_{i=1}^{m_1} \sum_{j=1}^{m_2} w_{ij}[X_{ij}^{[t]}] Z_{ij}[X_{ij}^{[t]}] [\fu{d^{[t]}}]_{ij} + \theta \sum_{i=1}^{m_1} \sum_{j=1}^{m_2} w_{ij}[X_{ij}^{[t]}] \left[ \fu{d^{[t]}} \right]_{ij}^2 + \nu\norm{d^{[t]}}[2]^2\\
+ \lambda_2 \left\{ \norm{\alpha^{[t]} + d^{[t]}}[1] - \norm{\alpha^{[t]}}[1] \right\}.
\end{multline*}
We set $\alpha^{[t+1]}=\alpha^{[t]}+\gamma^{[t]} d^{[t]}_{\alpha}$ and $X^{[t+1/2]}= \fu{\alpha^{[t+1]}}+L^{[t]}$. 
\paragraph{$L$-Update.} We first solve
\begin{equation}
\label{eq:search-D}
d_L^{[t]} \eqdef \argmin_{d \in \rset^{m_1 \times m_2}} \left\{ \approxQL(X^{[t+1/2]},0,d) + \lambda_1 \norm{L^{[t]}+d}[*] 
\right\},
\end{equation}
which is equivalent to
\begin{equation}
\label{eq:search-D1}
\argmin_{L \in \rset^{m_1 \times m_2}} \sum_{i=1}^{m_1} \sum_{j=1}^{m_2} (\nu+w_{ij}[X_{ij}^{[t+1/2]}]) (Z_{ij}^{[t+1/2]}-L_{ij})^2 + \lambda_1 \norm{L}[*],
\end{equation}
where for $i,j \in \nint{m_1} \times \nint{m_2}$ we have set $$Z_{ij}^{[t+1/2]}=  \frac{w_{ij}[X_{ij}^{[t+1/2]}](Z_{ij}[X^{[t+1/2]}]+L^{[t]}_{ij})+\nu L^{[t]}_{ij}}{\nu + w_{ij}[X_{ij}^{[t+1/2]}]}.$$
The minimisation problem \eqref{eq:search-D1} may be seen as a weighted version of \texttt{softImpute} \citep{softImpute}. \cite{Srebro:2003:WLA} proposed to solve \eqref{eq:search-D1} using an EM algorithm where the weights in $(0,1]$ are viewed as frequencies of observations in a missing value framework (see also \cite{mazumder2010spectral}). We use this procedure, which involves soft-thresholding of the singular values of $L$, by adapting the \href{https://CRAN.R-project.org/package=softImpute}{\texttt{softImpute}} package \citep{softImpute}.
To update $L^{[t]}$, we choose the step size using again the Armijo line search.
We set $\tau_{\text{init}} > 0$ and  let $\tau_L^{[t]}$ be the largest element of $\{\tau_{\text{init}} \beta^j \}_{j=0}^\infty$ satisfying
\begin{equation*}
f(\alpha^{[t+1]},L^{[t]}+\tau_L^{[t]}d_L^{[t]})+\lambda_1\norm{L^{[t]}+\gamma^{[t]}d_L^{[t]}}[*]\\ \leq f(\alpha^{[t+1]},L^{[t]})+\lambda_1\norm{L^{[t]}}[*] + \tau_L^{[t]} \zeta \Gamma_L^{[t]},
\end{equation*}
\begin{multline*}
\Gamma_L^{[t]} \eqdef -2 \sum_{i=1}^{m_1} \sum_{j=1}^{m_2} w_{ij}[X_{ij}^{[t+1/2]}] Z_{ij}[X_{ij}^{[t+1/2]}] {d_L^{[t]}}_{ij} + \theta \sum_{i=1}^{m_1} \sum_{j=1}^{m_2} w_{ij}[X_{ij}^{[t+1/2]}] { {d_L^{[t]}}_{ij}}^2 \\
+ \lambda_1 \left\{ \norm{L^{(t)} + d_L^{[t]}}[*] - \norm{d_L^{[t]}}[*] \right\} \eqsp.
\end{multline*}
We finally set $L^{[t+1]}= L^{[t]}+ \tau_L^{[t]} d_L^{[t]}$.

\subsection{Convergence of the BCGD algorithm}
\label{subsec:cvg-algo}
The algorithm described in \Cref{subsec:algo} is a particular case of the coordinate gradient descent method for nonsmooth minimisation introduced in \cite{tseng:yun:2009}. In the aforementioned paper, the authors studied the convergence of the iterate sequence to a stationary point of the objective function. Here, we apply their general result \cite[Theorem~1]{tseng:yun:2009} to our problem to obtain global convergence guarantees. Consider the following assumption on the dictionary $\mathcal{U}$.
\begin{assumption}
\label{ass:dict} For all $k\in \nint{N}$ and $(i,j)\in \nint{m_1}\times\nint{m_2}$, $U^{k}_{ij}\in[-1,1]$ and there exists $\ae> 0$ such that for all $(i,j)\in \nint{m_1}\times\nint{m_2}$, $\sum_{k=1}^N|U^{k}_{ij}| \leq \ae.$
\end{assumption}
Assumption \textbf{H}\ref{ass:dict} is satisfied in the three models introduced in Examples \ref{ex:groups}, \ref{ex:row-column} and \ref{ex:corruptions}: for group effects and corruptions with $\ae =1$ and for row and column effects with $\ae = 2$. In particular, it guarantees that $\TX = \fu{\Talpha} +\TL$ satisfies $\norm{\TX}[\infty]\leq (1+\ae)a$. Plugging this in the definition of $\TX$ in \eqref{exp-model}, this assumption also implies that $\PE[Y_{ij}]\in g_j'([-(1+\ae)a, (1+\ae)a])$ for all $(i,j)\in\nint{m_1}\times\nint{m_2}$. Note that \textbf{H}\ref{ass:dict} can be relaxed by $\norm{U_k}[\infty]\leq \rho$, with $\rho$ an arbitrary constant. Consider also the following assumption on the link functions.
\begin{assumption}
\label{ass:cvx}
For all $j\in\nint{m_2}$ the functions $g_j$ are twice continuously differentiable. Moreover, there exist $0<\smin,\smax < +\infty$ such that for all $|x|\leq (1+\ae)a$ and $j\in\nint{m_2}$,
$\smin^2\leq g_j''(x)\leq \smax^2.$
\end{assumption}
Assumptions \textbf{H}\ref{ass:dict}--\ref{ass:cvx} imply that the data-fitting term has Lipschitz gradient. Furthermore, the quadratic approximation defined in \eqref{eq:quad-approx} is strictly convex at every iteration. We obtain the following convergence result.
\begin{theorem}
\label{th:convergence}
Assume \textbf{H}\ref{ass:dict}--\ref{ass:cvx} and let $\{(\alpha^{[k]},L^{[k]})\}$ be the iterate sequence generated by the BCGD algorithm. Then the following results hold.
\begin{enumerate}[label=(\alph*)]
\item $\{(\alpha^{[k]},L^{[k]})\}$ has at least one accumulation point. Furthermore, all the accumulation points of $\{(\alpha^{[k]},L^{[k]})\}$ are global optima of $F$.\label{th:convergence1}
\item $\{F(\alpha^{[k]},L^{[k]})\}\rightarrow F(\hat\alpha, \hat L)$.\label{th:convergence2}
\end{enumerate}
\end{theorem}
\begin{proof}
See \Cref{proof:convergence}.
\end{proof}

\section{Statistical guarantees}
\label{main-results}

We now state our main statistical results. Denote by $\pscal{\cdot}{\cdot}$ the usual trace scalar product in $\R^{m_1\times m_2}$. For $a\geq 0$ and a sparsity pattern $\mathcal{I}\subset \nint{N}$, define the following sets
\begin{equation}
\label{eq:true-set}
\begin{aligned}
& \mathcal{E}_1(a,\mathcal{I}) = \left\lbrace \alpha\in \mathbb{R}^N, \norm{\alpha}[\infty]\leq a, \operatorname{supp}(\alpha)\subset \mathcal{I}\right \rbrace,\\
&\mathcal{E}_2(a,\mathcal{I}) = \left\lbrace L\in \mathbb{R}^{m_1\times m_2},\norm{L}[\infty]\leq a,
\underset{k\in\mathcal{I}}{\max}|\pscal{L}{U_k}|=0
\right\rbrace,\\
&\mathcal{X}(a,\mathcal{I}) = \left\lbrace X = \fu{\alpha}+L  ; (\alpha,L)\in \mathcal{E}_1(a,\mathcal{I})\times \mathcal{E}_2(a,\mathcal{I}) \right\rbrace.
\end{aligned}
\end{equation}
\begin{assumption}
\label{ass:true-set}
There exist $a>0$ and $\mathcal{I}\subset \nint{N}$ such that $(\Talpha,\TL) \in \mathcal{E}_1(a,\mathcal{I})\times\mathcal{E}_2(a,\mathcal{I}).$
\end{assumption}
Assumption \textbf{H}\ref{ass:true-set} can be relaxed to allow upper bounds to depend on the entries of $\Talpha$ and $\TL$, but we stick to  \textbf{H}\ref{ass:true-set} for simplicity.  
\subsection{Upper bounds}
\label{up-bounds}
We now derive upper bounds for the Frobenius and $\ell_2$ norms of the estimation errors $\TL - \hat{L}$ and $\Talpha-\hat{\alpha}$ respectively.
In \Cref{th:general-th} we give a general result under conditions on the regularization parameters $\lambda_1$ and $\lambda_2$, which depend on the random matrix $\nabla\mathcal{L}(\TX;Y,\Omega)$. Then, \Cref{lemma:SigmaR} and \ref{lemma:Sigma} allow us to compute values of $\lambda_1$ and $\lambda_2$ that satisfy the assumptions of \Cref{th:general-th} with high probability. Finally we combine these results in \Cref{th:upper-bound}.

We denote $\vee$ and $\wedge $ the $\max$ and $\min$ operators respectively, $M = m_1\vee m_2$, $m = m_1\wedge m_2$ and $d = m_1+m_2$. We also define $r = \rank{\TL}$, $s = \norm{\Talpha}[0]$ and $\umax = \max_k\norm{U_k}[1]$.
Let $(E_{ij})_{(i,j)\in \nint{m_1}\times\nint{m_2}}$ be the canonical basis of $\R^{m_1\times m_2}$ and $\{\epsilon_{ij}\}$ an i.i.d. Rademacher sequence independent of $Y$ and $\Omega$. Define
\begin{equation}
\label{eq:sigma}
\Sigma_R = \sum_{i=1}^{m_1}\sum_{j=1}^{m_2}\Omega_{ij}\epsilon_{ij}E_{ij}\quad \text{and}\quad\nabla \mathcal{L}(X; Y,\Omega) = \sum_{i = 1}^{m_1}\sum_{j = 1}^{m_2}\Omega_{ij}\left\lbrace -Y_{ij} + g_j'\left( X_{ij} \right)\right\rbrace E_{ij}.
\end{equation}
$\Sigma_R$ is a random matrix associated with the missingness pattern and $\nabla \mathcal{L}(X; Y,\Omega)$ is the gradient of $\mathcal{L}$ with respect to $X$. Define also
\begin{equation*}
\label{eq:upper-bounds-rate-alpha}
\begin{aligned}
&\Theta_1 && = \frac{\lambda_2}{\smin^2}+ a^2\umax\PE[{\norm{\Sigma_R}[\infty]}]+\frac{p}{\norm{\Talpha}[1]}\left(\frac{a}{p}\right)^2\log(d),\\
&\Theta_2 && = \lambda_1^2+(1+\ae)a\PE[{\norm{\Sigma_R}^2}],\\
& \Theta_3 && = \frac{\lambda_2}{\lambda_1} + 2\Theta_1.
\end{aligned}
\end{equation*}
\begin{theorem}
\label{th:general-th}
Assume \textbf{H}\ref{ass:dict}-\ref{ass:true-set} and let
 $$\lambda_1\geq 2\norm{\nabla\mathcal{L}(\TX;Y,\Omega)} \quad\text{and}\quad \lambda_2\geq 2\umax\left(\norm{\diff{\TX}}[\infty]+ 2\smax^2(1+\ae)a\right).$$
Then, with probability at least $1-8d^{-1}$,
\begin{equation}
\norm{\fu{\Talpha}-\fu{\hat\alpha}}[F]^2\leq \frac{as}{p}C_1\Theta_1\quad \text{and}\quad
 \norm{\TL - \hat{L}}[F]^2 \leq  \frac{r}{p^2}C_2\Theta_2 + \frac{as}{p}C_3\Theta_3,
\end{equation}
where $C_1$, $C_2$ and $C_3$ are numerical constants independent of $m_1$, $m_2$ and $p$.
\end{theorem}
\begin{proof}
See \Cref{proof:general-th}.
\end{proof}
We now give deterministic upper bounds on $\PE[\norm{\Sigma_R}]$ and $\PE[{\norm{\Sigma_R}[\infty]}]$ in \Cref{lemma:SigmaR}, and probabilistic upper bounds on $\norm{\nabla\mathcal{L}(\TX;Y,\Omega)}$ and $\norm{\nabla\mathcal{L}(\TX;Y,\Omega)}[\infty]$ in \Cref{lemma:Sigma}. We will use them to select values of $\lambda_1$ and $\lambda_2$ which satisfy the assumptions of \Cref{th:general-th} and compute the corresponding upper bounds.
\begin{lemma}
\label{lemma:SigmaR}
There exists an absolute constant $C^*$ such that the two following inequalities hold
\begin{equation*}
\PE[{\norm{\Sigma_R}[\infty]}] \leq 1 \quad \text{ and}  \quad \PE[\norm{\Sigma_R}] \leq C^*\left\{\sqrt{\beta} + \sqrt{\log m}\right\}.
\end{equation*}
\end{lemma}
\begin{proof}
See \Cref{proof:lemma:SigmaR}
\end{proof}
\begin{lemma}
\label{lemma:Sigma}
Assume \textbf{H}\ref{ass:dict}-\ref{ass:true-set}. Then, there exists an absolute constant $c^*$ such that the following two inequalities hold with probability at least $1-d^{-1}$:
\begin{equation}
\label{eq:Sigma-infty}
\begin{aligned}
&\norm{\nabla\mathcal{L}(\TX;Y,\Omega)}[\infty]&&\leq 6\max \left\{\smax\sqrt{\log d}, \frac{\log d}{\gamma}\right\},\\
&\norm{\nabla\mathcal{L}(\TX;Y,\Omega)}&&\leq c^*\max\left\lbrace \smax\sqrt{\beta\log d}, \frac{\log d}{\gamma }\log\left(\frac{1}{\smin}\sqrt{\frac{m_1m_2}{\beta}}\right) \right\rbrace,
\end{aligned}
\end{equation}
where $d = m_1+m_2$ , $\smax$ and $\gamma$ are defined in \textbf{H}~\ref{ass:cvx}, and $\beta$ in
\eqref{eq:beta}.
\end{lemma}
\begin{proof}
See \Cref{proof:lemma:Sigma}.
\end{proof}
We now combine \Cref{th:general-th}, \Cref{lemma:SigmaR} and \ref{lemma:Sigma} with a union bound argument to derive upper bounds on 
$\norm{\fu{\Talpha}-\fu{\hat\alpha}}[F]^2$ and $\norm{\TL-\hat L}[F]^2$. We assume that $M = (m_1\vee m_2)$ is large enough, that is
$$M\geq  \max\left\lbrace\frac{4\smax^2}{\gamma^6}\log^2\left(\frac{\sqrt{m}}{p\gamma\smin}\right), 2\exp\left(\smax^2/\gamma^2\vee\smax^2\gamma(1+\ae a)\right)\right\rbrace.$$
Define
\begin{equation*}
\begin{aligned}
&\Phi_{1} &&= a^2 + \frac{\log(d)}{u\smin^2\gamma} + \frac{a^2\log(d)}{pu\norm{\Talpha}[1]},\\
&\Phi_2 &&= \frac{\smax^2}{\smin^4}\log(d) + (1+\ae)a \left(1 \vee (\log m /\beta)\right),
\\
& \Phi_3&&=  \frac{12p\sqrt{\log(d)}}{\gamma(1+\ae)a\smax \sqrt{\beta}}+\frac{1}{\smin^2}\left(\frac{\log d}{\gamma}\right)+\frac{p}{u}
\frac{\log(d)}{u\smin^2\gamma} + \frac{a^2\log(d)}{pu\norm{\Talpha}[1]},
\end{aligned}
\end{equation*}
and recall that $s = \norm{\Talpha}[0]$, $r = \operatorname{rank}(\TL)$, $\beta\geq \max_{i,j}\left(\sum_{l=1}^{m_2}\pi_{il},\sum_{k=1}^{m_1}\pi_{kj}\right)$ and that the entries $Y_{ij}$ are sub-exponential with scale parameter $\gamma$.
\begin{theorem}
\label{th:upper-bound}
Assume \textbf{H}\ref{ass:dict}-\ref{ass:true-set} and let
$$\lambda_1= 2c^*\smax\sqrt{\beta\log d},\quad \lambda_2\geq \frac{24\umax\log(d)}{\gamma},$$
where $c_*$ is the absolute constant defined in \Cref{lemma:Sigma}. Then, with probability at least $1-10d^{-1}$,
\begin{equation}
\norm{\fu{\Talpha}-\fu{\hat\alpha}}[F]^2\leq C\frac{sa\umax}{p}\Phi_{1}\text{, and }
\norm{\TL - \hat{L}}[F]^2\leq C\left(\frac{r\beta}{p^2}
\Phi_2 + \frac{sa\umax}{p}\Phi_3\right),
\end{equation}
with $C$ an absolute constant.
\end{theorem}
Denoting by $\lesssim$ the inequality up to constant and logarithmic factors we get:
\begin{equation*}
\norm{\fu{\Talpha}-\fu{\hat\alpha}}[F]^2  \lesssim \frac{s\umax}{p}, \text{ and }\norm{\TL - \hat{L}}[F]^2  \lesssim  \frac{r\beta}{p^2}+ \frac{s\umax}{p},
\end{equation*}
In the case of almost uniform sampling, \textit{i.e.}, for all $(i,j)\in\nint{m_1}\times\nint{m_2}$ and two positive constants $c_1$ and $c_2$, $c_1 p \leq \pi_{ij}\leq c_2 p$ we obtain that $\beta \leq c_2 Mp$ and the following simplified bound:
\begin{equation}
\label{simplified-bound}
\norm{\TL - \hat{L}}[F]^2\lesssim \frac{rM}{p} + \frac{s\umax}{p}.
\end{equation}
The rate given in \eqref{simplified-bound} is the sum of the usual convergence rate of low-rank matrix completion $rM/p$ and of the usual sparse vector convergence rate $s$ \citep{Buhlmann2011, Tsybakov:2008:INE:1522486} multiplied by $u/p$. This additional factor accounts for missing observations ($p^{-1}$) and interplay between main effects and interactions ($u$). Furthermore, the estimation risk of $\fu{\Talpha}$ is also the usual sparse vector convergence rate, with an additional $up^{-1}$ factor accounting for interactions and missing values. 

Note that whenever the dictionary $\mathcal{U}$ is linearly independent, \Cref{th:upper-bound} also provides an upper bound on the estimation error of $\Talpha$. Let $G\in\mathbb{R}^{N\times N}$ be the Gram matrix of the dictionary $\mathcal{U}$ defined by $G_{kl} = \pscal{U_k}{U_l}$ for all $(k,l)\in\nint{N}\times\nint{N}$.
\begin{assumption}
\label{ass:gram} For $\kappa>0$ and all $\alpha \in \mathbb{R}^N$, $\alpha^{\top}G\alpha \geq \kappa^2\norm{\alpha}[2]^2.$
\end{assumption}
Recall that in the group effects model, we denote by $I_h$ the set of rows which belong to group $h$. \textbf{H}\ref{ass:gram} is satisfied for the group effects model with $\kappa^2 = \min_h |I_h|$, the row and column effects model with $\kappa^2 = \min(m_1, m_2)$ and the corruptions model with $\kappa^2 =  1$. If \textbf{H}\ref{ass:gram} is satisfied then, \Cref{th:upper-bound} implies that (up to constant and logarithmic factors):
\begin{equation*}
\norm{\Talpha-\hat\alpha}[2]^2  \lesssim \frac{s\umax}{p\kappa^2}.
\end{equation*}
\subsection{Lower bounds}
\label{low-bound}
To characterize the tightness of the convergence rates given in \Cref{th:upper-bound}, we now provide lower bounds on the estimation errors. We need three additional assumptions.
\begin{assumption}
\label{ass:sampling-unif}
The sampling of entries is uniform, i.e. for all $(i,j)\in\nint{m_1}\times\nint{m_2}$, $\pi_{ij} = p$.
\end{assumption}
\begin{assumption}
\label{ass:distribution}
There exists $\mathcal{I}\subset \nint{N}$, $a>0$ and $X\in\mathcal{X}_{\mathcal{I}, a}$ such that for all $(i,j)\in\nint{m_1}\times\nint{m_2}$, $Y_{ij}\sim\operatorname{Exp}^{(h_j, g_j)}(X_{ij})$.
\end{assumption}
Denote  $\tau = \max_k \sum_{l\neq k}|\pscal{U_k}{U_l}|$. Without loss of generality we assume $m_1 = m_1\vee m_2 = M$. For all $X\in \mathbb{R}^{m_1\times m_2}$ we denote $\mathbb{P}_X$ the product distribution of $(Y,\Omega)$ satisfying \textbf{H}\ref{ass:sampling-unif} and \ref{ass:distribution}. Consider two integers $s\leq (m_1\wedge m_2)/2$ and $r\leq (m_1\wedge m_2)/2$.
 We define the following set
\begin{equation}
\label{eq:set}
\mathcal{F}(r,s) = \bigcup_{|\mathcal{I}|\leq s} \left\lbrace (\alpha,L)\in \mathcal{E}_1(a,\mathcal{I})\times\mathcal{E}_2(a,\mathcal{I}); \rank{L}\leq r\right\rbrace.
\end{equation}
\begin{theorem}
\label{th:lower-bound}
Assume \textbf{H}\ref{ass:dict}-\ref{ass:sampling-unif} and $p\geq \frac{r}{m_1\wedge m_2}$. Then, there exists a constant $\delta >0$ such that
\begin{equation}
\label{eq:inf}
\inf_{\hat{L},\hat{\alpha}}\sup_{(\TL,\Talpha)\in \mathcal{F}(r,s)} \mathbb{P}_{\TX}\left(\norm{\TL - \hat L}[F]^2+\norm{\fu{\Talpha}-\fu{\hat\alpha}}[F]^2>\Psi_1\frac{rM}{p}+ \Psi_2\frac{s\kappa^2}{p}\right) \geq \delta,
\end{equation}
\begin{equation}
\begin{aligned}
&\Psi_1 &&= C\min\left(\smax^{-2},\min(a,\smax)^2\right),\\
&\Psi_2 &&=  C\left(\frac{1}{\smax^2\left(\max_k\norm{U^k}[F]^{2}+2\tau\right)}\wedge (a\wedge\smax)^2\right).
\end{aligned}
\end{equation}
\end{theorem}
\begin{proof}
See \Cref{th:lower-bound-proof}.
\end{proof}
\subsection{Examples}
\label{examples}
We now specialize our theoretical results to Examples \ref{ex:groups}, \ref{ex:row-column} and \ref{ex:corruptions} presented in \Cref{subsec:gen-model}. We compute the values of $\umax$, $\tau$ and $\max_k\norm{U^k}[F]^2$ for the group effects, row and column effects and corruption models, and obtain the rates of \Cref{th:upper-bound} and \Cref{th:lower-bound} for these particular cases. Recall that in the group effects model, we denote by $I_h$ the set of rows which belong to group $h$. The orders of magnitude are summarized in \Cref{fig:upper-bounds} for the upper bound and in \Cref{fig:lower-bounds} for the lower bound. 
\begin{table}
\footnotesize
\begin{center}
\begin{tabular}{|c|c|c|c|}
\hline
\textbf{Model} & \textbf{Group effects} & \textbf{Row \& col effects} & \textbf{Corruptions}\\
\hline
$\umax$ & $\max_h|I_h|$ & $M$ & $1$\\
\hline
$ \norm{\Delta L}[F]^2+\norm{\fu{\Talpha}-\fu{\hat\alpha}}[F]^2$ & $rM/p + s\max_h|I_h|/p$ & $rM/p + sM/p $ & $rM/p+s/p$\\
\hline
\end{tabular}
\end{center}
\caption{Order of magnitude of the upper bound for Examples \ref{ex:groups}, \ref{ex:row-column} and \ref{ex:corruptions} (up to logarithmic factors).}
\label{fig:upper-bounds}
\end{table}
\begin{table}
\footnotesize
\begin{center}
\begin{tabular}{|c|c|c|c|}
\hline
\textbf{Model} & \textbf{Group effects} & \textbf{Row \& col effects} & \textbf{Corruptions}\\
\hline
$\umax$ & $\max_h|I_h|$ & $M$ & $1$\\
\hline
$\max_{k}\norm{U_k}[F]^2$ & $\max_h|I_h|$ & $M$ & $1$\\
\hline
$\kappa^2$ & $\min_h|I_h|$ & $m$ & $1$\\
\hline
$\norm{\Delta L}[F]^2$ +$\norm{\fu{\Talpha}-\fu{\hat\alpha}}[F]^2$ & $rM/p+(s\min_h|I_h|)/(p\max_h|I_h|)$ & $rM/p+sm/(pM)$ & $rM/p+s/p$\\
\hline
\end{tabular}
\end{center}
\caption{Order of magnitude of the lower bound for Examples \ref{ex:groups}, \ref{ex:row-column} and \ref{ex:corruptions}.}
\label{fig:lower-bounds}
\end{table}
Comparing \Cref{fig:upper-bounds} and \Cref{fig:lower-bounds} we see that the convergence rates obtained in \Cref{th:upper-bound} are minimax optimal across the three examples whenever $s<r$. Furthermore, in the corruptions model our rates are optimal (up to constant and logarithmic factors) for any values of $r,s$ and $M$, and equal to the minimax rates derived in \cite{Klopp2017}.
In the case of group effects, the rates are optimal when $r>s\max_h|I_h|/M$ or when $\max_h|I_h|$ is of the order of a constant. When $s>rM/ \max_h|I_h|$, we have an additional factor of the order $(\max_h|I_h|)^2/\min_h|I_h|$ in the upper bound. Note that the bounds have the same dependence in the sparsity pattern $s$.
In the row and column model, when $r<s$, we have an additional factor of the order $s/r$ in the upper bound.

\section{Numerical results}
\label{sec:numerical-results}
\subsection{Estimation of main effects and interactions}
\label{subsec:simulated-data}
We start by evaluating our method (referred to as ``mimi": \textbf{m}ain effects and \textbf{i}nteractions in \textbf{m}ixed and \textbf{i}ncomplete data) in terms of estimation of main effects and interactions. In this experiment, we focus on the group effects model presented in \Cref{sec:acs-intro}, with $H = 5$ groups of equal size.
We select at random $s$ non-zero coefficients in $\Talpha$, and construct a matrix $\TL$ of rank $k$. Then, $\TX = \sum_{h=1}^H\sum_{j=1}^{m_2} \Talpha_{hj}U^{h,j} + \TL,$ with $U_{h,j}$, $1\leq h\leq H$ and $1\leq j\leq m_2$ defined in \Cref{ex:groups}. Finally, every entry of the matrix is observed with probability $p$. 

In this first experiment, we consider only numeric variables to compare mimi to the following two-step method. In this alternative method, the main effects $\Talpha$ are estimated by the means of the variables taken by group; this corresponds to the preprocessing step performed in \cite{glrm} and \cite{Landgraf15} for instance. Then, $\TL$ is estimated using \texttt{softImpute} \citep{softImpute}; we refer to this method as ``group mean + softImpute". The regularization parameters of both methods are selected with cross-validation.

The results are displayed in Figure \ref{fig:alpha} where we plot the estimation errors $\norm{\hat \alpha - \Talpha}[2]^2$ and in Figure \ref{fig:theta} $\norm{\hat L - \TL}[F]^2$ 
for different levels of sparsity and different ranks.

\begin{figure}
\centering
\includegraphics[scale=0.35]{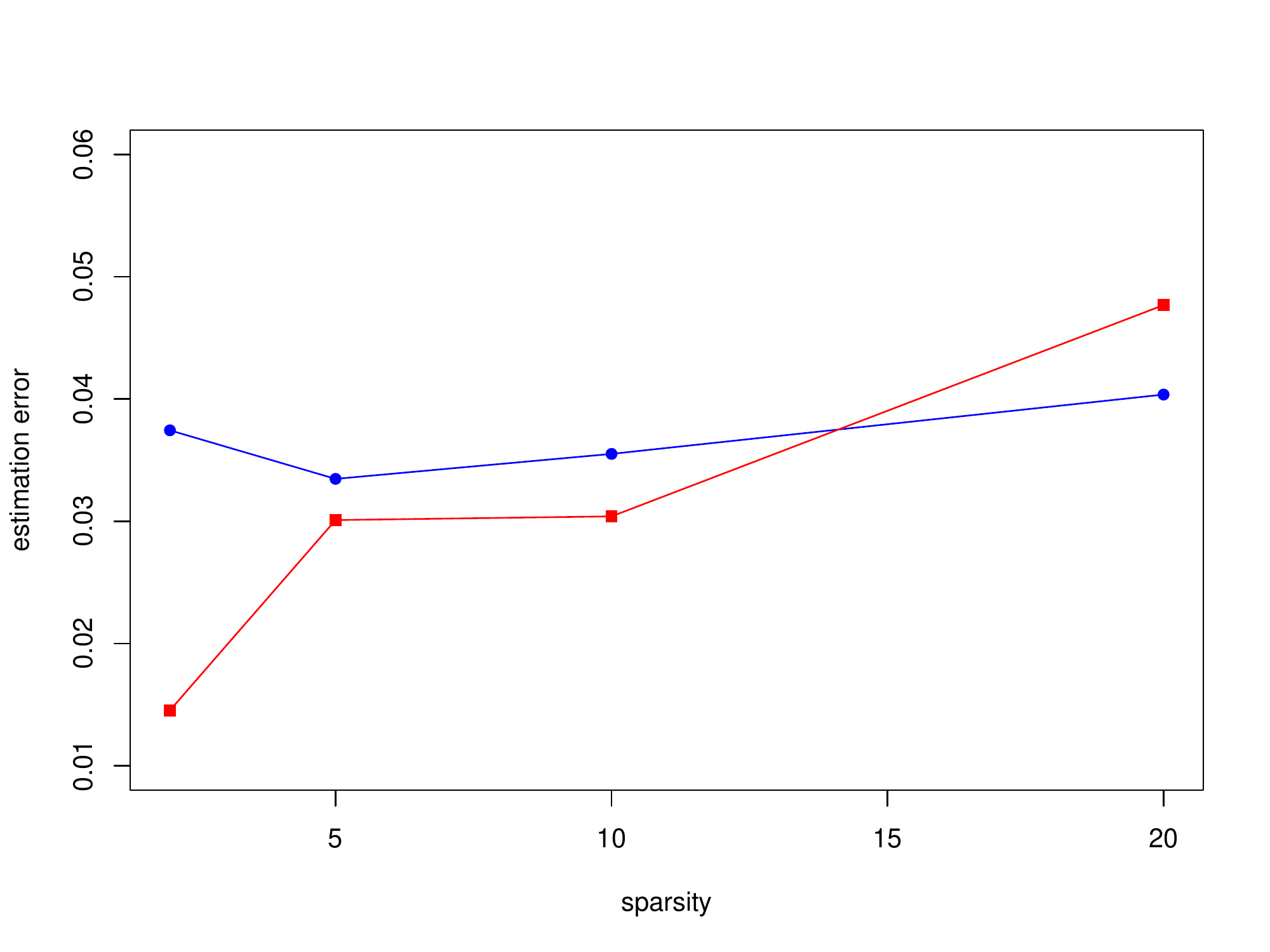}
\includegraphics[scale=0.35]{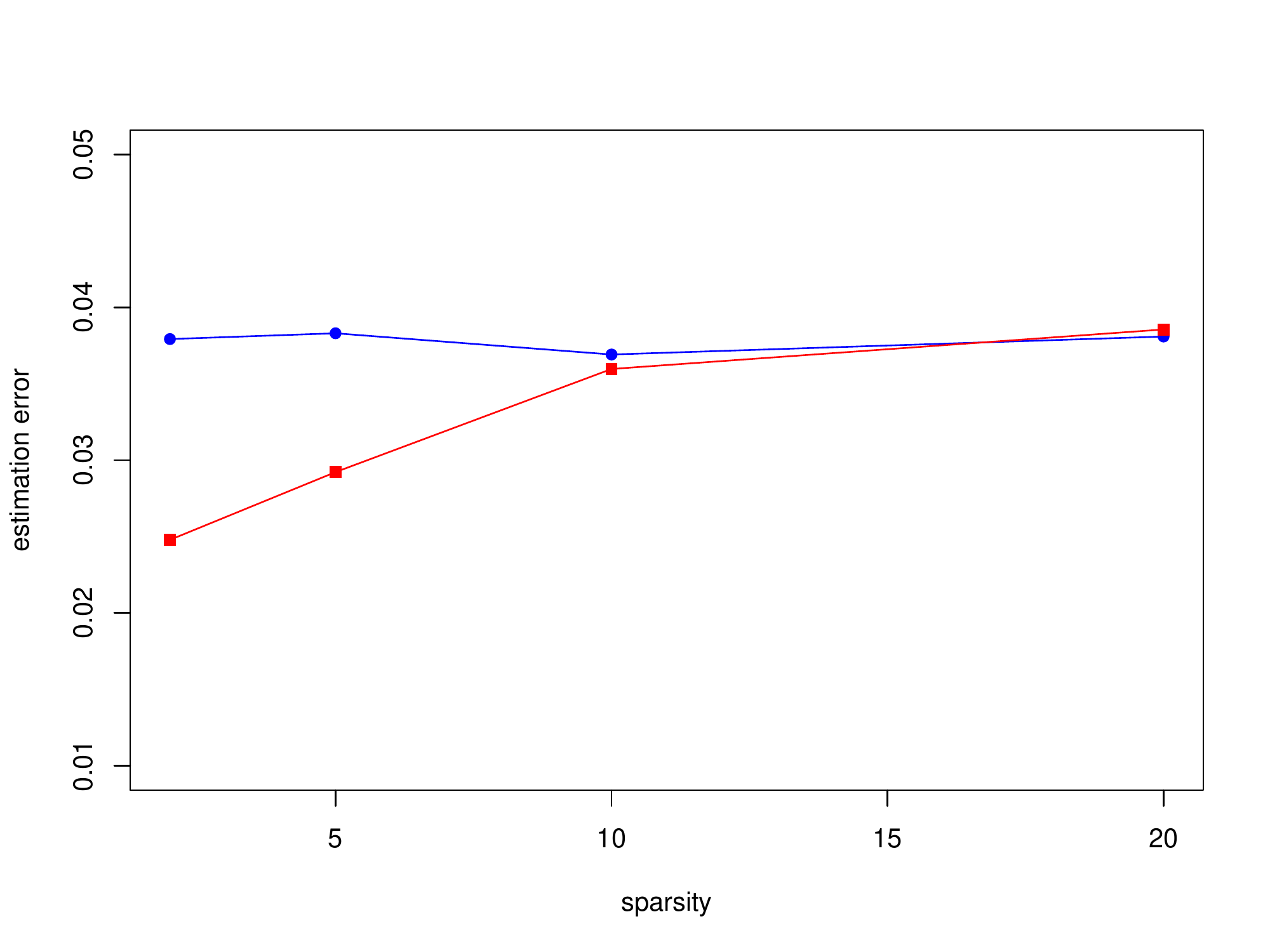}
\includegraphics[scale=0.35]{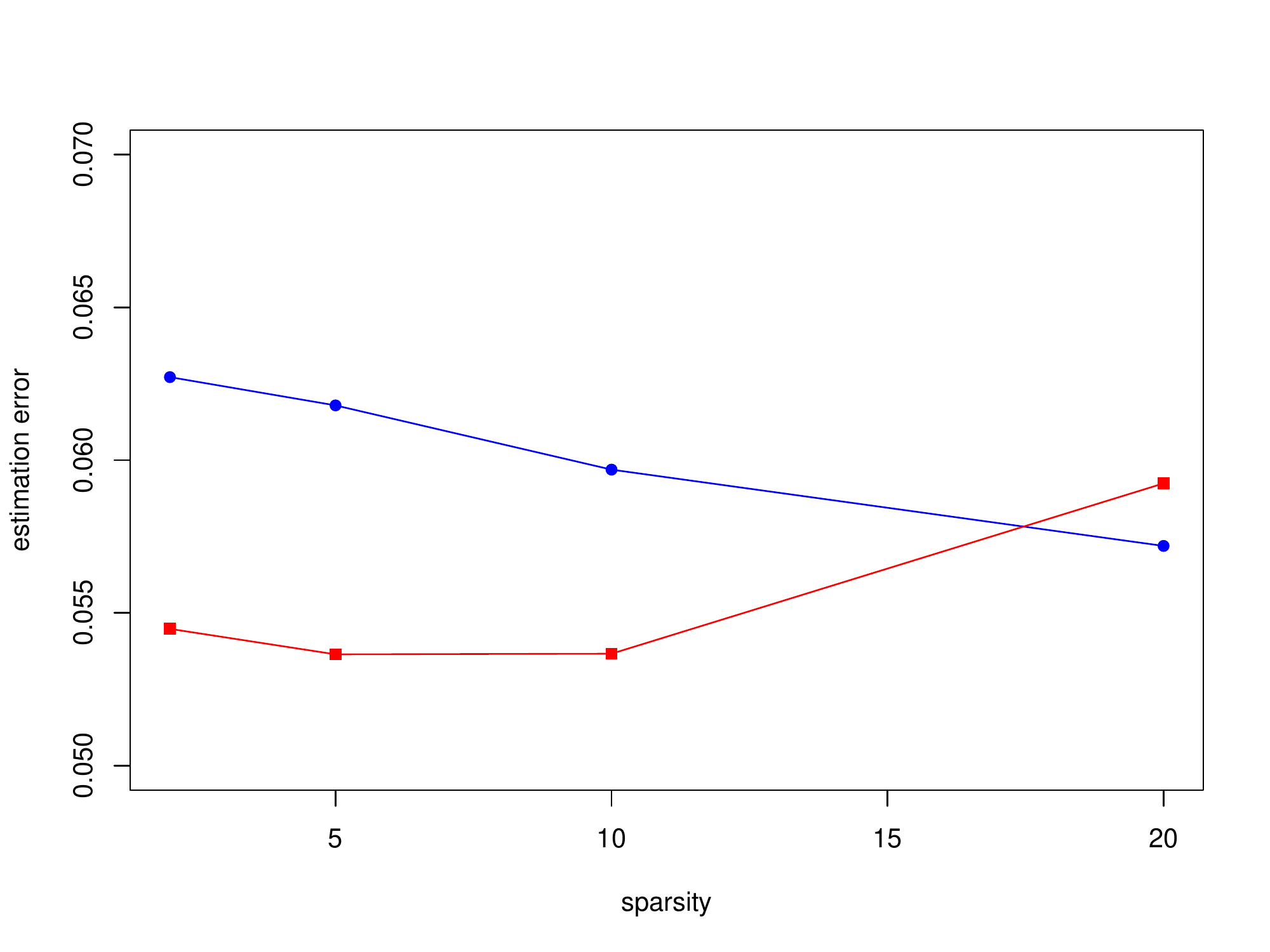}
\includegraphics[scale=0.35]{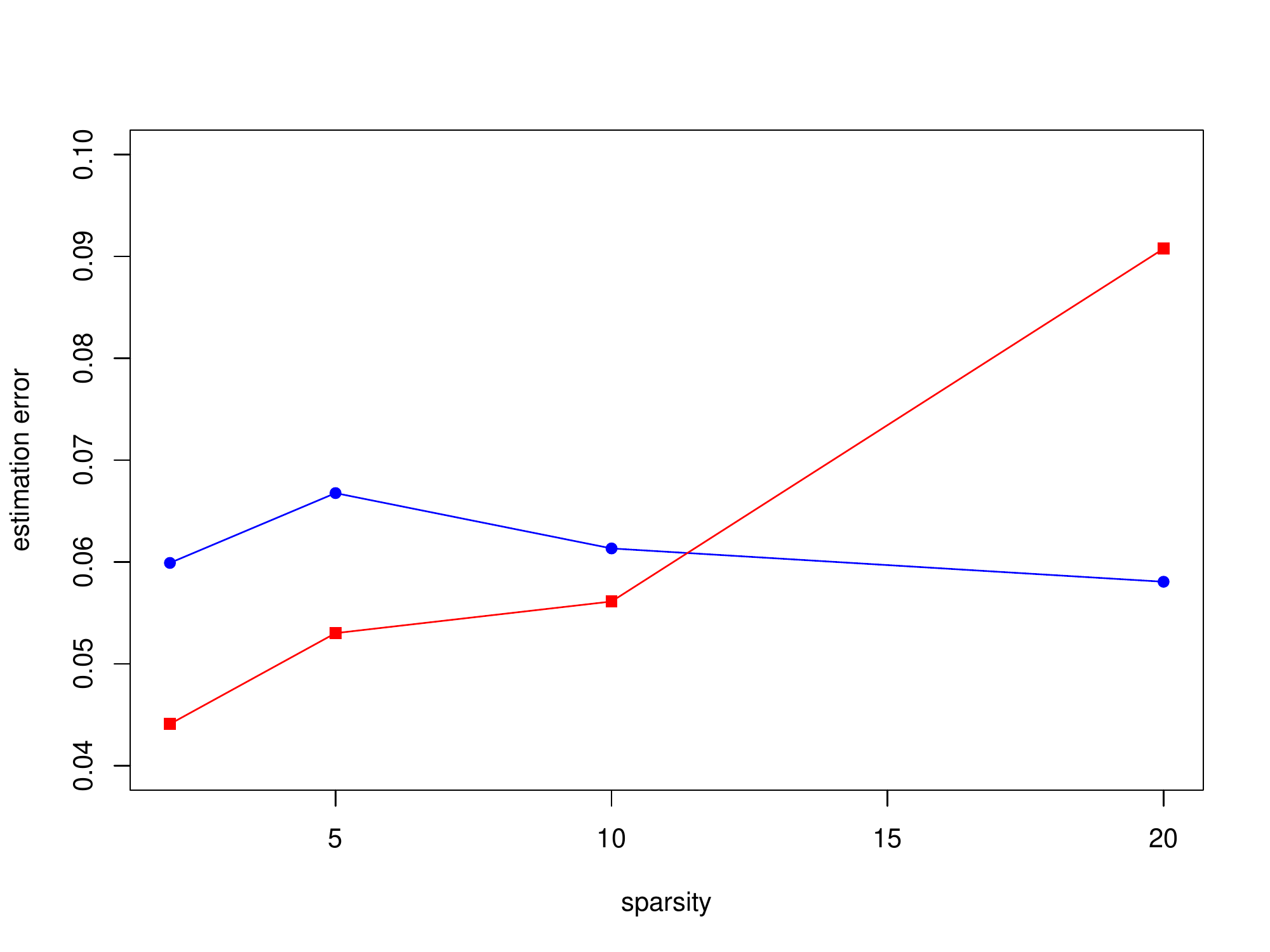}
\caption{Estimation error $\norm{\hat\alpha-\Talpha}[2]^2$ of mimi (red squares) and of groups means + softImpute (blue points) for increasing problem sparsity levels and ranks. The sparsity $s=2,5,10,20$ is indicated in the abscissa and the rank $k=2,5,10,20$ corresponds to different plots: top left $k=2$, top right $k=5$, bottom left $k=10$, bottom right $k=20$. The dimensions are fixed to $m_1=300$ and $m_2=30$ and the proportion of missing entries to $p=0.2$.}
\label{fig:alpha}
\end{figure}

On Figure \ref{fig:alpha} we observe that for a fixed rank, mimi has a smaller error ($\norm{\hat\alpha-\Talpha}[2]^2$) than the two-step procedure for small sparsity levels, and that the difference between the two methods cancels as the sparsity level increases. Furthermore, as the rank also increases (from top to bottom and from left to right), the difference between mimi and the two-step procedure also decreases. Finally, for large ranks and sparsity levels simultaneously, mimi has a large estimation error $\norm{\hat\alpha-\Talpha}[2]^2$ compared to the two-step procedure which does not assume sparsity. This case can be seen as a model mis-specification setting.

\begin{figure}
\centering
\includegraphics[scale=0.35]{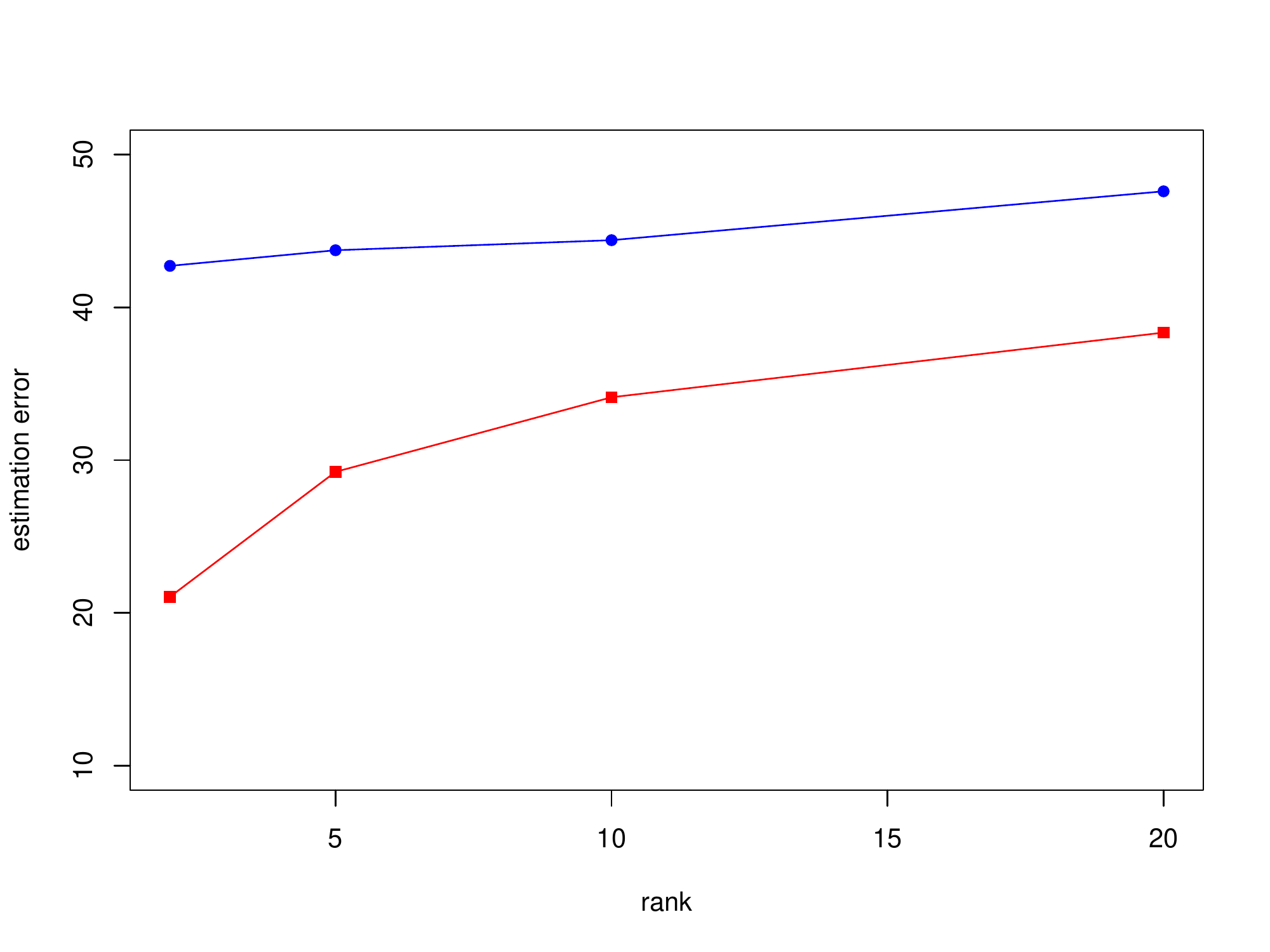}
\includegraphics[scale=0.35]{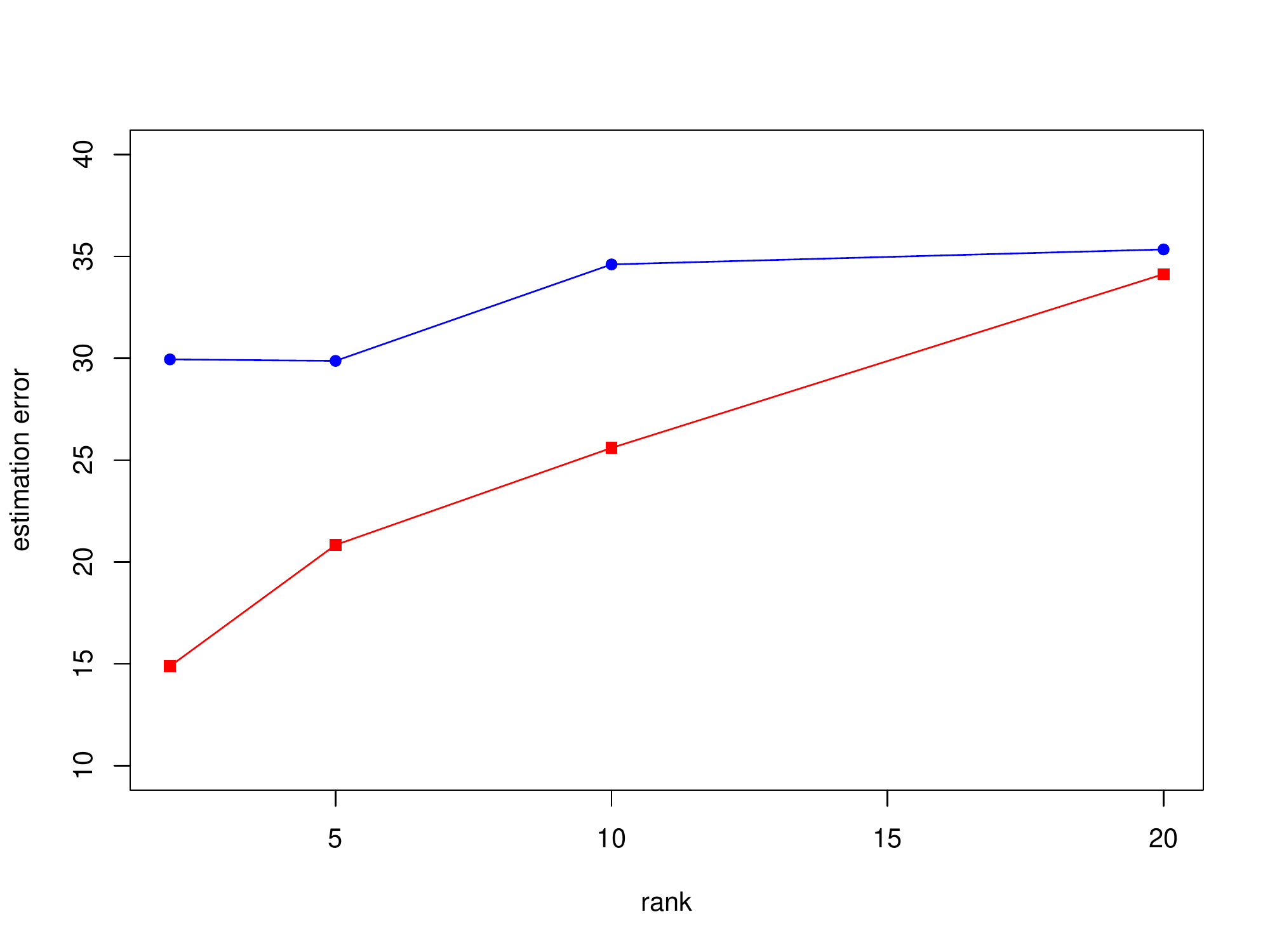}
\includegraphics[scale=0.35]{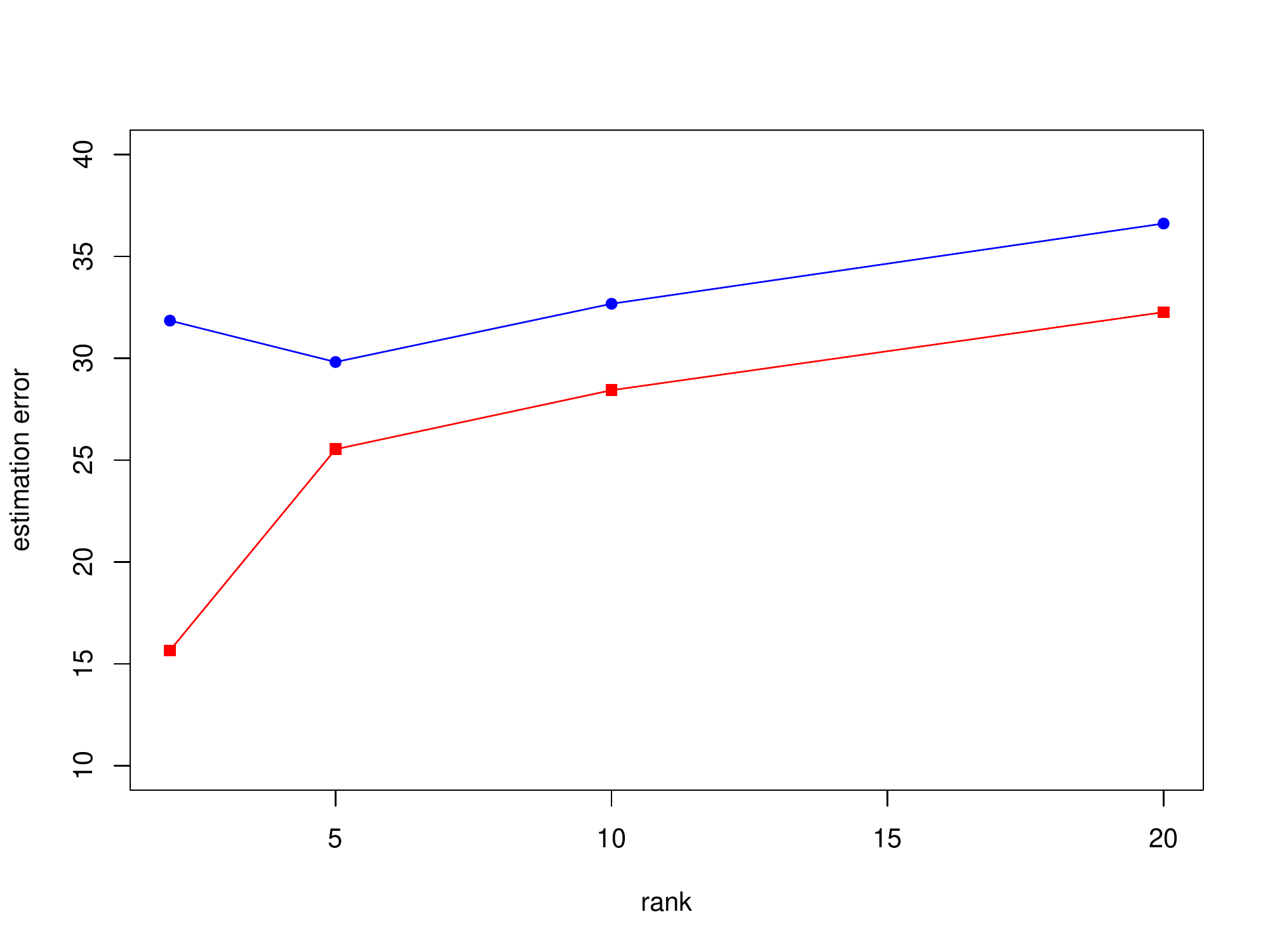}
\includegraphics[scale=0.35]{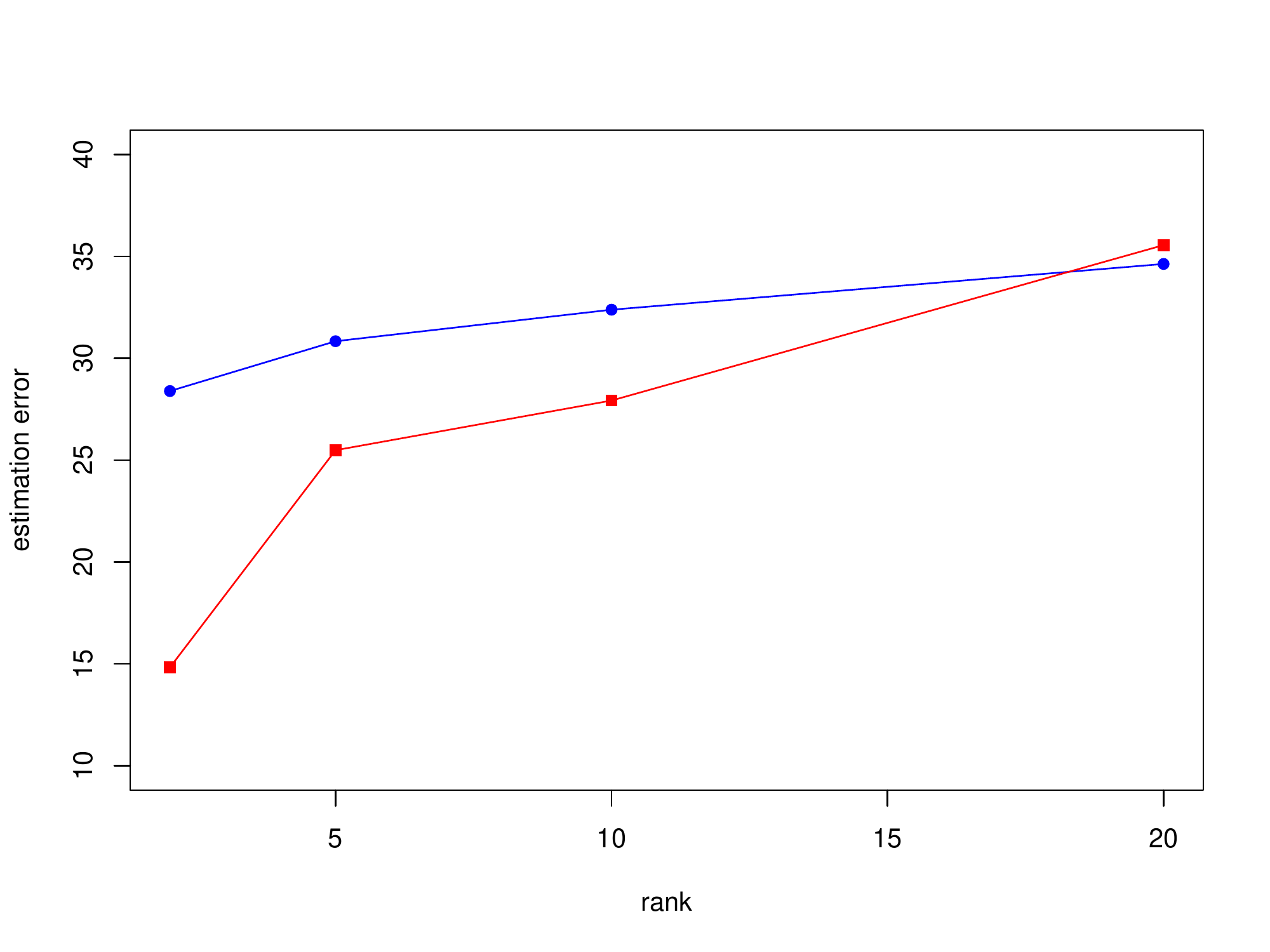}
\caption{Estimation error $\norm{\hat L-\TL}[F]^2$ of mimi (red squares) and of groups means + softImpute (blue points) for increasing problem sparsity levels and ranks. The rank $k=2,5,10,20$ is indicated in the abscissa and the sparsity $s=2,5,10,20$ corresponds to different plots: top left $s=2$, top right $s=5$, bottom left $s=10$, bottom right $s=20$. The dimensions are fixed to $m_1=300$ and $m_2=30$ and the proportion of missing entries to $p=0.2$.}
\label{fig:theta}
\end{figure}

On Figure \ref{fig:theta} we observe that mimi has overall smaller errors ($\norm{\hat L-\TL}[F]^2$) than the two-step procedure. The difference between the two methods cancels as the rank increases. We also observe that the level of sparsity has little impact on the results. However, for large ranks and sparsity levels simultaneously, mimi has a larger estimation error $\norm{\hat L-\TL}[F]^2$ than the two-step procedure.\\

Secondly, we fix the level of sparsity to $s=5$ and the rank to $k=5$, and perform the same experiment for increasing problem sizes ($150\times 30$, $1500\times 300$ and $1500\times 3000$). The results are given in Figure \ref{fig:imput}. We observe that the excess risk   $\norm{\hat L - \TL}[F]^2$  the two methods are similar. In terms of estimation of $\Talpha$, the estimation error of mimi is constant as the problem size increases but the sparsity level of $\Talpha$ is kept constant, as predicted by \Cref{th:upper-bound}. On the contrary, we observe that estimating $\Talpha$ in a preprocessing step yields large errors in high dimensions.
\begin{figure}
\centering
\includegraphics[scale=0.3]{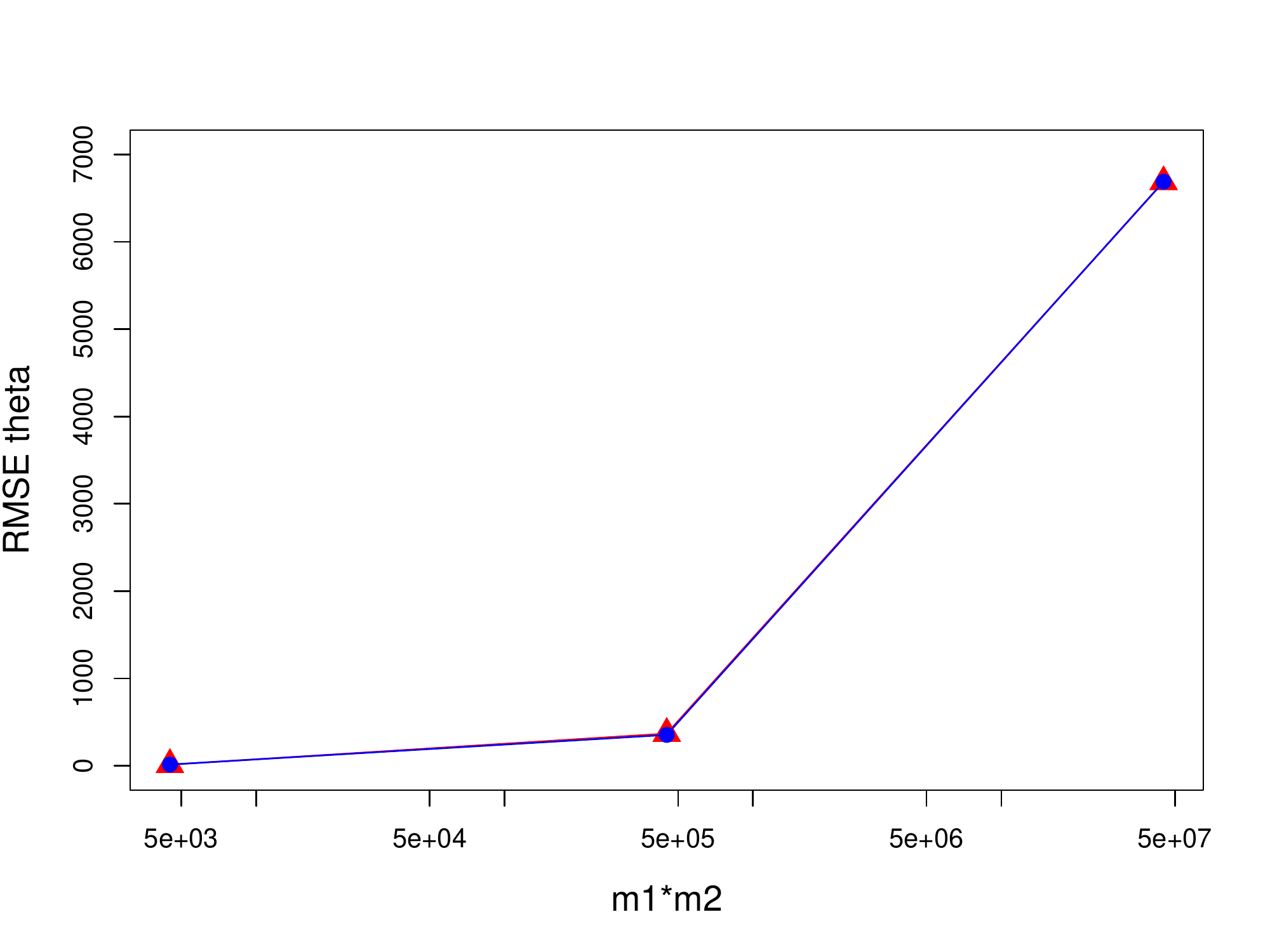}
\includegraphics[scale=0.3]{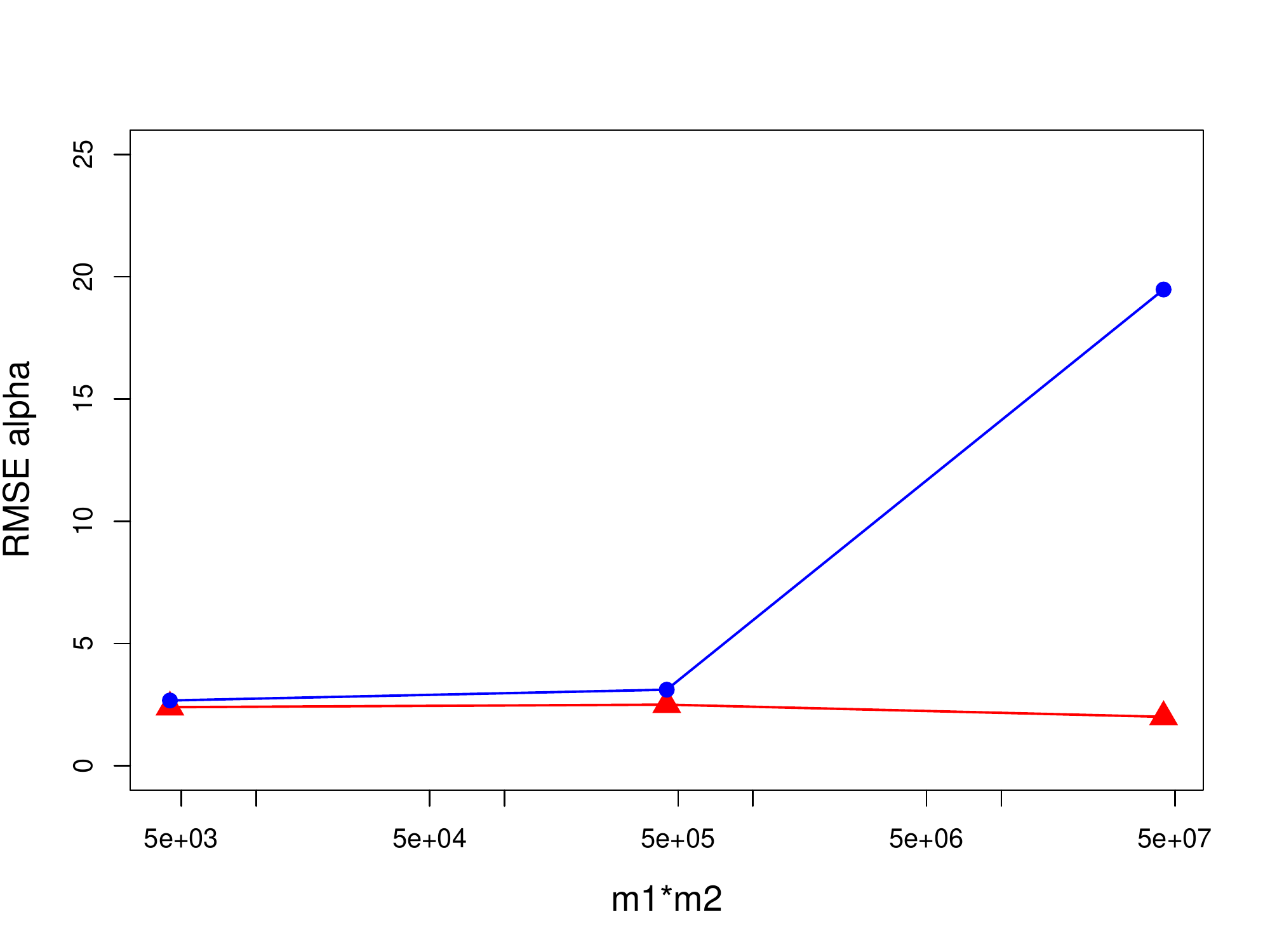}
\caption{Estimation error of mimi (red triangles) and of groups means + softImpute (blue points) for increasing problem sizes ($m_1m_2$, in log scale).}
\label{fig:imput}
\end{figure}

\subsection{Imputation of mixed data}
\label{subsec:mixed-data}

To evaluate mimi in a mixed data setting, we compare it in terms of imputation of missing values to five state-of-the-art methods:
\begin{itemize}
\item softImpute \citep{softImpute}, a method based on soft-thresholding of singular values to impute numeric data implemented in the R package \href{ https://CRAN.R-project.org/package=softImpute }{\texttt{softImpute}}.
\item Generalized Low-Rank Model (GLRM, \cite{glrm}), a matrix factorization framework for mixed data implemented in R in the \href{https://CRAN.R-project.org/package=h2o}{\texttt{h2o}} package.
\item Factorial Analysis of Mixed Data (FAMD, \cite{Pages2015}), a principal component method for mixed data implemented in the R package \href{https://CRAN.R-project.org/package=missMDA}{\texttt{missMDA}} \citep{missMDA}.
\item Multilevel Factorial Analysis of Mixed Data (MLFAMD, \cite{MLFAMD}), an extension of FAMD to impute \textit{multilevel} data, i.e. when individual are nested within groups. The method is also implemented in \href{https://CRAN.R-project.org/package=missMDA}{\texttt{missMDA}}.
\item Multivariate Imputation by Chained Equations (mice, \cite{mice}), an implementation of multiple imputation using Fully Conditional Specification. In the package \href{https://CRAN.R-project.org/package=mice}{\texttt{mice}}, different models can be set for each column to account for mixed data.
\end{itemize}
Note that we also add a comparison to imputation by the column means, in order to have a baseline reference. We fix a dictionary $\mathcal{U}$ of indicator matrices corresponding to group effects (see \Cref{ex:groups}), and generate a parameter matrix satisfying the decomposition \eqref{eq:X_decomp}. Then, columns are sampled from different data types, namely Gaussian and Bernoulli. For varying proportions of missing entries and values of the ratio $\rho=\norm{\fu{\Talpha}}[F]/\norm{\TL}[F]$, we evaluate the six methods in terms of imputation error of the two different data types. 
The parameters of all the methods (number of components for GLRM and FAMD and regularization parameters for softImpute and mimi) are selected using cross-validation. In addition, we use an optional ridge regularization in the h2o implementation of the GLRM method, which penalizes the $\ell_2$ norm of the left and right principal components ($U$ and $V$), and improved the imputation in practice. The details are available in the associated code provided as supplementary material. 

\begin{table}[!ht]
\centering
\begin{scriptsize}    
    \begin{tabular}{|l|l|l|l|l|l|l|l|l|l|}
    \hline
    {\bf \% missing}& \multicolumn{3}{c|}{\bf 20} &\multicolumn{3}{c|}{\bf 40}  &\multicolumn{3} {c|}{\bf 60}\\
    \hline
    $\mathbf{\rho}$ & {\bf 0.2} &{\bf 1}  & {\bf 5} &{\bf 0.2} &{\bf 1}   &{\bf 5}&{\bf 0.2} &{\bf 1}  & {\bf 5}\\
    \hline
    {\bf mean}   & 24.5(0.7)&23.3(0.7)&22.9(0.4) &      24.4(1.15)&33.2(1.1)&31.0(1.0)     & 42.1(1.2) & 40.7(1.2) & 39.9(0.6)   \\
    \hline
    {\bf mimi}& 18.6(0.4)& \textbf{18.3}(0.3)&\textbf{17.7}(0.3)    & 18.8(0.3) &27.0(0.5)& \textbf{24.8}(0.6)  &    36.0(1.0) & \textbf{33.7}(0.8)& \textbf{30.6}(0.4)  \\
    \hline
    {\bf GLRM} &21.5(0.7) &22.0(0.8) & 19.9(0.5) &    21.5(0.7) &31.7(1.2) & 31.0(0.9) &   44.5(10.8)& 49.4(16.2) & 50.7(3.2)\\
    \hline
    {\bf softImpute}  & \textbf{18.5}(0.3)& 18.5(0.2)& 17.9(0.3)&     18.6(0.3)& \textbf{26.8}(0.6)& 24.9(0.5) &       34.9(1.0)& 34.9(0.8)& 32.2(0.5)  \\
    \hline
    {\bf FAMD}   & \textbf{18.5}(0.4) &18.9(0.4)&18.1(0.4) &      18.7(0.3)&28.3(0.6)&25.6(0.7)     & 36.0(1.5) & 40.6(0.8)& 32.7(0.5)   \\
    \hline
    {\bf MLFAMD}   & \textbf{18.5}(0.4)&19.2(0.4)&18.3(0.4) &      \textbf{18.5}(0.5)&27.7(0.6)&26.3(0.5)     & \textbf{34.9}(1.3)& 40.7(1.0)& 33.5(0.6)   \\
    \hline
    {\bf mice}   & 22.3(0.8)&22.6(0.6)&22.1(0.6) &      22.7(0.6)&32.9(0.6)&30.1(0.9)    & 48.1(2.4) & 48.1(0.9) & 44.7(1.4)   \\
    \hline
    \end{tabular}
\end{scriptsize}
\caption{Imputation error (MSE) of mimi, GLRM, softImpute and FAMD for different percentages of missing entries ($20$\%, $40$\%, $60$\%) and different values of the ratio $\norm{\fu{\Talpha}}[F]/\norm{\TL}[F]$ ($0.2$, $1$, $5$). The values are averaged across $100$ replications and the standard deviation is given between parenthesis. In this simulation $m_1=150$, $m_2=30$, $s=3$ and $r=2$.} 
    \label{tbl:imput_all}
\end{table}
The results, presented in \Cref{tbl:imput_all}, reveal that mimi, softImpute, FAMD and MLFAMD yield imputation errors of comparable order. In this simulation setting, our method mimi improves on these existing methods when the ratio $\rho=\norm{\fu{\Talpha}}[F]/\norm{\TL}[F]$ is large, i.e. when the scale of the main effects is large compared to the interactions. The size of this improvement also increases with the amount of missing values. The imputation error by data type (quantitative and qualitative) are given in \Cref{sec:imput_by_type}, along with average experimental computational times of all the compared methods.
\subsection{American Community Survey}
\label{subsec:american-community-survey}
We next apply our method on the American Community Survey data presented in \Cref{sec:acs-intro}. We use the 2016 survey \footnote{available at \url{https://factfinder.census.gov/faces/nav/jsf/pages/searchresults.xhtml?refresh=t}} and restrict ourselves to the population of Alabama (24,614 household units). We focus on twenty variables (11 quantitative and 9 binary), and use mimi to estimate the effect of the Family Employment Status categories on these 20 variables. In other words, we place ourselves in the framework of \Cref{sec:acs-intro} and \Cref{ex:groups}. 
We model the quantitative attributes using Gaussian distributions, and the binary attributes with Bernoulli distributions. Using the same notations as in \Cref{sec:acs-intro}, $c(i)$ denotes the group (the FES) to which household $i$ belongs. Thus, if the $j$-th column is continuous (income), our model implies:
\[
 \mathbb{E}[Y_{ij}] = \Talpha_{c(i)j} + \TL_{ij}.
\]
If the $j$-th column is binary (food stamps allocation for instance), we model
$$ \mathbb{P}(Y_{ij} = 1) = \frac{\mathrm{e}^{\TX_{ij}}}{1+\e^{\TX_{ij}}},\quad \TX_{ij} = \Talpha_{c(i)j} + \TL_{ij}.$$
In \Cref{ACS}, we display the value of the parameter $\alpha_{c(i)j}$ for all possible groups $c(i)$ and some variables $j$ corresponding to the number of people in household, food stamps and allocations attributions. The value of $\alpha_{c(i)j}$ is related to the expected value $\mathbb{E}[Y_{ij}]$: everything else being fixed,  $\mathbb{E}[Y_{ij}]$ is an increasing function of $\alpha_{c(i)j}$. Thus, in terms of interpretation, the "group effect" $\alpha_{c(i)j}$ indicates (everything else being equal) whether belonging to category $c(i)$ yields larger or smaller values for $\mathbb{E}[Y_{ij}]$ compared to other categories.\\ 

We observe that household categories corresponding to married couples and single women have positive group effects on the variable "Number of people", meaning that these categories of households tend to have more children. We also observe that household categories containing employed people tend to receive less food stamps than other categories.

\begin{table}[ht]
\centering
\footnotesize
\begin{tabular}{|l|r|r|r|}
  \hline
FES & Nb of people & Food stamps (0: no, 1: yes)& Allocations (0: no, 1: yes) \\
  \hline
 Couple - both in LF & 0.38 & -1.8 & -0.68\\
 Couple - male in LF & 0.32 & -1.4 & -0.39\\
 Couple - female in LF & 0 & -0.9 & 0\\
 Couple - neither in LF & 0 & -1.6 & -0.12\\
 Male - in LF & 0 & 0&0 \\
 Male - not in LF & 0 & 0 &0\\
 Female - in LF & 0.28 & -0.19 &0\\
 Female  - not in LF & 0.13 & 0 &0\\
   \hline
\end{tabular}
\caption{Effect of family type and employment status estimated with mimi.}
\label{ACS}
\end{table}
The estimated low-rank component $\hat L$ has rank $5$, indicating that $5$ dimensions, in addition to the family type and employment status covariate, are needed to explain the observed data points.

\section{Conclusion}

This article introduces a general framework to analyze high-dimensonal, mixed and incomplete data frames with main effects and interactions. Upper bounds on the estimation error of main effects and interactions are derived. These bounds match with the lower-bounds under weak additional assumptions.
Our theoretical results are supported by a numerical experiments on synthetic and survey data, showing that the introduced method performs best when the proportion of missing values is large and the main effects and interactions are of comparable size. 

Our work opens several directions of future research. A natural extension would be to consider the inference problem, \textit{i.e.}, to derive confidence intervals for the main effects coefficients. Another useful direction would be to consider exponential family distributions with multi-dimensional parameters, for example multinomial, distributions, to incorporate categorical variables with more than two categories. One could also learn the scale parameter (which we currently assume fixed) adaptively.

\bibliographystyle{Chicago}

\bibliography{references}

\begin{thebibliography}{}

\bibitem[\protect\citeauthoryear{Agarwal, Zhang, and Mazumder}{Agarwal
  et~al.}{2011}]{agarwal2011}
Agarwal, D., L.~Zhang, and R.~Mazumder (2011, September).
\newblock Modeling item-–item similarities for personalized recommendations
  on yahoo! front page.
\newblock {\em Ann. Appl. Stat.\/}~{\em 5\/}(3), 1839--1875.

\bibitem[\protect\citeauthoryear{Agresti}{Agresti}{2013}]{Agresti13}
Agresti, A. (2013).
\newblock {\em Categorical Data Analysis, 3rd Edition}.
\newblock Wiley.

\bibitem[\protect\citeauthoryear{Aubin and Ekeland}{Aubin and
  Ekeland}{1984}]{AubinEkeland}
Aubin, J.-P. and I.~Ekeland (1984).
\newblock {\em Applied nonlinear analysis}.
\newblock Pure and applied mathematics. John Wiley, New-York.
\newblock A Wiley-Interscience publication.

\bibitem[\protect\citeauthoryear{Bandeira and van Handel}{Bandeira and van
  Handel}{2016}]{bandeira2016}
Bandeira, A.~S. and R.~van Handel (2016, July).
\newblock Sharp nonasymptotic bounds on the norm of random matrices with
  independent entries.
\newblock {\em Ann. Probab.\/}~{\em 44\/}(4), 2479--2506.

\bibitem[\protect\citeauthoryear{B{\"u}hlmann and van~de Geer}{B{\"u}hlmann and
  van~de Geer}{2011}]{Buhlmann2011}
B{\"u}hlmann, P. and S.~van~de Geer (2011).
\newblock {\em Statistics for High-Dimensional Data: Methods, Theory and
  Applications}.
\newblock Springer.

\bibitem[\protect\citeauthoryear{Cai and Zhou}{Cai and
  Zhou}{2013}]{Cai:2013:MCM:2567709.2627673}
Cai, T. and W.-X. Zhou (2013, December).
\newblock A max-norm constrained minimization approach to 1-bit matrix
  completion.
\newblock {\em J. Mach. Learn. Res.\/}~{\em 14\/}(1), 3619--3647.

\bibitem[\protect\citeauthoryear{Cand\`{e}s, Li, Ma, and Wright}{Cand\`{e}s
  et~al.}{2011}]{Candes:2011:RPC}
Cand\`{e}s, E.~J., X.~Li, Y.~Ma, and J.~Wright (2011, June).
\newblock Robust principal component analysis?
\newblock {\em J. ACM\/}~{\em 58\/}(3), 11:1--11:37.

\bibitem[\protect\citeauthoryear{Cao and Xie}{Cao and Xie}{2016}]{Cao2016}
Cao, Y. and Y.~Xie (2016, March).
\newblock Poisson matrix recovery and completion.
\newblock {\em IEEE Transactions on Signal Processing\/}~{\em 64\/}(6).

\bibitem[\protect\citeauthoryear{Chandrasekaran, Sanghavi, Parrilo, and
  Willsky}{Chandrasekaran et~al.}{2011}]{chandra}
Chandrasekaran, V., S.~Sanghavi, P.~A. Parrilo, and A.~S. Willsky (2011).
\newblock Rank-sparsity incoherence for matrix decomposition.
\newblock {\em SIAM Journal on Optimization\/}~{\em 21\/}(2), 572--596.

\bibitem[\protect\citeauthoryear{Chatterjee}{Chatterjee}{2015}]{Chat2014univ}
Chatterjee, S. (2015, February).
\newblock Matrix estimation by universal singular value thresholding.
\newblock {\em Ann. Statist.\/}~{\em 43\/}(1), 177--214.

\bibitem[\protect\citeauthoryear{Davenport, Plan, van~den Berg, and
  Wootters}{Davenport et~al.}{2012}]{Davenport2012}
Davenport, M.~A., Y.~Plan, E.~van~den Berg, and M.~Wootters (2012).
\newblock 1-bit matrix completion.
\newblock {\em CoRR\/}~{\em abs/1209.3672}.

\bibitem[\protect\citeauthoryear{Fithian and Mazumder}{Fithian and
  Mazumder}{2018}]{fithian2013scalable}
Fithian, W. and R.~Mazumder (2018, 05).
\newblock Flexible low-rank statistical modeling with missing data and side
  information.
\newblock {\em Statist. Sci.\/}~{\em 33\/}(2), 238--260.

\bibitem[\protect\citeauthoryear{Friedman, Hastie, and Tibshirani}{Friedman
  et~al.}{2010}]{friedman2010regularization}
Friedman, J., T.~Hastie, and R.~Tibshirani (2010).
\newblock Regularization paths for generalized linear models via coordinate
  descent.
\newblock {\em Journal of Statistical Software\/}~{\em 33\/}(1), 1.

\bibitem[\protect\citeauthoryear{Gelman and Hill}{Gelman and
  Hill}{2007}]{gelman_multi}
Gelman, A. and J.~Hill (2007, June).
\newblock {\em Data Analysis Using Regression and Multilevel/Hierarchical
  Models}.
\newblock Cambridge University Press.

\bibitem[\protect\citeauthoryear{Gunasekar, Ravikumar, and Ghosh}{Gunasekar
  et~al.}{2014}]{Gunasekar:2014:EFM:3044805.3045106}
Gunasekar, S., P.~Ravikumar, and J.~Ghosh (2014).
\newblock Exponential family matrix completion under structural constraints.
\newblock In {\em Proceedings of the 31st International Conference on
  International Conference on Machine Learning - Volume 32}, ICML'14, pp.\
  II--1917--II--1925. JMLR.org.

\bibitem[\protect\citeauthoryear{{Hastie}, {Mazumder}, {Lee}, and
  {Zadeh}}{{Hastie} et~al.}{2015}]{softImpute}
{Hastie}, T., R.~{Mazumder}, J.~{Lee}, and R.~{Zadeh} (2015, January).
\newblock {Matrix Completion and Low-Rank SVD via Fast Alternating Least
  Squares}.
\newblock {\em The Journal of Machine Learning Research\/}~{\em 16},
  3367--3402.

\bibitem[\protect\citeauthoryear{Heeringa, West, and Berlung}{Heeringa
  et~al.}{2010}]{survey}
Heeringa, S., B.~West, and P.~Berlung (2010).
\newblock {\em Applied Survey Data Analysis}.
\newblock New Yor: Chapman and Hall/CRC.

\bibitem[\protect\citeauthoryear{Hsu, Kakade, and Zhang}{Hsu
  et~al.}{2011}]{Hsu_robustmatrix}
Hsu, D., S.~M. Kakade, and T.~Zhang (2011).
\newblock Robust matrix decomposition with sparse corruptions.
\newblock {\em EEE Transactions on Information Theory\/}~{\em 57\/}(11),
  7221--7234.

\bibitem[\protect\citeauthoryear{{Husson}, {Josse}, {Narasimhan}, and
  {Robin}}{{Husson} et~al.}{2018}]{MLFAMD}
{Husson}, F., J.~{Josse}, B.~{Narasimhan}, and G.~{Robin} (2018, April).
\newblock {Imputation of mixed data with multilevel singular value
  decomposition}.
\newblock {\em arXiv e-prints\/}, arXiv:1804.11087.

\bibitem[\protect\citeauthoryear{Josse and Husson}{Josse and
  Husson}{2016}]{missMDA}
Josse, J. and F.~Husson (2016).
\newblock {missMDA}: A package for handling missing values in multivariate data
  analysis.
\newblock {\em Journal of Statistical Software\/}~{\em 70\/}(1), 1--31.

\bibitem[\protect\citeauthoryear{Kiers}{Kiers}{1991}]{Kiers1991}
Kiers, H. A.~L. (1991, June).
\newblock Simple structure in component analysis techniques for mixtures of
  qualitative and quantitative variables.
\newblock {\em Psychometrika\/}~{\em 56\/}(2), 197--212.

\bibitem[\protect\citeauthoryear{Klopp}{Klopp}{2014}]{Klopp2014}
Klopp, O. (2014).
\newblock Noisy low-rank matrix completion with general sampling distribution.
\newblock {\em Bernoulli\/}~{\em 20\/}(1), 282--303.

\bibitem[\protect\citeauthoryear{Klopp}{Klopp}{2015}]{klopp:hal-01111757}
Klopp, O. (2015).
\newblock {Matrix completion by singular value thresholding: sharp bounds}.
\newblock {\em {Electronic journal of statistics }\/}~{\em 9\/}(2), 2348--2369.

\bibitem[\protect\citeauthoryear{Klopp, Lafond, Moulines, and Salmon}{Klopp
  et~al.}{2015}]{KloppLafond2015}
Klopp, O., J.~Lafond, {\'E}.~Moulines, and J.~Salmon (2015).
\newblock Adaptive multinomial matrix completion.
\newblock {\em Electronic Journal of Statistics\/}~{\em 9}, 2950--2975.

\bibitem[\protect\citeauthoryear{Klopp, Lounici, and Tsybakov}{Klopp
  et~al.}{2017}]{Klopp2017}
Klopp, O., K.~Lounici, and A.~B. Tsybakov (2017, October).
\newblock Robust matrix completion.
\newblock {\em Probability Theory and Related Fields\/}~{\em 169\/}(1),
  523--564.

\bibitem[\protect\citeauthoryear{Koltchinskii}{Koltchinskii}{2011}]{Koltchinskii2011}
Koltchinskii, V. (2011).
\newblock {\em Oracle Inequalities in Empirical Risk Minimization and Sparse
  Recovery}.
\newblock Springer.

\bibitem[\protect\citeauthoryear{Kumar and Schneider}{Kumar and
  Schneider}{2017}]{review-low-rank}
Kumar, N.~K. and J.~Schneider (2017).
\newblock Literature survey on low rank approximation of matrices.
\newblock {\em Linear and Multilinear Algebra\/}~{\em 65\/}(11), 2212--2244.

\bibitem[\protect\citeauthoryear{Lafond}{Lafond}{2015}]{Lafond2015}
Lafond, J. (2015).
\newblock Low rank matrix completion with exponential family noise.
\newblock {\em Journal of Machine Learning Research: Workshop and Conference
  Proceedings\/}~{\em 40}, 1--18.

\bibitem[\protect\citeauthoryear{Landgraf and Lee}{Landgraf and
  Lee}{2015}]{Landgraf15}
Landgraf, A.~J. and Y.~Lee (2015, June).
\newblock Generalized principal component analysis: Projection of saturated
  model parameters.
\newblock Technical report, The Ohio State University, Department of
  Statistics.

\bibitem[\protect\citeauthoryear{Ledoux}{Ledoux}{2001}]{Ledoux2001}
Ledoux, M. (2001).
\newblock {\em The concentration of measure phenomenon}, Volume~89 of {\em
  Mathematical Surveys and Monographs}.
\newblock American Mathematical Society, Providence.

\bibitem[\protect\citeauthoryear{Legendre, Galzin, and
  Harmelin-Vivien}{Legendre et~al.}{1997}]{legendre4corner}
Legendre, P., R.~Galzin, and M.~L. Harmelin-Vivien (1997).
\newblock Relating behavior to habitat: solutions to the fourth-corner problem.
\newblock {\em Ecology\/}~{\em 78\/}(2), 547--562.

\bibitem[\protect\citeauthoryear{Little and Rubin}{Little and
  Rubin}{2002}]{Little02}
Little, R. J.~A. and D.~B. Rubin (2002).
\newblock {\em Statistical Analysis with Missing Data}.
\newblock New-York: John Wiley \& Sons series in probability and statistics.

\bibitem[\protect\citeauthoryear{Mao, Chen, and Wong}{Mao et~al.}{2018}]{Mao}
Mao, X., S.~X. Chen, and R.~K.~W. Wong (2018).
\newblock Matrix completion with covariate information.
\newblock {\em Journal of the American Statistical Association\/}~{\em 0\/}(0),
  1--13.

\bibitem[\protect\citeauthoryear{Mardani, Mateos, and Giannakis}{Mardani
  et~al.}{2013}]{mardani2013}
Mardani, M., G.~Mateos, and G.~B. Giannakis (2013, Aug).
\newblock Recovery of low-rank plus compressed sparse matrices with application
  to unveiling traffic anomalies.
\newblock {\em IEEE Transactions on Information Theory\/}~{\em 59\/}(8),
  5186--5205.

\bibitem[\protect\citeauthoryear{Mazumder, Hastie, and Tibshirani}{Mazumder
  et~al.}{2010}]{mazumder2010spectral}
Mazumder, R., T.~Hastie, and R.~Tibshirani (2010).
\newblock Spectral regularization algorithms for learning large incomplete
  matrices.
\newblock {\em The Journal of Machine Learning Research\/}~{\em 11},
  2287--2322.

\bibitem[\protect\citeauthoryear{Murdoch and Detsky}{Murdoch and
  Detsky}{2013}]{health-mixed}
Murdoch, T. and A.~Detsky (2013).
\newblock The inevitable application of big data to health care.
\newblock {\em JAMA\/}~{\em 309\/}(13), 1351--1352.

\bibitem[\protect\citeauthoryear{Pag\`{e}s}{Pag\`{e}s}{2015}]{Pages2015}
Pag\`{e}s, J. (2015).
\newblock {\em Multiple factor analysis by example using R}.
\newblock Chapman \& Hall/CRC the R series (CRC Press). Taylor \& Francis
  Group.

\bibitem[\protect\citeauthoryear{{R Core Team}}{{R Core Team}}{2017}]{R}
{R Core Team} (2017).
\newblock {\em R: A Language and Environment for Statistical Computing}.
\newblock Vienna, Austria: R Foundation for Statistical Computing.

\bibitem[\protect\citeauthoryear{Srebro and Jaakkola}{Srebro and
  Jaakkola}{2003}]{Srebro:2003:WLA}
Srebro, N. and T.~Jaakkola (2003).
\newblock Weighted low-rank approximations.
\newblock In {\em Proceedings of the Twentieth International Conference on
  International Conference on Machine Learning}, ICML'03, pp.\  720--727. AAAI
  Press.

\bibitem[\protect\citeauthoryear{Talagrand}{Talagrand}{1996}]{talagrand1996}
Talagrand, M. (1996, January).
\newblock A new look at independence.
\newblock {\em Ann. Probab.\/}~{\em 24\/}(1), 1--34.

\bibitem[\protect\citeauthoryear{ter Braak, Peres-Neto, and Dray}{ter Braak
  et~al.}{2017}]{10.7717/peerj.2885}
ter Braak, C.~J., P.~Peres-Neto, and S.~Dray (2017, January).
\newblock A critical issue in model-based inference for studying trait-based
  community assembly and a solution.
\newblock {\em PeerJ\/}~{\em 5}, e2885.

\bibitem[\protect\citeauthoryear{Tropp}{Tropp}{2012}]{Tropp2012}
Tropp, J.~A. (2012).
\newblock User-friendly tail bounds for sums of random matrices.
\newblock {\em Foundations of Computational Mathematics\/}~{\em 12\/}(4),
  389--434.

\bibitem[\protect\citeauthoryear{Tseng and Yun}{Tseng and
  Yun}{2009}]{tseng:yun:2009}
Tseng, P. and S.~Yun (2009).
\newblock A coordinate gradient descent method for nonsmooth separable
  minimization.
\newblock {\em Math. Program.\/}~{\em 117\/}(1-2, Ser. B), 387--423.

\bibitem[\protect\citeauthoryear{Tsybakov}{Tsybakov}{2008}]{Tsybakov:2008:INE:1522486}
Tsybakov, A.~B. (2008).
\newblock {\em Introduction to Nonparametric Estimation\/} (1st ed.).
\newblock Springer Publishing Company, Incorporated.

\bibitem[\protect\citeauthoryear{Udell, Horn, Zadeh, and Boyd}{Udell
  et~al.}{2016}]{glrm}
Udell, M., C.~Horn, R.~Zadeh, and S.~Boyd (2016).
\newblock Generalized low rank models.
\newblock {\em Foundations and Trends in Machine Learning\/}~{\em 9\/}(1).

\bibitem[\protect\citeauthoryear{van Buuren and Groothuis-Oudshoorn}{van Buuren
  and Groothuis-Oudshoorn}{2011}]{mice}
van Buuren, S. and K.~Groothuis-Oudshoorn (2011).
\newblock mice: Multivariate imputation by chained equations in r.
\newblock {\em Journal of Statistical Software, Articles\/}~{\em 45\/}(3),
  1--67.

\bibitem[\protect\citeauthoryear{Xu, Caramanis, and Sanghavi}{Xu
  et~al.}{2010}]{Xu:2010:RPV}
Xu, H., C.~Caramanis, and S.~Sanghavi (2010).
\newblock Robust pca via outlier pursuit.
\newblock In {\em Proceedings of the 23rd International Conference on Neural
  Information Processing Systems}, NIPS'10, USA, pp.\  2496--2504. Curran
  Associates Inc.

\end{thebibliography}

\newpage
\bigskip
\begin{center}
{\large\bf SUPPLEMENTARY MATERIAL}
\end{center}
\appendix
\section{Imputation error by data type and timing results}
\label{sec:imput_by_type}
In this section we provide more details on the simulations of \Cref{subsec:mixed-data}. \Cref{tbl:imput_quanti} presents the imputation errors of the compared methods for quantitative variables only, and \Cref{tbl:imput_quali} for binary variables. For the quantitative variables, mimi and MLFAMD, which both model main group effects, perform best. As already noticed in \Cref{subsec:mixed-data}, mimi has smaller imputation errors than other methods when the size of the main effects compared to the interactions, and the proportion of missing entries, are both large. For the binary variables, suprisingly, softImpute outperforms consistently the other methods, although it is not designed for mixed data. Finally, \Cref{tbl:times} shows the average computational times of the different compared methods. We observe that the computational times of mimi, GLRM, FAMD and MLFAMD are of comparable order. The aforementioned methods are an order of magnitude slower than softImpute and mice.
\begin{table}[!ht]
\centering
\begin{scriptsize}    
    \begin{tabular}{|l|l|l|l|l|l|l|l|l|l|}
    \hline
    {\bf \% missing}& \multicolumn{3}{c|}{\bf 20} &\multicolumn{3}{c|}{\bf 40}  &\multicolumn{3} {c|}{\bf 60}\\
    \hline
    $\mathbf{\rho}$ & {\bf 0.2} &{\bf 1}  & {\bf 5} &{\bf 0.2} &{\bf 1}   &{\bf 5}&{\bf 0.2} &{\bf 1}  & {\bf 5}\\
    \hline
    {\bf mean}   & 20.7(1.3)&19.8(0.7)&19.6(0.6) &      28.0(2.6)&28.2(1.3)&26.9(1.1)   & 35.5(1.6)&34.2(1.3)&34.1(0.5)   \\
    \hline

    {\bf mimi}& 13.0(0.4)&\textbf{12.3}(0.4)&\textbf{11.4}(0.3)   & 19.8(1.1) & \textbf{19.0}(0.7) &\textbf{16.1}(0.5)  &    27.1(1.0)& \textbf{24.3}(1.1)&\textbf{20.2}(0.4) \\
    \hline
    {\bf GLRM} &16.1(1.0) &16.9(0.7) & 13.8(0.4) &    24.0(5.3)&24.5(1.5) &23.4(1.1)&   36.5(12.3)& 41.9(18.0) & 44.1(3.7)\\
    \hline
    {\bf softImpute}    & 14.0(0.5)& 14.0(0.4)& 13.3(0.4)&     20.3(1.2)&20.9(0.7)& 18.5(0.8) &       27.3(1.2)& 27.4(1.0)& 24.4(0.5) \\
    \hline
    {\bf FAMD}   & 12.7(0.5)&12.9(0.6)&12.1(0.3) &      \textbf{19.2}(1.3)&20.2(0.6)&17.3(0.6)     & 26.9(1.8)& 31.2(1.0) &22.7(0.4)   \\
    \hline
    {\bf MLFAMD}   & \textbf{12.6}(0.6)&13.7(0.6)&12.2(0.4) &      18.8(1.0)&19.7(0.6)&17.6(0.7) & \textbf{25.4}(1.5)& 26.2(1.2) & 23.5(0.6)   \\
    \hline
    {\bf mice}   & 17.3(0.8)&17.2(1.0)&16.9(0.6) &      25.1(1.2)&26.0(0.7)&23.1(1.0)   & 40.7(2.8)& 40.1(0.9)& 36.8(1.8)   \\
    \hline
    \end{tabular}
\end{scriptsize}
\caption{Quantitative variables: Imputation error (MSE) of mimi, GLRM, softImpute and FAMD for different percentages of missing entries ($20$\%, $40$\%, $60$\%) and different values of the ratio $\norm{\fu{\Talpha}}[F]/\norm{\TL}[F]$ ($0.2$, $1$, $5$). The values are averaged across $100$ replications and the standard deviation is given between parenthesis.} 
\label{tbl:imput_quanti}
\end{table}

\begin{table}[!ht]
\centering
\begin{scriptsize}    
    
    \begin{tabular}{|l|l|l|l|l|l|l|l|l|l|}
    \hline
    {\bf \% missing}& \multicolumn{3}{c|}{\bf 20} &\multicolumn{3}{c|}{\bf 40}  &\multicolumn{3} {c|}{\bf 60}\\
    \hline
    $\mathbf{\rho}$ & {\bf 0.2} &{\bf 1}  & {\bf 5} &{\bf 0.2} &{\bf 1}   &{\bf 5}&{\bf 0.2} &{\bf 1}  & {\bf 5}\\
    \hline
    {\bf mean}   & 13.0(0.3)&12.4(0.3)&\textbf{11.8}(0.4) &      18.33(0.4)&17.4(0.3)&16.9(0.3)     & 22.6(0.5)& 22.0(0.6) & \textbf{20.8}(0.6)   \\
    \hline
    {\bf mimi}& 13.5(0.3)& 13.5(0.3)&13.5(0.3)    &18.9(0.5)&19.1(0.3)&18.9(0.6) &    23.7(0.6)& 23.4(0.5)& 23.1(0.4)  \\
    \hline
    {\bf GLRM} &14.2(0.4) &14.1(0.6)&14.2(0.5) &    20.0(0.4) &20.2(0.4)&20.4(0.3) &   24.9(0.5)& 25.1(0.6) & 24.9(0.3)\\
    \hline
    {\bf softImpute}    & \textbf{12.2}(0.1)& \textbf{12.0}(0.3)& 12.0(0.6)&     \textbf{17.0}(0.3)& \textbf{16.7}(0.2)& \textbf{16.6}(0.4)&       \textbf{21.6}(0.4)& \textbf{21.6}(0.3)& 21.0(0.5)  \\
    \hline
    {\bf FAMD}   & 13.6(0.4)&13.8(0.4)&13.5(0.3) &      19.2(0.5)&19.8(0.3)&18.8(0.6)     & 24.0(0.5)& 25.0(0.4) & 23.6(0.4)   \\
    \hline
    {\bf MLFAMD}   & 13.6(0.5)&13.5(0.4)&13.6(0.4) &      19.4(0.5)&19.5(0.4)&19.6(0.5) & 24.0(0.5)& 24.1(0.4)& 23.9(0.4)   \\
    \hline
    {\bf mice}   & 14.6(0.3)&14.5(0.4)&14.4(0.4) &      20.5(0.4)&20.3(0.2)&20.5(0.4)    & 25.7(0.4)& 25.7(0.6)& 25.3(0.2)   \\
    \hline
    \end{tabular}
\end{scriptsize}
\caption{Binary variables: Imputation error (MSE) of mimi, GLRM, softImpute and FAMD for different percentages of missing entries ($20$\%, $40$\%, $60$\%) and different values of the ratio $\norm{\fu{\Talpha}}[F]/\norm{\TL}[F]$ ($0.2$, $1$, $5$). The values are averaged across $100$ replications and the standard deviation is given between parenthesis.} 
\label{tbl:imput_quali}
\end{table}

\begin{table}
\begin{tabular}{|l|l|l|l|l|l|l|l|}
\hline
\textbf{method} & \textbf{mean} &\textbf{mimi} &\textbf{GLRM} &\textbf{softImpute} &\textbf{FAMD} &\textbf{MLFAMD} &\textbf{mice} \\
\hline
\textbf{time (s)} &1.7e-4& 6.6&5.5&0.1&2.6&3.5&0.2\\
\hline
\end{tabular}
\caption{Computation time of the seven compared methods (averaged across $100$ simulations).}
\label{tbl:times}
\end{table}

\section{Proof of Theorem \ref{th:convergence}}
\label{proof:convergence}

To prove global convergence of the BCGD algorithm, we use a result from \cite[Theorem 1]{tseng:yun:2009} summarized below in \Cref{th:tseng}, combined with the compacity of the level sets of the objective $F$, proved using \Cref{lem:datfitbound} and \Cref{lem:compactlevelsets}.
\begin{theorem} 
\label{th:tseng}
Let $\{(\alpha^{[k]},L^{[k]})\}$ be the current iterates, $\{(d_{\alpha}^{[k]},d_{L}^{[k]})\}$ the descent directions and $\{(\Gamma^{[k]}_{\alpha}, \Gamma^{[k]}_{L})\}$ the functionals generated by the BCGD algorithm. Then the following results hold.
\begin{enumerate}[label=(\alph*)]
\item \label{th:tseng1} $\{F(\alpha^{[k]},L^{[k]})\}$ is nonincreasing and for all $k$, $(\Gamma^{[k]}_{\alpha}, \Gamma^{[k]}_{L})$ satisfies
\begin{equation*}
 -\Gamma^{[k]}_{\alpha}  \geq (1-\theta)\nu\norm{d_{\alpha}^{[k]}}[2]^2 \text{ and }
 -\Gamma^{[k]}_{L}  \geq (1-\theta)\nu\norm{d_{L}^{[k]}}[F]^2.
\end{equation*}
\item  Every cluster point of $\{(\alpha^{[k]},L^{[k]})\}$ is a stationary point of $F$. \label{th:tseng2}
\end{enumerate}
\end{theorem} 
Assumptions \textbf{H}\ref{ass:dict} and \ref{ass:cvx}, combined with the separability of the $\ell_1$ and nuclear norm penalties, guarantee that the conditions of \cite[Theorem 1]{tseng:yun:2009} are satisfied. We now show that the data-fitting term $\mathcal{L}(\fu{\alpha}+L; Y, \Omega)$ is lower-bounded.
\begin{lemma}
\label{lem:datfitbound}
There exists a constant $c>-\infty$ such that, for all $X\in\mathbb{R}^{m_1\times m_2}$, $\mathcal{L}(X; Y,\Omega)\geq c$.
\end{lemma}
\begin{proof}
Recall that $\mathcal{L}(X; Y,\Omega)=\sum_{i=1}^{m_1}\sum_{j=1}^{m_2}\Omega\{-Y_{ij}X_{ij}+g_j(X_{ij})\}$. Thus, we only need to prove that for all $(i,j)\in\nint{m_1}\times\nint{m_2}$, the function $x\mapsto -Y_{ij}x + g_j(x)$ is lower bounded by a constant $c_{ij}>-\infty$. Assume that this is not the case; by the convexity of $x\mapsto -Y_{ij}x + g_j(x)$ we have that either $-Y_{ij}x + g_j(x)\underset{x\rightarrow+\infty}{\rightarrow} - \infty$ or $-Y_{ij}x + g_j(x)\underset{x\rightarrow-\infty}{\rightarrow} - \infty$. Assume without loss of generality that $-Y_{ij}x + g_j(x)\underset{x\rightarrow+\infty}{\rightarrow} - \infty$. Then, there exists $x_0\in \mathbb{R}$ such that for all $x\geq x_0$, $-Y_{ij}x + g_j(x)< \log\int _{\substack{y\in \mathcal{Y}_j\\y\geq Y_{ij}}} h_j(y)\mu_j(d_y)$. Thus, for all $x\geq \max(x_0,0)$, we have that 
\begin{equation*}
\begin{aligned}
&\int _{y\in \mathcal{Y}_j} h_j(y)\mathrm{e}^{yx-g_j(x)}\mu_j(d_y)&&= \int _{\substack{y\in \mathcal{Y}_j\\y< Y_{ij}}} h_j(y)\mathrm{e}^{yx-g_j(x)}\mu_j(d_y)+\int _{\substack{y\in \mathcal{Y}_j\\y\geq Y_{ij}}} h_j(y)\mathrm{e}^{yx-g_j(x)}\mu_j(d_y)\\
& && > \int _{\substack{y\in \mathcal{Y}_j\\y< Y_{ij}}} h_j(y)\mathrm{e}^{yx-g_j(x)}\mu_j(d_y) + 1 >1,
\end{aligned}
\end{equation*}
contradicting normality of the density $h_j(y)\mathrm{e}^{yx-g_j(x)}$. Thus, there exists $c_{ij}>-\infty$, such that for all $x\in\mathbb{R}$, $-Y_{ij}x + g_j(x)\geq c_{ij}$. Finally we obtain that
$\mathcal{L}(X; Y,\Omega)\geq c=\sum_{i=1}^{m_1}\sum_{j=1}^{m_2}c_{ij}$.
\end{proof}
Finally, we use \Cref{lem:datfitbound} to show the compactness of the level sets of the objective function $F$, defined for $C\in \mathbb{R}$ by 
$$L_C = \{(\alpha, L)\in \mathbb{R}^N\times\mathbb{R}^{m_1\times m_2}; F(\alpha, L)\leq C\}.$$
\begin{lemma}
\label{lem:compactlevelsets}
The level sets of the objective function $F$ are compact.
\end{lemma}
\begin{proof}
For all $(\alpha,L)\in \mathbb{R}^N\times\mathbb{R}^{m_1\times m_2}$, $F(\alpha, L) \geq c + \lambda_1\norm{L}[*] + \lambda_2\norm{\alpha}[1]$, where $c$ is the constant defined in \Cref{lem:datfitbound}. Thus, for all $C\in \mathbb{R}$, the level set $L_C$ is included in the compact set 
$$\left\{(\alpha, L)\in \mathbb{R}^N\times\mathbb{R}^{m_1\times m_2}; \norm{L}[*] \leq \frac{C-c}{2\lambda_1}\text{ and } \norm{\alpha}[1]\leq \frac{C-c}{2\lambda_2}\right\}.$$
Furthermore, by the continuity of $F$, the level set $L_C$ is also a closed set. Thus we obtain that for all $C\in\mathbb{R}$, the level set $L_C$ is compact.
\end{proof}
We can now combine \Cref{th:tseng}, \Cref{lem:datfitbound} and \Cref{lem:compactlevelsets} to prove \Cref{th:convergence}. Let $(\alpha^{[0]},L^{[0]})$ be an initialization point. \Cref{th:tseng}~\ref{th:tseng1} implies that the sequence $(\alpha^{[k]}, L^{[k]})$ generated by the BCGD algorithm lies in the level set of $F$ 
$$L_{F(\alpha^{[0]},L^{[0]})} =  \left\{(\alpha, L)\in \mathbb{R}^N\times\mathbb{R}^{m_1\times m_2}; F(\alpha,L)\leq F(\alpha^{[0]},L^{[0]}) \right\}.$$
Furthermore, $L_{F(\alpha^{[0]},L^{[0]})}$ is compact by \Cref{lem:compactlevelsets}, showing that the sequence $(\alpha^{[k]}, L^{[k]})$ has at least one accumulation point. Combined with \Cref{th:tseng}~\ref{th:tseng2} and the convexity of $F$, this shows \Cref{th:convergence}~\ref{th:convergence1}. \\

\Cref{th:tseng}~\ref{th:tseng1} and \Cref{lem:datfitbound} combined imply that the sequence $\{F(\alpha^{[k]},L^{[k]})\}$ converges to a limit $F^*$. Furthermore, \Cref{th:convergence}~\ref{th:convergence1} and the continuity of $F$ imply that there exists a sub-sequence $\{F(\alpha^{[k]},L^{[k]})\}_{k\in \mathcal{K}}$ such that $\{F(\alpha^{[k]},L^{[k]})\}_{k\in \mathcal{K}}\rightarrow F(\hat\alpha,\hat L)$. Thus, $F^*=F(\hat\alpha,\hat L)$, which proves \Cref{th:convergence}~\ref{th:convergence2}.

\section{Proof of Theorem \ref{th:general-th}}
\label{proof:general-th}
Let $\Pi = (\pi_{ij})_{(i,j)\in\nint{m_1}\times \nint{m_2}}$ be the distribution of the mask $\Omega$. For $B\in\R^{m_1\times m_2}$ we denote $B_{\Omega}$ the projection of $B$ on the set of observed entries. We define $\norm{B}[\Omega]^2 = \norm{B_{\Omega}}[F]^2$, and $\norm{B}[\Pi]^2 = \PE[{\norm{B}[\Omega]^2}]$, where the expectation is taken with respect to $\Pi$.
The proof of \Cref{th:general-th} will follow the subsequent two steps. We first derive an upper bound on the Frobenius error restricted to the observed entries $\norm{\Delta X }[\Omega]^2$, then show that the expected Frobenius error $\norm{\Delta X }[\Pi]^2$ is upper bounded by $\norm{\Delta X }[\Omega]^2$ with high probability, and up to a residual term defined later on.

Let us derive the upper bound on $\norm{\Delta X }[\Omega]^2$. By definition of $\hat{L}$ and $\hat{\alpha}$:
$
\mathcal{L}(\hat{X};Y,\Omega) -  \mathcal{L}(\TX;Y,\Omega)
\leq \lambda_1\left(\norm{\TL}[*]-\norm{\hat{L}}[*] \right) + \lambda_2\left(\norm{\Talpha}[1] - \norm{\hat{\alpha}}[1] \right).
$
Recall that, for $\alpha\in\R^N$, we use the notation $\fu{\alpha} = \sum_{k=1}^N\alpha_kU^k.$
Adding $\pscal{\nabla  \mathcal{L}(\TX; Y,\Omega)}{\Delta X}$ on both sides of the last inequality, we get
\begin{multline}
\label{eq:key0}
\mathcal{L}(\hat{X};Y,\Omega) -  \mathcal{L}(\TX;Y,\Omega)+\pscal{\nabla  \mathcal{L}(\TX; Y,\Omega)}{\Delta X} \leq \\
\lambda_1\left(\norm{\TL}[*]-\norm{\hat{L}}[*] \right)-\pscal{\nabla  \mathcal{L}(\TX; Y,\Omega)}{\Delta L}\\
  +\lambda_2\left(\norm{\Talpha}[1] - \norm{\hat{\alpha}}[1]\right) - \pscal{\nabla  \mathcal{L}(\TX; Y,\Omega)}{\fu{\Delta \alpha}}.
\end{multline}
Assumption \textbf{H}\ref{ass:cvx} implies that for any pair of matrices  $X^1$ and $X^2$ in $\mathbb{R}^{m_1\times m_2}$ satisfying $\norm{X^1}[\infty]\vee\|X^2\|_{\infty}\leq (1+\ae)a$, the two following inequalities hold for all $\Omega$:
\begin{equation}
\label{eq:strong-cvx}
\mathcal{L}(X;Y,\Omega) - \mathcal{L}(\tilde{X};Y,\Omega) - \pscal{\nabla\mathcal{L}(\tilde{X};Y,\Omega)}{X-\tilde{X}} \geq \frac{\smin^2}{2}\norm{X-\tilde{X}}[\Omega]^2,
\end{equation}
\begin{equation}
\label{eq:smooth}
\norm{\nabla\mathcal{L}(X;Y,\Omega)-\nabla\mathcal{L}(\tilde{X};Y,\Omega)}[F] \leq \smax^2\norm{X-\tilde{X}}[\Omega].
\end{equation}
Plugging \eqref{eq:strong-cvx} into \eqref{eq:key0} allows to construct a lower bound on the left hand side term  and obtain $\smin^2\norm{\Delta X }[\Omega]^2/2 \leq  \mathsf{A}_1 + \mathsf{A}_2$,
\begin{equation}
\label{eq:subres-X}
\begin{aligned}
& \mathsf{A}_1  && =  \lambda_1\left(\norm{\TL}[*]-\norm{\hat{L}}[*] \right)+\left|\pscal{\nabla  \mathcal{L}(\TX; Y,\Omega)}{\Delta L}\right|,\\
& \mathsf{A}_2 && = \lambda_2\left(\norm{\Talpha}[1] - \norm{\hat{\alpha}}[1]\right)+   \left|\pscal{\nabla  \mathcal{L}(\TX; Y,\Omega)}{\fu{\Delta \alpha}}\right|.
\end{aligned}
\end{equation}
Let us upper bound $\textsf{A}_1$. The duality of the norms $\norm{\cdot}[*]$ and $\norm{\cdot}$ implies that
\begin{equation*}
\left|\pscal{\nabla  \mathcal{L}(\TX; Y,\Omega)}{\Delta L}\right| \leq \norm{\nabla  \mathcal{L}(\TX; Y,\Omega)}\norm{\Delta L}[*].
\end{equation*}
Denote by $S_1$ and $S_2$ the linear subspaces spanned respectively by the left and right singular vectors of $\TL$, and $P_{S_1^{\perp}}$ and $P_{S_2^{\perp}}$ the orthogonal projectors on the orthogonal of $S_1$ and $S_2$, $P_{{\TL}^{\perp}} : X \mapsto P_{S_1^{\perp}}XP_{S_2^{\perp}}$ and $P_{{\TL}} : X \mapsto X - P_{S_1^{\perp}}XP_{S_2^{\perp}}$. The triangular inequality yields
\begin{equation}
\label{eq:subres-th1-L}
\norm{\hat{L}}[*]= \norm{\TL - P_{{\TL}^{\perp}}(\Delta L) - P_{{\TL}}(\Delta L)}[*] \geq \norm{\TL+P_{{\TL}^{\perp}}(\Delta L)}[*] - \norm{P_{{\TL}}(\Delta L)}[*].
\end{equation}
Moreover, by definition of $P_{{\TL}^{\perp}}$, the left and right singular vectors of $P_{{\TL}^{\perp}}(\Delta L)$ are respectively orthogonal to the left and right singular spaces of $\TL$, implying $\norm{\TL+P_{{\TL}^{\perp}}(\Delta L)}[*]  =\norm{\TL}[*]+\norm{P_{{\TL}^{\perp}}(\Delta L)}[*]$. Plugging this identity into \eqref{eq:subres-th1-L} we obtain
\begin{equation}\label{eq:L-nuc}
\norm{\TL}[*] -\norm{\hat{L}}[*] \leq \norm{P_{{\TL}}(\Delta L)}[*]- \norm{P_{{\TL}^{\perp}}(\Delta L)}[*],
\end{equation}
and $\textsf{A}_1\leq \lambda_1\left(\norm{P_{{\TL}}(\Delta L)}[*]- \norm{P_{{\TL}^{\perp}}(\Delta L)}[*] \right)+\norm{\nabla  \mathcal{L}(\TX; Y, \Omega)}\norm{\Delta L}[*].$

Using $\norm{\Delta L}[*] \leq \norm{P_{{\TL}}(\Delta L)}[*]+\norm{P_{{\TL}^{\perp}}(\Delta L)}[*]$ and the assumption $\lambda_1\geq 2\norm{\nabla  \mathcal{L}(\TX; Y,\Omega)}$ we get
$
\textsf{A}_1 \leq 3\lambda_1\norm{P_{{\TL}}(\Delta L)}[*]/2.
$
In addition, $\norm{P_{{\TL}}(\Delta L)}[*]\leq \sqrt{\rank{P_{{\TL}}(\Delta L)}}\norm{P_{{\TL}}(\Delta L)}[F]$, and $\rank{P_{{\TL}}(\Delta L)}\leq 2\rank{\TL}$ (see, \textit{e.g.} \cite[Theorem 3]{Klopp2014}). Together with $\norm{P_{{\TL}}(\Delta L)}[F]\leq \norm{\Delta L}[F]$, this finally implies the following upper bound:
\begin{equation}
\label{eq:sub-result-L-2}
\textsf{A}_1 \leq \frac{3\lambda_1}{2}\sqrt{2r}\norm{\Delta L}[F].
\end{equation}
We now derive an upper bound for $\textsf{A}_2$. The duality between $\norm{\cdot}[1]$ and $\norm{\cdot}[\infty]$ ensures
\begin{equation}
\label{eq:subres}
\left|\pscal{\nabla  \mathcal{L}(\TX; Y, \Omega)}{\fu{\Delta \alpha}}\right|
\leq \norm{\Delta\alpha}[1]\norm{\nabla\mathcal{L}(\TX; Y, \Omega)}[\infty]\umax.
\end{equation}
The assumption $\lambda_2 \geq 2\norm{\nabla\mathcal{L}(\TX; Y, \Omega)}[\infty]\umax$ in conjunction with \eqref{eq:subres} and the triangular inequality $\norm{\Delta\alpha}[1]\leq\norm{\Talpha}[1]+ \norm{\hat{\alpha}}[1]$ yield
\begin{equation}
\label{eq:sub-result-S}
\textsf{A}_2 \leq \frac{3\lambda_2}{2}\norm{\Talpha}[1].
\end{equation}
Combining inequalities \eqref{eq:subres-X}, \eqref{eq:sub-result-L-2} and \eqref{eq:sub-result-S} we obtain
\begin{equation}
\label{eq:ineq_x_omega}
\norm{\Delta X  }[\Omega]^2\leq \frac{3\lambda_1}{\smin^2}\sqrt{2r}\norm{\Delta L}[F]+ \frac{3\lambda_2}{\smin^2}\norm{\Talpha}[1].
\end{equation}

We now show that when the errors $\Delta L$ and $\Delta \alpha$ belong to a subspace $\mathcal{C}$ and for a residual $\mathsf{D}$ - both defined later on - the following holds with high probability:
\begin{equation}
\label{eq:res-str-cvx}
\norm{\Delta X }[\Omega]^2\geq \norm{\Delta X}[\Pi]^2 -\mathsf{D}.
\end{equation}
We start by defining our constrained set and prove that it contains the errors $\Delta L$ and $\Delta \alpha$ with high probability (Lemma \ref{lemma:alpha-l1}-\ref{lemma:L-nuc}); then we show that restricted strong convexity holds on this subspace (\Cref{lemma:rsc}).
For non-negative constants $d_1$, $d_{\Pi}$, $\rho < m$ and $\varepsilon$ that will be specified later on, define the two following sets where $\Delta\alpha$ and $\Delta L$ should lie:
\begin{equation}
\label{eq:sets-A}
\mathcal{A}(d_1,d_{\Pi}) =  \left\lbrace \alpha\in\R^{N}\text{ : } \norm{\alpha}[1]\leq d_1\text{, }\norm{\fu{\alpha}}[\Pi]^2\leq d_{\Pi}\right\rbrace.
\end{equation}
\begin{equation}
\label{eq:sets-C}
\begin{aligned}
& \mathcal{L}(\rho,\varepsilon) && = \Bigg\lbrace L\in\R^{m_1\times m_2}, \alpha\in\R^N: \norm{L+\fu{\alpha}}[\Pi]^2\geq \frac{72\log(d)}{p\log(6/5)},\\
& && \quad \quad \quad \quad\quad \quad \quad \quad
\norm{L+\fu{\alpha}}[\infty]\leq 1,\norm{L}[*]\leq \sqrt{\rho}\norm{L}[F] + \varepsilon \Bigg\rbrace
\end{aligned}
\end{equation}
If $\norm{\Delta X}[\Pi]^2$ is too small, the right hand side of \eqref{eq:res-str-cvx} is negative. The first inequality in the definition of $\mathcal{L}(\rho,\varepsilon)$ prevents from this. Condition $\norm{L}[*]\leq \sqrt{\rho}\norm{L}[F] + \varepsilon$ is a relaxed form of the condition $\norm{L}[*]\leq \sqrt{\rho}\norm{L}[F]$ satisfied for matrices of rank $\rho$.
Finally, we define the constrained set of interest:
$$\mathcal{C}(d_1,d_{\Pi},\rho,\varepsilon) = \mathcal{L}(\rho,\varepsilon) \cap \left\lbrace \R^{m_1\times m_2}\times \mathcal{A}(d_1,d_{\Pi}) \right\rbrace.$$
Recall $\umax = \max_k\norm{U_k}[1]$ and let
\begin{equation*}
 d_1  = 4\norm{\Talpha}[1],\text{ and } d_{\Pi}  = \frac{3\lambda_2}{\smin^2}\norm{\Talpha}[1] + 64a^2\umax\PE[{\norm{\Sigma_R}[\infty]}]\norm{\Talpha}[1] + 3072a^2p^{-1} + \frac{72a^2\log(d)}{\log(6/5)}.
\end{equation*}
\begin{lemma}
\label{lemma:alpha-l1}
Let $\lambda_2\geq 2\umax\left(\norm{\diff{\TX}}[\infty]+ 2\smax^2(1+u)a\right)$ and assume \textbf{H}\ref{ass:dict}-\ref{ass:cvx} hold. Then, with probability at least $1-8d^{-1}$,
$\Delta\alpha \in \mathcal{A}(d_1,d_{\Pi}).$
\end{lemma}
\begin{proof}
See \Cref{lemma:alpha-l1-proof}.
\end{proof}
\Cref{lemma:alpha-l1} implies the upper bound on $\norm{\Delta\alpha}[2]^2$ of \Cref{th:general-th}. Thus, we only need to prove the upper bound on $\norm{\Delta L}[F]^2$. Let $\rho = 32r$ and $\varepsilon = 3\lambda_2/\lambda_1\norm{\Talpha}[1]$.
\begin{lemma}
\label{lemma:L-nuc}
Assume \textbf{H}\ref{ass:cvx} and let
$$\lambda_1\geq 2\norm{\diff{\TX}},\quad \lambda_2\geq 2\umax\left(\norm{\diff{\TX}}[\infty]+ 2\smax^2(1+u)a\right).$$
Then $
\norm{\Delta L}[*] \leq \sqrt{\rho}\norm{\Delta L}[F] + \varepsilon.$
\end{lemma}
\begin{proof}
See \Cref{lemma:L-nuc-proof}
\end{proof}
As a consequence, under the conditions on the regularization parameters $\lambda_1$ and $\lambda_2$ given in \Cref{lemma:L-nuc} and whenever $\norm{\Delta L+\fu{\Delta\alpha}}[\Pi]^2\geq 72\log(d)/(p\log(6/5)),$ the error terms $(\Delta L,\Delta \alpha)$ belong to the constrained set $\mathcal{C}(d_1,d_{\Pi},\rho,\varepsilon)$ with high probability. \\

\textbf{Case 1:} Suppose $\norm{\Delta L + \fu{\Delta \alpha}}[\Pi]^2<72\log(d)/(p\log(6/5))$. Then, \Cref{lemma:alpha-l1} combined with the fact that $\norm{M}[F]^2\leq p^{-1}\norm{M}[\Pi]^2$ for all $M$, and the identity $(a+b)^2\geq a^2/4-4b^2$ ensures that
$\norm{\Delta L}[F]^2 \leq 4\norm{\Delta L + \fu{\Delta \alpha}}[F]^2 + 16 \norm{\fu{\Delta\alpha}}[F]^2.$
Therefore we obtain (ii) of \Cref{th:general-th}:
$$\norm{\Delta L}[F]^2 \leq \frac{288a^2\log(d)}{\log(6/5)} + 16\frac{\norm{\Talpha}[1]}{p}\Theta_1.$$

\textbf{Case 2:}
 Suppose $\norm{\Delta L + \fu{\Delta \alpha}}[\Pi]^2\geq 72\log(d)/(p\log(6/5))$. Then, \Cref{lemma:alpha-l1}  and \ref{lemma:L-nuc} yield that with probability at least $1-8d^{-1}$,
$$\left(\frac{\Delta L}{2(1+\ae)a},\frac{\Delta \alpha}{2(1+\ae)a}\right)\in\mathcal{C}(d_1',d_{\Pi}',\rho', \varepsilon'), \text{where}$$
$$
d_1' = \frac{d_1}{2(1+\ae)a},\quad
 d_{\Pi}' = \frac{d_{\Pi}}{4(1+\ae)^2a^2},\quad
 \rho' = \rho, \quad \varepsilon'  = \frac{\varepsilon}{2(1+\ae)a},
$$
and where $d_1,d_{\Pi},\rho$ and $\varepsilon$ are the same as in \Cref{lemma:alpha-l1}  and \ref{lemma:L-nuc}.
We use the following result, proven in \Cref{proof:lemma:rsc}. Recall that we assume for all $(i,j)\in\nint{m_1}\times\nint{m_2}$, $\mathbb{P}(\Omega_{ij} = 1)\geq p$
 and define:
\begin{equation*}
\label{eq:set-A}
\tilde{\mathcal{A}}(d_1) = \left\lbrace \alpha\in\R^N:\quad \norm{\alpha}[\infty]\leq 1;\quad \norm{\alpha}[1]\leq d_1;\quad \norm{\fu{\alpha}}[\Pi]^2\geq \frac{18\log(d)}{p\log(6/5)}\right\rbrace,
\end{equation*}
\begin{equation}
\label{eq:residuals}
\begin{aligned}
& \mathsf{D}_{\alpha} && = 8 \ae d_1 \umax\PE[{\norm{\Sigma_R}[\infty]}]+768p^{-1},\\
& \mathsf{D}_{X} && = \frac{112\rho}{p}\PE[{\norm{\Sigma_R}}]^2 + 8\ae \varepsilon \PE[{\norm{\Sigma_R}}]+8\ae d_1\umax\PE[{\norm{\Sigma_R}[\infty]}] + d_{\Pi}+768p^{-1}.
\end{aligned}
\end{equation}
\begin{lemma}
\label{lemma:rsc}
\begin{enumerate}[label=(\roman*)]
\item For any $\alpha\in\tilde{\mathcal{A}}(d_1)$, with probability at least $1-8d^{-1}$,
$$\norm{\fu{\alpha} }[\Omega]^2\geq \frac{1}{2}\norm{\fu{\alpha}}[\Pi]^2- \mathsf{D}_{\alpha}.$$
\item For any pair $(L,\alpha)\in \mathcal{C}(d_1,d_{\Pi},\rho,\varepsilon)$, with probability at least $1-8d^{-1}$
\begin{equation}
\norm{L+\fu{\alpha}}[\Omega]^2\geq \frac{1}{2}\norm{L+\fu{\alpha}}[\Pi]^2- \mathsf{D}_{X}.
\end{equation}
\end{enumerate}
\end{lemma}
\begin{proof}
See \Cref{proof:lemma:rsc}.
\end{proof}
\Cref{lemma:rsc} (ii) applied to $\left(\frac{\Delta L}{2(1+\ae)a},\frac{\Delta \alpha}{2(1+\ae)a}\right)$ implies that with probability at least $1-8d^{-1}$,
$\norm{\Delta X}[\Pi]^2 \leq 2\norm{\Delta X}[\Omega]^2 + 4(1+\ae)a\mathsf{D}_X $. Combined with \eqref{eq:ineq_x_omega}, $\norm{\Delta X}[F]^2\leq p^{-1}\norm{\Delta X}[\Pi]^2$, $\norm{\Delta X}[F]^2\geq \norm{\Delta L}[F]^2/2 - \norm{\fu{\Delta\alpha} }[F]^2$ and
$6\sqrt{2r}\lambda_1/(p\smin^2)\norm{\Delta L}[F] \leq \norm{\Delta L}[F]^2/4 + 288r\lambda_1^2/(p^2\smin^4)$, we obtain the result of \Cref{th:general-th} (ii):
$$\norm{\Delta L}[F]^2\leq \frac{1152r\lambda_1^2}{p^2\smin^4} + \frac{24\lambda_2\norm{\Talpha}[1]}{p\smin^2} +4(1+\ae)a\mathsf{D}_X+4\frac{\norm{\Talpha}}{p}\Theta_1.$$

\section{Proof of \Cref{th:lower-bound}}
\label{th:lower-bound-proof}
We will establish separately two lower bounds of order $rM/p$ and $s/p$ respectively.
Define
\begin{equation*}
\tilde{\mathcal{L}} = \left\lbrace \tilde{L}\in\mathbb{R}^{m_1\times r}: \tilde{L}_{ij}\in\left\lbrace 0,\eta \min(a,\smax)\left(\frac{r}{p m}\right)^{1/2}\right\rbrace,
 \forall (i,j)\in\nint{m_1}\times\nint{r} \right\rbrace,
\end{equation*}
where $0\leq \eta \leq 1$ will be chosen later. Define also the associated set of block matrices
\begin{equation*}
\mathcal{L} = \left\lbrace L =  ( \tilde{L}|\ldots |\tilde{L}|O)\in \mathbb{R}^{m_1\times m_2} : \tilde{L}\in \tilde{\mathcal{L}}\right\rbrace,
\end{equation*}
where $O$ denotes the $m_1\times (m_2 -r \left\lfloor m_2/r \right \rfloor)$ null matrix and, for some $x\in\R$, $\left\lfloor x\right \rfloor$ is the integer part of $x$. We also define the following set of vectors
\begin{equation*}
\mathcal{A} = \left\lbrace \alpha  = (\tilde{O}|\tilde{\alpha})\in \mathbb{R}^{N}, \tilde{\alpha}_k\in\{0, \tilde\eta\min(a,\smax)\}\text{ } \forall 1\leq k \leq s\right\rbrace,
\end{equation*}
with $\tilde{O}\in \mathbb{R}^{m_2-s}$ denoting the null vector. Finally, we set
\begin{equation*}
\mathcal{X} = \left\lbrace X = L + \fu{\alpha} \in \mathbb{R}^{m_1\times m_2}, \alpha\in \mathcal{A}, L\in\mathcal{L} \right\rbrace.
\end{equation*}
For any $X\in\mathcal{X}$ there exists a matrix $L\in\mathcal{L}$ of rank at most $r$ and a vector $\alpha$ with at most $s$ non-zero components satisfying $X = L +\fu{\alpha}$. Furthermore, for any $\tilde X\in\mathcal{X}$ there exists a matrix $\tilde L\in\mathcal{L}$ of rank at most $r$ and a vector $\tilde \alpha$ with at most $s$ non-zero components satisfying $X -\tilde X = \tilde L +\fu{\tilde\alpha}$.
Finally, for all $X\in\mathcal{X}$ and $(i,j)\in\nint{m_1}\times\nint{m_2}$, $0\leq X_{ij}\leq(1+\ae)a$. Thus, $\mathcal{X}\subset \mathcal{F}(r,s)$, where $\mathcal{F}(r,s)$ is defined in \eqref{eq:set}.
\paragraph{ Lower bound of order $rM/p$.} Consider the set
$$\mathcal{X}_L = \{X = L + \fu{\alpha} \in \mathcal{X}; \alpha = 0\}.$$
Lemma 2.9 in \cite{Tsybakov:2008:INE:1522486} (Varshamov Gilbert bound) implies that there exists a subset $\mathcal{X}_L^0\subset \mathcal{X}_L$ satisfying $\operatorname{Card}(\mathcal{X}_L^0)\geq 2^{rM/8}+1$, such that the zero $m_1\times m_2$ matrix $\mathbf{0} \in \mathcal{X}_L^0$, and that for any two $X$ and $X'$ in $\mathcal{X}_L^0$, $X\neq X'$ we have
\begin{equation}
\label{eq:min-norm}
\norm{X-X'}[F]^2\geq \frac{Mr}{8}\left(\eta^2\min(a,\smax)^2\frac{r}{pm}\left\lfloor\frac{m_2}{r}\right\rfloor\right)\geq\frac{\eta^2}{16} \min(a^2,\smax^2)\frac{rM}{p}.
\end{equation}
For $X\in \mathcal{X}_L^0$ we compute the Kullback-Leibler divergence $\operatorname{KL}(\mathbb{P}_{\bf 0},\mathbb{P}_X)$ between $\mathbb{P}_{\bf 0}$ and $\mathbb{P}_X$. Using Assumption \textbf{H}\ref{ass:cvx} we obtain
\begin{equation}
\label{eq:kl-div}
 \operatorname{KL}(\mathbb{P}_{\bf 0},\mathbb{P}_X) = \sum_{i,j}\pi_{ij}\left(g_j(X_{ij})-g_j(0)-g_j'(0)X_{ij}\right)
 \leq \frac{\smax^2\eta^2\min(a,\smax)^2Mr}{2}.
\end{equation}
Inequality \eqref{eq:kl-div} implies that
\begin{equation}
\frac{1}{\operatorname{Card}(\mathcal{X}_L^0)-1}\sum_{X\in \mathcal{X}_L^0}\operatorname{KL}(\mathbb{P}_{\bf 0},\mathbb{P}_X)\leq \frac{1}{16}\log(\operatorname{Card}(\mathcal{X}_L^0)-1)
\end{equation}
 is satisfied for $\tilde\eta = \min\left\lbrace 1, \left(8\smax\min(a,\smax)\right)^{-1}\right\rbrace$. Then, conditions \eqref{eq:min-norm} and \eqref{eq:kl-div} guarantee that we can apply Theorem 2.5 from \cite{Tsybakov:2008:INE:1522486}. We obtain that for some constant $\delta>0$ and with $\Psi_1 = C\min\left(\smax^{-2},\min(a,\smax^2)\right)$:
\begin{equation}
\label{eq:inf}
\inf_{\hat{L},\hat{\alpha}}\sup_{(\TL,\Talpha)\in \mathcal{E}} \mathbb{P}_{\TX}\left(\norm{\Delta L}[F]^2+\norm{\Delta\alpha}[2]^2>\frac{\Psi_1rM}{p}\right)\geq \delta,
\end{equation}
\paragraph{Lower bound of order $s/p$.} Using again the Varshamov-Gilbert bound (\cite{Tsybakov:2008:INE:1522486}, Lemma 2.9) we obtain that there exists a subset $\mathcal{A}^0\in \mathcal{A}$ satisfying $\operatorname{Card}(\mathcal{A}^0)\geq 2^{s/8}+1$ and containing the null vector ${\bf 0}\in\mathbb{R}^N$  and such that, for any $\alpha$ and $\alpha '$ of $\mathcal{A}^0$, $\alpha\neq \alpha '$,
\begin{equation}
\label{eq:min-norm-alpha}
\norm{\alpha-\alpha'}[2]^2\geq \frac{s}{8}\tilde\eta^2\min(a,\smax)^2.
\end{equation}
Define $\mathcal{X}_{\alpha} \subset \mathcal{X}$ the set of matrices $X =\fu{\alpha}$ such that $\alpha\in \mathcal{A}^0$ and $L = 0$. For any $X\in \mathcal{X}_{\alpha}$ we compute the Kullback-Leibler divergence $\operatorname{KL}(\mathbb{P}_{\bf 0},\mathbb{P}_X)$ between $\mathbb{P}_{\bf 0}$ and $\mathbb{P}_X$
\begin{equation}
\operatorname{KL}(\mathbb{P}_{\bf 0},\mathbb{P}_X) = \sum_{i,j}\pi_{ij}(g_j(X_{ij})-g_j(0)-g_j'(0)X_{ij}\leq \smax^2\norm{\fu{\alpha}}[\Pi]^2\leq \smax^2p\norm{\fu{\alpha}}[F]^2.
\end{equation}
Using Assumption \textbf{H}\ref{ass:cvx}
\begin{equation}
\label{eq:kl-div-alpha}
\begin{aligned}
&\operatorname{KL}(\mathbb{P}_0,\mathbb{P}_X) &&\leq \smax^2p\left(\max_k\norm{U^k}[F]^2 + 2\tau\right)\norm{\alpha}[2]^2\\
& &&\leq s\smax^2p\left(\max_k\norm{U^k}[F]^2 + 2\tau\right)\tilde\eta^2\min(a,\smax)^2.
\end{aligned}
\end{equation}
From \eqref{eq:kl-div-alpha} we deduce that
\begin{equation}
\frac{1}{\operatorname{Card}(\mathcal{A}^0)-1}\sum_{\mathcal{A}^0}\operatorname{KL}(\mathbb{P}_0,\mathbb{P}_X)\leq sp\left(\max_k\norm{U^k}[F]^2+2\tau\right)\smax^2\tilde\eta^2\min(a,\smax)^2.
\end{equation}
Choosing $\tilde\eta = \min\left\lbrace 1, \left(\sqrt{p}\smax\max_k(\norm{U^k}[F]+2\tau)\min(a,\smax)\right)^{-1}\right\rbrace,$ we now use \cite{Tsybakov:2008:INE:1522486}, Theorem 2.5 which implies for some constant $\delta>0$
\begin{equation}
\label{eq:inf-alpha}
\inf_{\hat{L},\hat{\alpha}}\sup_{(\TL,\Talpha)\in \mathcal{E}} \mathbb{P}_{\TX}\left\lbrace\norm{\Delta L}[F]^2+\norm{\sum_{k=1}^N(\Talpha_k - \hat{\alpha}_k)U^k}[F]^2>\Psi_2\frac{s\kappa^2}{p}\right\rbrace \geq \delta,
\end{equation}
$$\Psi_2 =  C\left(\frac{1}{\smax^2\left(\max_k\norm{U^k}[F]^{2}+2\tau\right)}\wedge (a\wedge\smax)^2\right),$$
where we have used that 
$\norm{\sum_{k=1}^N(\Talpha_k - \hat{\alpha}_k)U^k}[F]^2\geq \kappa^2\norm{\hat\alpha-\Talpha}[2]^2$. We finally obtain the result by combining \eqref{eq:inf} and \eqref{eq:inf-alpha}.

\section{Proof of \Cref{lemma:alpha-l1}}
\label{lemma:alpha-l1-proof}

We start by proving $\norm{\Delta\alpha}[1]\leq 4\norm{\Talpha}[1]$. By the optimality conditions over a convex set \cite[Chapter 4, Section 2, Proposition 4]{AubinEkeland}, there exist two subgradients $\hat{f}_{L} $ in the subdifferential of $\norm{\cdot}[*]$ taken at $\hat{L}$ and $\hat{f}_{\alpha}$ in the subdifferential of $\norm{\cdot}[1]$ taken at $\hat{\alpha}$, such that for all feasible pairs $(L,\alpha)$ we have
\begin{equation}
\label{eq:opt-cond}
\pscal{\diff{\hat{X}}}{L-\hat{L}+\sum_{k=1}^N(\alpha_k-\hat{\alpha}_k)U^k}+\lambda_1\pscal{\hat{f}_{L}}{L-\hat{L}}+\lambda_2\pscal{\hat{f}_{\alpha}}{\alpha-\hat{\alpha}}\geq 0.
\end{equation}
Applying inequality \eqref{eq:opt-cond} to the pair $(\hat{L},\Talpha)$ we obtain
$\pscal{\diff{\hat{X}}}{\sum_{k=1}^N\Delta\alpha_kU^k}+\lambda_2\pscal{\hat{f}_{\alpha}}{\Delta\alpha}\geq 0.
 $
Denote $\tilde X = \hat{L} + \sum_{k=1}^N \Talpha_kU^k$. The last inequality is equivalent to
\begin{multline*}
\underbrace{\pscal{\diff{\TX}}{\sum_{k=1}^N\Delta\alpha_kU^k}}_{\mathsf{B}_1}+\underbrace{\pscal{ \diff{\tilde{X}}-\diff{\TX}}{\sum_{k=1}^N\Delta\alpha_kU^k}}_{\mathsf{B}_2}\\+\underbrace{\pscal{\diff{\hat{X}}- \diff{\tilde{X}}}{\sum_{k=1}^N\Delta\alpha_kU^k}}_{\mathsf{B}_3}+\lambda_2\pscal{\hat{f}_{\alpha}}{\Delta\alpha}\geq 0.
\end{multline*}
We now derive upper bounds on the three terms $\mathsf{B}_1$, $\mathsf{B}_2$ and $\mathsf{B}_3$ separately. Recall that we denote $\umax = \max_k\norm{U^k}[1]$ and use \eqref{eq:subres} to bound $\mathsf{B}_1$:
\begin{equation}
\label{eq:boundI}
\mathsf{B}_1\leq \norm{\Delta\alpha}[1]\norm{\diff{\TX}}[\infty]\umax.
\end{equation}
The duality between $\linf$ and $\lone$ gives
$\mathsf{B}_2\leq \norm{\Delta \alpha}[1]\norm{ \diff{\tilde{X}}-\diff{\TX}}[\infty]\umax.$
Moreover,  $\diff{\tilde{X}}-\diff{\TX}$ is a matrix with entries $g_j'(\tilde{X}_{ij})-g_j'(\TX_{ij})$, therefore assumption \textbf{H}\ref{ass:cvx} ensures
$\norm{\diff{\tilde{X}}-\diff{\TX}}[\infty]\leq 2\smax^2(1+\ae)a, $
and finally we obtain
\begin{equation}
\label{eq:boundII}
\mathsf{B}_2\leq \norm{\Delta \alpha}[1]2\smax^2(1+\ae)a\umax .
\end{equation}
We finally bound $\mathsf{B}_3$ as follows. We have that
$\mathsf{B}_3 = \sum_{i=1}^{m_1}\sum_{j=1}^{m_2}\Omega_{ij}(g_j'(\hat{X}_{ij})-g_j'(\tilde{X}_{ij}))(\tilde{X}_{ij}-\hat{X}_{ij}).$
Now, for all $j\in\nint{m_2}$, $g_j'$ is increasing therefore
$(g_j'(\hat{X}_{ij})-g_j'(\tilde{X}_{ij}))(\tilde{X}_{ij}-\hat{X}_{ij})\leq 0,$ which implies $\mathsf{B}_3\leq 0.$
 Combined with \eqref{eq:boundI} and \eqref{eq:boundII} this yields
$$\lambda_2\pscal{\hat{f}_{\alpha}}{\hat\alpha - \Talpha}\leq  \norm{\Delta\alpha}[1]\umax\left(\norm{\diff{\TX}}[\infty]+ 2\smax^2(1+\ae)a\right).$$
Besides, the convexity of $\lone$ gives $\pscal{\hat{f}_{\alpha}}{\hat\alpha - \Talpha}\geq \norm{\hat{\alpha}}[1]-\norm{\Talpha}[1]$, therefore
\begin{multline*}
\left\lbrace\lambda_2-\umax\left(\norm{\diff{\TX}}[\infty]+ 2\smax^2(1+\ae)a\right)\right\rbrace\norm{\hat\alpha}[1]\leq\\
\left\lbrace\lambda_2+\umax\left(\norm{\diff{\TX}}[\infty]+ 2\smax^2(1+\ae)a\right)\right\rbrace\norm{\Talpha}[1],
\end{multline*}
and the condition $\lambda_2\geq 2\left\lbrace\umax\left(\norm{\diff{\TX}}[\infty]+ 2\smax^2(1+\ae)a\right)\right\rbrace$ gives $\norm{\hat \alpha}[1]\leq 3\norm{\Talpha}[1]$ and finally
\begin{equation}
\label{eq:ineq-alpha-l1}
\norm{\Delta\alpha}[1]\leq 4\norm{\Talpha}[1].
\end{equation}
\paragraph{Case 1:} $\norm{\fu{\Delta\alpha}}[\Pi]^2 < 72a^2\log(d)/(p\log(6/5))$. Then the result holds trivially.
\paragraph{Case 2:} $\norm{\fu{\Delta\alpha}}[\Pi]^2 \geq 72a^2\log(d)/(p\log(6/5))$. For $d_1>0$ recall the definition of the set
\begin{equation*}
\label{eq:set-A}
\tilde{\mathcal{A}}(d_1) = \left\lbrace \alpha\in\R^N:\quad \norm{\alpha}[\infty]\leq 1;\quad \norm{\alpha}[1]\leq d_1;\quad \norm{\fu{\alpha}}[\Pi]^2\geq \frac{18\log(d)}{p\log(6/5)}\right\rbrace.
\end{equation*}
Inequality \eqref{eq:ineq-alpha-l1} and $\norm{\Delta\alpha}[\infty]\leq 2a$ imply that
$\Delta \alpha/(2a) \in \tilde{\mathcal{A}}(2\norm{\Talpha}[1]/a).$
Therefore we can apply \Cref{lemma:rsc}(i) and obtain that with probability at least $1-8d^{-1}$,
\begin{equation}
\label{eq:upper-bound-alpha-subres}
\norm{\fu{\Delta\alpha}}[\Pi]^2\leq 2\norm{\fu{\Delta\alpha}}[\Omega]^2 + 64\ae a\norm{\Talpha}[1]\umax\PE[{\norm{\Sigma_R}[\infty]}] + 3072a^2p^{-1}.
\end{equation}
We now must upper bound the quantity $\norm{\fu{\Delta\alpha}}[\Omega]^2$. Recall that $\tilde{X} = \sum_{k=1}^N\Talpha_kU^k+\hat{X}$. By definition, $\mathcal{L}(\hat{X};Y,\Omega)+ \lambda_1\norm{\hat{L}}[*] + \lambda_2\norm{\hat{\alpha}}[1]\leq \mathcal{L}(\tilde{X};Y,\Omega)+ \lambda_1\norm{\hat{L}}[*] + \lambda_2\norm{\Talpha}[1],$ i.e.
\begin{equation*}
\mathcal{L}(\hat{X};Y,\Omega)- \mathcal{L}(\tilde{X};Y,\Omega)  \leq
\lambda_2\left(\norm{\Talpha}[1]-\norm{\hat{\alpha}}[1]\right).
\end{equation*}
Substracting $\pscal{\nabla\mathcal{L}(\tilde{X};Y,\Omega)} {\hat{X} - \tilde{X}}$ on both sides and using the restricted strong convexity (\eqref{eq:strong-cvx}), we obtain
\begin{equation}
\label{eq:th2-subres}
\begin{aligned}
& \frac{\smin^2}{2}\norm{\fu{\Delta \alpha}}[\Omega]^2 && \leq \lambda_2\left(\norm{\Talpha}[1]-\norm{\hat{\alpha}}[1]\right) +\pscal{\nabla  \mathcal{L}(\tilde{X}; Y,\Omega)}{\fu{\Delta\alpha}}\\
& && \leq \lambda_2\left(\norm{\Talpha}[1]-\norm{\hat{\alpha}}[1]\right) + \underbrace{\left|\pscal{\nabla  \mathcal{L}(\TX; Y,\Omega)}{\fu{\Delta \alpha}}\right|}_{\mathsf{C}_1}\\
& && +\underbrace{\left|\pscal{\nabla  \mathcal{L}(\TX; Y,\Omega)-\nabla  \mathcal{L}(\tilde{X}; Y)}{\fu{\Delta \alpha}}\right|}_{\mathsf{C}_2}
\end{aligned}.
\end{equation}
The duality of $\norm{\cdot}[1]$ and $\norm{\cdot}[\infty]$ yields $\mathsf{C}_1 \leq  \norm{\nabla  \mathcal{L}(\TX; Y,\Omega)}[\infty]u\norm{\Delta\alpha}[1]$, and
$$\mathsf{C}_2 \leq  \norm{\nabla  \mathcal{L}(\TX; Y,\Omega)-\nabla  \mathcal{L}(\tilde{X}; Y,\Omega)}[\infty]\umax\norm{\Delta \alpha}[1].$$
Furthermore, $\norm{\nabla  \mathcal{L}(\TX; Y,\Omega)-\nabla  \mathcal{L}(\tilde{X}; Y,\Omega)}[\infty]\leq 2\smax^2a,$
since for all $(i,j)\in\nint{m_1}\times\nint{m_2}$ $|\tilde{X}_{ij} - \TX_{ij}|\leq 2a$ and $g_j''(\tilde{X}_{ij})\leq \smax^2$. The last three inequalities plugged in \eqref{eq:th2-subres} give
\begin{equation*}
\begin{aligned}
& \frac{\smin^2}{2}\norm{\fu{\Delta \alpha}}[\Omega]^2 && \leq \lambda_2\left(\norm{\Talpha}[1]-\norm{\hat{\alpha}}[1]\right) +\umax\norm{\Delta\alpha}[1]\left\lbrace\norm{\nabla  \mathcal{L}(\TX; Y,\Omega)}[\infty]+2\smax^2a\right\rbrace.
\end{aligned}
\end{equation*}
The triangular inequality gives
\begin{equation*}
\begin{aligned}
& \frac{\smin^2}{2}\norm{\fu{\Delta \alpha}}[\Omega]^2 &&  \leq \left\lbrace\umax\left(\norm{\nabla  \mathcal{L}(\TX; Y,\Omega)}[\infty] + 2\smax^2a\right)+\lambda_2\right\rbrace\norm{\Talpha}[1]  \\
& && +\left\lbrace\umax\left(\norm{\nabla  \mathcal{L}(\TX; Y,\Omega)}[\infty] + 2\smax^2a\right)-\lambda_2\right\rbrace\norm{\hat{\alpha}}[1].
\end{aligned}
\end{equation*}
Then, the assumption $\lambda_2\geq 2\umax\left(\norm{\diff{\TX}}[\infty]+ 2\smax^2(1+\ae)a\right)$ gives
$$\norm{\fu{\Delta \alpha}}[\Omega]^2 \leq \frac{3\lambda_2}{\smin^2}\norm{\Talpha}[1].$$
Plugged into \eqref{eq:upper-bound-alpha-subres}, this last inequality implies that with probability at least $1-8d^{-1}$
\begin{equation}
\label{eq:upper-bound-alpha-subres-2}
\norm{\fu{\Delta\alpha}}[\Pi]^2\leq \frac{3\lambda_2}{\smin^2}\norm{\Talpha}[1] + 64\ae a\norm{\Talpha}[1]\umax\PE[{\norm{\Sigma_R}[\infty]}] + 3072a^2p^{-1}.
\end{equation}
Combining \eqref{eq:ineq-alpha-l1} and \eqref{eq:upper-bound-alpha-subres-2} gives the result.
\section{Proof of \Cref{lemma:L-nuc}}
\label{lemma:L-nuc-proof}
Using \eqref{eq:opt-cond} for $L=\TL$ and $\alpha = \Talpha$ we obtain
\begin{equation*}
\label{eq:subdiff-L}
\pscal{\diff{\hat{X}}}{\Delta L+\sum_{k=1}^N(\Delta \alpha_k)U^k}+\lambda_1\pscal{\hat f_L}{\Delta L}+\lambda_2\pscal{\hat f_{\alpha}}{\Delta \alpha}\geq 0.
\end{equation*}
Then, the convexity of $\lnuc$ and $\lone$ imply that
$\norm{\TL}[*]\geq \norm{\hat L}[*]+\pscal{\partial\norm{\hat L}[*]}{\Delta L}$ and
$\norm{\Talpha}[1]\geq \norm{\hat \alpha}[*]+\pscal{\partial\norm{\hat \alpha}[1]}{\Delta \alpha}$.
The last three inequalities yield
\begin{multline*}
\lambda_1\left(\norm{\hat L}[*]-\norm{\TL}[*]\right) + \lambda_2\left(\norm{\hat \alpha}[1]-\norm{\Talpha}[1]\right) \leq \pscal{\diff{\hat{X}}}{\Delta L}\\
+ \pscal{\diff{\hat{X}}}{\sum_{k=1}^N(\Delta \alpha_k)U^k}\\
\leq \norm{\diff{\hat{X}}}\norm{\Delta L}[*] + \umax\norm{\diff{\hat{X}}}[\infty]\norm{\Delta \alpha}[1].
\end{multline*}
Using \eqref{eq:L-nuc} and the conditions $$\lambda_1\geq 2\norm{\diff{\TX}},\quad\lambda_2\geq 2\umax\left\lbrace\norm{\diff{\TX}}[\infty]+ 2\smax^2(1+\ae)a\right\rbrace,$$
we get
\begin{multline*}
\lambda_1\left(\norm{P_{\TL}^{\perp}(\Delta L)}[*]-\norm{P_{\TL}(\Delta L)}[*]\right) + \lambda_2\left(\norm{\hat \alpha}[1]-\norm{\Talpha}[1]\right)\leq\\
 \frac{\lambda_1}{2}\left(\norm{P_{\TL}^{\perp}(\Delta L)}[*]+\norm{P_{\TL}(\Delta L)}[*]\right) + \frac{\lambda_2}{2}\norm{\Delta \alpha}[1],
\end{multline*}
which implies
$\norm{P_{\TL}^{\perp}(\Delta L)}[*] \leq 3\norm{P_{\TL}(\Delta L)}[*] + 3\lambda_2/\lambda_1\norm{\Talpha}[1].$
Now, using $$\norm{\Delta L}[*]\leq \norm{P_{\TL}^{\perp}(\Delta L)}[*]+\norm{P_{\TL}(\Delta L)}[*],\quad \norm{P_{\TL}(\Delta L)}[F]\leq \norm{\Delta L}[F]$$ and $\operatorname{rank}(P_{{\TL}}(\Delta L))\leq 2 r$, we get
$\norm{\Delta L}[*] \leq \sqrt{32r}\norm{\Delta L}[F] + 3\lambda_2/\lambda_1\norm{\Talpha}[1].$
This completes the proof of \Cref{lemma:L-nuc}.

\section{Proof of \Cref{lemma:rsc}}
\label{proof:lemma:rsc}
\paragraph{Proof of (i):} Recall $\mathsf{D}_{\alpha} = 8\ae d_1\umax\PE[{\norm{\Sigma_R}[\infty]}]+768p^{-1}$ and
$$\tilde{\mathcal{A}}(d_1) = \left\lbrace \alpha\in\R^N:\quad \norm{\alpha}[\infty]\leq 1;\quad \norm{\alpha}[1]\leq d_1;\quad \norm{\fu{\alpha}}[\Pi]^2\geq \frac{18\log(d)}{p\log(6/5)}\right\rbrace.
$$
 We will show that the probability of the following event is small:
$$\mathcal{B} = \left\lbrace \exists \alpha\in\tilde{\mathcal{A}}(d_1) \text{ such that } \left|\norm{\fu{\alpha} }[\Omega]^2 -  \norm{\fu{\alpha}}[\Pi]^2\right|>  \frac{1}{2}\norm{\fu{\alpha}}[\Pi]^2+ \mathsf{D}_{\alpha}\right\rbrace.$$
Indeed, $\mathcal{B}$ contains the complement of the event we are interested in. We use a peeling argument to upper bound the probability of event $\mathcal{B}$. Let $\nu = 18\log(d)/(p\log(6/5))$ and $\eta = 6/5$. For $l\in\mathbb{N}$ set
$$\mathcal{S}_l = \left\lbrace \alpha\in\tilde{\mathcal{A}}(d_1):\quad \eta^{l-1}\nu\leq \norm{\fu{\alpha}}[\Pi]^2 \leq \eta^l\nu \right\rbrace.$$
Under the event $\mathcal{B}$, there exists $l\geq 1$ and $\alpha \in \tilde{\mathcal{A}}(d_1)\cap S_l$ such that
\begin{equation}
\label{eq:peel}
\begin{aligned}
 &\left|\norm{\fu{\alpha} }[\Omega]^2 -  \norm{\fu{\alpha}}[\Pi]^2\right|&&> \frac{1}{2}\norm{\fu{\alpha}}[\Pi]^2+ \mathsf{D}_{\alpha}
 > \frac{1}{2}\eta^{l-1}\nu + \mathsf{D}_{\alpha} = \frac{5}{12}\eta^l\nu + \mathsf{D}_{\alpha}.
\end{aligned}
\end{equation}
For $T>\nu$, consider the set of vectors
$$\tilde{\mathcal{A}}(d_1, T) = \left\lbrace \alpha \in \tilde{\mathcal{A}}(d_1):\norm{\fu{\alpha}}[\Pi]^2\leq T\right \rbrace $$
and the event
$$\mathcal{B}_l = \left\lbrace \exists \alpha\in \tilde{\mathcal{A}}(d_1, \eta^l\nu): \left|\norm{\fu{\alpha}}[\Omega]^2 -  \norm{\fu{\alpha}}[\Pi]^2\right|>\frac{5}{12}\eta^l\nu + \mathsf{D}_{\alpha} \right\rbrace.$$
If $\mathcal{B}$ holds, then \eqref{eq:peel} implies that $\mathcal{B}_l$ holds for some $l\leq 1$. Therefore ,$\mathcal{B}\subset \cup_{l=1}^{+\infty} \mathcal{B}_l$, and it is enough to estimate the probability of the events $\mathcal{B}_l$ and then apply the union bound. Such an estimation is given in the following lemma, adapted from Lemma 10 in \cite{klopp:hal-01111757}.
\begin{lemma}
\label{lemma:ZT-concentration}
Define
$Z_T = {\sup}_{\alpha\in\tilde{\mathcal{A}}(d_1, T)}\left|\norm{\fu{\alpha}}[\Omega]^2 -  \norm{\fu{\alpha}}[\Pi]^2\right|.$ Then,
$$\prob{Z_T \geq \mathsf{D}_{\alpha}+\frac{5}{12}T}\leq 4\e^{-pT/18}.$$
\end{lemma}
\begin{proof}
By definition,
$$Z_T = {\sup}_{\alpha\in\tilde{\mathcal{A}}(d_1, T)}\left|\sum_{(i,j)}\Omega_{ij}\fu{\alpha}_{ij}^2 - \PE[{\sum_{(i,j)}\Omega_{ij}\fu{\alpha}_{ij}^2}]\right|.$$
We use the following Talagrand's concentration inequality, proven in \cite{talagrand1996} and \cite{Chat2014univ}.
\begin{lemma}
\label{lemma:talagrand}
Assume $f:[-1,1]^n\mapsto \R$ is a convex Lipschitz function with Lipschitz constant L. Let $\Xi_1,\ldots,\Xi_n$ be independent random variables taking values in $[-1,1]$. Let $Z := f(\Xi_1,\ldots,\Xi_n)$. Then, for any $t\geq 0$,
$\prob{\left|Z-\PE[{Z}]\right|\geq 16L+t}\leq 4\e^{-t^2/2L^2}.$
\end{lemma}
We apply this result to the function $$f(x_{11},\ldots,x_{m_1m_2}) = {\sup}_{\alpha\in\tilde{\mathcal{A}}(d_1, T)}\left|\sum_{(i,j)}(x_{ij}-\pi_{ij})\fu{\alpha}_{ij}^2\right|,$$ which is Lipschitz with Lipschitz constant $\sqrt{p^{-1}T}$. Indeed, for any $(x_{11},\ldots,x_{m_1m_2})\in\R^{m_1\times m_2}$ and $(z_{11},\ldots,z_{m_1m_2})\in\R^{m_1\times m_2}$:
\begin{equation*}
\begin{aligned}
& \left|f(x_{11},\ldots,x_{m_1m_2})-f(z_{11},\ldots,z_{m_1m_2})\right|\\
& \quad = \left|{\sup}_{\alpha\in\tilde{\mathcal{A}}(d_1, T)}\left|\sum_{(i,j)}(x_{ij}-\pi_{ij})\fu{\alpha}_{ij}^2\right|-{\sup}_{\alpha\in\tilde{\mathcal{A}}(d_1, T)}\left|\sum_{(i,j)}(z_{ij}-\pi_{ij})\fu{\alpha}_{ij}^2\right| \right|\\
& \quad \leq {\sup}_{\alpha\in\tilde{\mathcal{A}}(d_1, T)}\left|\left|\sum_{(i,j)}(x_{ij}-\pi_{ij})\fu{\alpha}_{ij}^2\right|-\left|\sum_{(i,j)}(z_{ij}-\pi_{ij})\fu{\alpha}_{ij}^2\right| \right|\\
&\quad \leq {\sup}_{\alpha\in\tilde{\mathcal{A}}(d_1, T)}\left|\sum_{(i,j)}(x_{ij}-\pi_{ij})\fu{\alpha}_{ij}^2-\sum_{(i,j)}(z_{ij}-\pi_{ij})\fu{\alpha}_{ij}^2 \right|\\
&\quad \leq{\sup}_{\alpha\in\tilde{\mathcal{A}}(d_1, T)}\left|\sum_{(i,j)}(x_{ij}-z_{ij})\fu{\alpha}_{ij}^2 \right|\\
&\quad \leq{\sup}_{\alpha\in\tilde{\mathcal{A}}(d_1, T)}\sqrt{\sum_{(i,j)}\pi_{ij}^{-1}(x_{ij}-z_{ij})^2}\sqrt{\sum_{(i,j)}\pi_{ij}\fu{\alpha}_{ij}^4}\\
&\quad \leq{\sup}_{\alpha\in\tilde{\mathcal{A}}(d_1, T)}\sqrt{p^{-1}}\sqrt{\sum_{(i,j)}(x_{ij}-z_{ij})^2}\sqrt{\sum_{(i,j)}\pi_{ij}\fu{\alpha}_{ij}^2}\\
& \quad \leq \sqrt{p^{-1}T}\sqrt{\sum_{(i,j)}(x_{ij}-z_{ij})^2},
\end{aligned}\\
\end{equation*}
where we used $||a|-|b||\leq |a-b|$,$\norm{\fu{\alpha}}[\infty]\leq 1$ and $\norm{A}[\Pi]^2\leq T$. Thus, \Cref{lemma:talagrand} and the identity $\sqrt{p^{-1}T}\leq \frac{96p^{-1}}{2} + \frac{T}{2\times 96}$ imply
$$\prob{\left|Z-\PE[{Z}]\right|\geq 768p^{-1}+\frac{1}{12}T+t}\leq 4e^{-t^2p/2T}.$$
Taking $t = T/3$ we get
\begin{equation}
\label{eq:concentration-ZT}
\prob{\left|Z-\PE[{Z}]\right|\geq 768p^{-1}+\frac{5}{12}T}\leq 4e^{-pT/18}.
\end{equation}
Now we must bound the expectation $\PE[{Z_T}]$.  To do so, we use a symmetrization argument \citep{Ledoux2001} which gives
\begin{equation*}
\begin{aligned}
& \PE[{Z_T}] && = \PE[{{\sup}_{\alpha\in\tilde{\mathcal{A}}(d_1, T)}\left|\sum_{(i,j)}\Omega_{ij}\fu{\alpha}_{ij}^2 - \PE{\sum_{(i,j)}\Omega_{ij}\fu{\alpha}_{ij}^2}\right|}]\\
& && \leq 2\PE[{{\sup}_{\alpha\in\tilde{\mathcal{A}}(d_1, T)}\left|\sum_{(i,j)}\epsilon_{ij}\Omega_{ij}\fu{\alpha}_{ij}^2 \right|}],
\end{aligned}
\end{equation*}
where $\{\epsilon_{ij}\}$ is an i.i.d. Rademacher sequence independent of $\left\{\Omega_{ij}\right\}$. We apply an extension Talagrand's contraction inequality to Lipschitz functions (see \cite{Koltchinskii2011}, Theorem 2.2) and obtain
\begin{multline*}
\PE[{Z_T}] = \PE[{\sup_{A\in\mathcal{T}}\left|\sum_{i,j}\epsilon_{ij}\Omega_{ij}A_{ij}^2\right|}]\leq 4\ae\PE[{{\sup}_{\alpha\in\tilde{\mathcal{A}}(d_1, T)}\left|\sum_{(i,j)}\epsilon_{ij}\Omega_{ij}A_{ij} \right|}]\\ = 4\ae\PE[{{\sup}_{\alpha\in\tilde{\mathcal{A}}(d_1, T)}\left|\pscal{\Sigma_R}{\fu{\alpha}}\right|}],
\end{multline*}
where $\Sigma_R = \sum_{(i,j)}\epsilon_{ij}\Omega_{ij}E_{ij}$. Moreover, for $\alpha\in \tilde{\mathcal{A}}(d_1, T)$ we have
\begin{equation*}
\begin{aligned}
&\left|\pscal{\Sigma_R}{\fu{\alpha}}\right| &&= \left|\pscal{\Sigma_R}{\sum_{k=1}^N\alpha_kU^k}\right|\leq \norm{\alpha}[1]\umax\norm{\Sigma_R}[\infty].
\end{aligned}
\end{equation*}
Finally, we get $\PE[{Z_T}]\leq 4\ae d_1 \umax\PE[{\norm{\Sigma_R}[\infty]}].$
Combining this with the concentration inequality \eqref{eq:concentration-ZT} we complete the proof of \Cref{lemma:ZT-concentration}:
$$\prob{Z_T \geq 8\ae d_1\umax\PE[{\norm{\Sigma_R}[\infty]}]+768p^{-1}+\frac{5}{12}T}\leq 4\e^{-pT/18}.$$
\end{proof}
\Cref{lemma:ZT-concentration} gives that $\prob{\mathcal{B}_l}\leq 4\exp(-p\eta^l\nu/18)$. Applying the union bound we obtain
\begin{equation*}
\begin{aligned}
&\prob{\mathcal{B}} && \leq \sum_{l=1}^{\infty}\prob{\mathcal{B}_l}\leq 4\sum_{l=1}^{\infty}\exp(-p\eta^l\nu/18)\\
& && \leq 4\sum_{l=1}^{\infty}\exp(-p\log(\eta)l\nu/18),
\end{aligned}
\end{equation*}
where we used $e^x\geq x$. Finally, for $\nu = 18\log(d)/(p\log(6/5))$ we obtain
$$\prob{\mathcal{B}}\leq \frac{4\exp(-p\nu\log(\eta)/18)}{1-\exp(-p\nu\log(\eta)/18)}\leq \frac{4\exp(-\log(d))}{1-\exp(-\log(d))}\leq \frac{8}{d}, $$
since $d-1\geq d/2$, which concludes the proof of (i).

\paragraph{Proof of (ii):} The proof is similar to that of (i); we recycle some of the notations for simplicity. Recall
$\mathsf{D}_X=112\rho p^{-1} \PE[{\norm{\Sigma_R}}]^2 + 8\ae\varepsilon \PE[{\norm{\Sigma_R}}]+8\ae d_1\umax\PE[{\norm{\Sigma_R}[\infty]}] +d_{\Pi}+768p^{-1},$
and let\begin{multline*}
\mathcal{B} = \Big\lbrace \exists (L,\alpha)\in\mathcal{C}(d_1,d_{\Pi},\rho,\varepsilon);\\
 \left|\norm{L+\fu{\alpha} }[\Omega]^2 -  \norm{L+\fu{\alpha}}[\Pi]^2\right|>  \frac{1}{2}\norm{L+\fu{\alpha}}[\Pi]^2+ \mathsf{D}_X\Big\rbrace,
\end{multline*}
 $\nu = 72\log(d)/(p\log(6/5))$, $\eta = 6/5$ and for $l\in\mathbb{N}$ $$\mathcal{S}_l = \left\lbrace (L,\alpha)\in\mathcal{C}(d_1,d_{\Pi},\rho,\varepsilon):\quad \eta^{l-1}\nu\leq \norm{L+\fu{\alpha}}[\Pi]^2 \leq \eta^l\nu \right\rbrace.$$
As before, if $\mathcal{B}$ holds, then there exist $l\geq 2$ and $(L,\alpha)\in\mathcal{C}(d_1,d_{\Pi},\rho,\varepsilon)\cap S_l$ such that
\begin{equation}
\label{eq:peel2}
\begin{aligned}
 &\left|\norm{L+\fu{\alpha} }[\Omega]^2 -  \norm{L+\fu{\alpha}}[\Pi]^2\right|&&> \frac{5}{12}\eta^l\nu + \mathsf{D}_X.
\end{aligned}
\end{equation}
For $T>\nu$, consider the set $\tilde{\mathcal{C}}(T) = \left\lbrace (L,\alpha) \in \mathcal{C}(d_1,d_{\Pi},\rho,\varepsilon): \norm{L+\fu{\alpha}}[\Pi]^2\leq T\right \rbrace $, and the event
$$\mathcal{B}_l = \left\lbrace \exists (L,\alpha)\in \tilde{\mathcal{C}}(\eta^l\nu):\quad \left|\norm{L+\fu{\alpha}}[\Omega]^2 -  \norm{L+\fu{\alpha}}[\Pi]^2\right|>\frac{5}{12}\eta^l\nu + \mathsf{D}_X \right\rbrace.$$
Then, \eqref{eq:peel2} implies that $\mathcal{B}_l$ holds and $\mathcal{B}\subset \cup_{l=1}^{+\infty} \mathcal{B}_l$. Thus, we estimate in \Cref{lemma:ZT-concentration2} the probability of the events $\mathcal{B}_l$, and then apply the union bound.
\begin{lemma}
\label{lemma:ZT-concentration2}
Let $W_T = {\sup}_{(L,\alpha)\in\tilde{\mathcal{C}}(T)}\left|\norm{L+\fu{\alpha}}[\Omega]^2 -  \norm{L+\fu{\alpha}}[\Pi]^2\right|.$
$$\prob{W_T \geq \mathsf{D}_X+\frac{5}{12}T}\leq 4\e^{-pT/72}.$$
\end{lemma}
\begin{proof}
The proof is two-fold: first we show that $W_T$ concentrates around its expectation, then bound its expectation.
By definition,
$$W_T = {\sup}_{(L,\alpha)\in\tilde{\mathcal{C}}(T)}\left|\sum_{(i,j)}\Omega_{ij}(L_{ij}+\fu{\alpha}_{ij})^2 - \PE[{\sum_{(i,j)}\Omega_{ij}(L_{ij}+\fu{\alpha}_{ij})^2}]\right|.$$
The concentration proof is exactly similar to the proof in \Cref{lemma:ZT-concentration}, but we choose $t = T/6$, and we obtain
\begin{equation}
\label{eq:concentration-ZT2}
\prob{\left|W_T-\PE[{W_T}]\right|\geq 768p^{-1}+\frac{3}{12}T}\leq 4e^{-pT/72}.
\end{equation}
Let us now bound the expectation $\PE[{W_T}]$. Again, we use a standard symmetrization argument \citep{Ledoux2001} which gives
\begin{equation*}
\begin{aligned}
& \PE[{W_T}] && \leq 2\PE[{{\sup}_{(L,\alpha)\in\tilde{\mathcal{C}}(T)}\left|\sum_{(i,j)}\epsilon_{ij}\Omega_{ij}(L_{ij}+\fu{\alpha}_{ij})^2 \right|}],
\end{aligned}
\end{equation*}
where $\{\epsilon_{ij}\}$ is an i.i.d. Rademacher sequence independent of $\Omega_{ij}$. Then, the contraction inequality (see \cite{Koltchinskii2011}, Theorem 2.2) yields
\begin{equation*}
\PE[{W_T}]\leq 4\ae\PE[{{\sup}_{(L,\alpha)\in\tilde{\mathcal{C}}(T)}\left|\pscal{\Sigma_R}{L+\fu{\alpha}}\right|}],
\end{equation*}
where $\Sigma_R = \sum_{(i,j)}\epsilon_{ij}\Omega_{ij}E_{ij}$. Moreover
\begin{equation*}
\begin{aligned}
&\left|\pscal{\Sigma_R}{L+\fu{\alpha}}\right| &&\leq  \left|\pscal{\Sigma_R}{L}\right|+\left|\pscal{\Sigma_R}{\fu{\alpha}}\right|\\
& && \leq \norm{L}[*]\norm{\Sigma_R}+\norm{\alpha}[1]\umax\norm{\Sigma_R}[\infty].
\end{aligned}
\end{equation*}
For $(L,\alpha)\in \tilde{\mathcal{C}}(T)$ we have by assumption $\norm{\alpha}[1]\leq d_1 $, $\norm{\fu{\alpha}}[\Pi]\leq \sqrt{d_{\Pi}} $ and $\norm{L}[*]\leq \sqrt{\rho}\norm{L}[F]+\varepsilon$. We obtain
\begin{equation*}
\begin{aligned}
& \norm{L}[*] && \leq \sqrt{\frac{\rho}{p}}\norm{L}[\Pi]+\varepsilon \leq \sqrt{\frac{\rho}{p}}\left(\norm{L+\fu{\alpha}}[\Pi]+\norm{\fu{\alpha}}[\Pi]\right)+\varepsilon\\
& && \leq \sqrt{\frac{\rho}{p}}\left(\sqrt{T}+\sqrt{d_{\Pi}}\right)+\varepsilon.
\end{aligned}
\end{equation*}
This gives
\begin{equation*}
\begin{aligned}
&\PE[{W_T}] &&\leq 4\ae\left\lbrace\sqrt{\frac{\rho}{p}}\left(\sqrt{T}+\sqrt{d_{\Pi}}\right) +\varepsilon\right\rbrace\norm{\Sigma_R} + 4\ae d_1\umax\norm{\Sigma_R}[\infty]\\
& && \leq \frac{T}{12} + \frac{d_{\Pi}}{2} + 56\ae^2\frac{\rho}{p}\norm{\Sigma_R}^2 + 4\ae\varepsilon\norm{\Sigma_R} + 4\ae d_1\umax\norm{\Sigma_R}[\infty].
\end{aligned}
\end{equation*}
Combining this with the concentration inequality \eqref{eq:concentration-ZT2} we finally obtain:
\begin{equation*}
\prob{W_T \geq \mathsf{D}_X+\frac{5}{12}T}\leq 4\e^{-pT/72}.
\end{equation*}
\end{proof}

\Cref{lemma:ZT-concentration2} gives that $\prob{\mathcal{B}_l}\leq 4\exp(-p\eta^l\nu/72)$. Applying the union bound we obtain
\begin{equation*}
\begin{aligned}
&\prob{\mathcal{B}} && \leq \sum_{l=1}^{\infty}\prob{\mathcal{B}_l} \leq 4\sum_{l=1}^{\infty}\exp(-p\eta^l\nu/72)\\
& && \leq 4\sum_{l=1}^{\infty}\exp(-p\log(\eta)l\nu/72),
\end{aligned}
\end{equation*}
where we used $e^x\geq x$. Finally, for $\nu = 72\log(d)/(p\log(6/5))$ we obtain
$$\prob{\mathcal{B}}\leq \frac{4\exp(-p\nu\log(\eta)/72)}{1-\exp(-p\nu\log(\eta)/72)}\leq \frac{4\exp(-\log(d))}{1-\exp(-\log(d))}\leq 8d^{-1}, $$
since $d-1\geq d/2$, which concludes the proof of (ii).

\section{Proof of \Cref{lemma:SigmaR}}
\label{proof:lemma:SigmaR}
The first inequality is trivially true using that $\norm{\Sigma}[\infty] = \max_{i,j}|\Omega_{ij}\epsilon_{ij}|\leq 1$. We prove the second inequality using an extension to rectangular matrices via self-adjoint dilation of Corollary 3.3 in \cite{bandeira2016}.
\begin{proposition}
\label{prop:sigmaR-nuc}
Let $A$ be an $m_1\times m_2$ rectangular matrix with $A_{ij}$ independent centered bounded random variables. then, there exists a universal constant $C^*$ such that
$$\PE[{\norm{A}}] \leq C^*\left\{\sigma_1\vee \sigma_2 + \sigma_*\sqrt{\log(m_1\wedge m_2)}\right\},$$
\begin{equation*}
\sigma_1 = \max_i\sqrt{\sum_j\PE[{A_{ij}^2}]},\quad
\sigma_2 = \max_j\sqrt{\sum_i\PE[{A_{ij}^2}]},\quad
\sigma_* = \max_{i,j}|A_{ij}|.
\end{equation*}
\end{proposition}
Applying \Cref{prop:sigmaR-nuc} to $\Sigma_R$ with $\sigma_1 \vee \sigma_2 \leq \sqrt{\beta}$ and $\sigma_*\leq 1$ we obtain
$$\PE[{\norm{\Sigma_R}}] \leq C^*\left\{\sqrt{\beta} + \sqrt{\log(m_1\wedge m_2)}\right\}.$$

\section{Proof of \Cref{lemma:Sigma}}
\label{proof:lemma:Sigma}
Denote $\Sigma = \nabla \mathcal{L}(\TX; Y,\Omega)$.
Definition \eqref{exp-model} implies that $\PE[Y_{ij}] = g_j'(\TX_{ij})$, $(i,j)\in\nint{m_1}\times\nint{m_2}$. Combined with the sub-exponentiality of the entries $Y_{ij}$, we obtain that for all $i,j$, $Y_{ij}-g_j'(\TX_{ij})$ is sub-exponential with scale and variance parameters $1/\gamma$ and $\smax^2$ respectively. Then, noticing that $|\Omega_{ij}|\leq 1$ implies that for all $t\geq 0$, $$\mathbb{P}\left\{\left|\Omega_{ij}\left(Y_{ij}-g_j'(\TX_{ij})\right)\right|\geq t\right\}\leq \mathbb{P}\left\{\left|Y_{ij}-g_j'(\TX_{ij})\right|\geq t\right\},$$
we obtain that the random variables $\Sigma_{ij}=\Omega_{ij}\left(Y_{ij}-g_j^{'}(\TX_{ij})\right)$ are also sub-exponential. Thus, for all $i,j$ and for all $t\geq 0$ we have that
$|\Sigma_{ij}|\leq t$ with probability at least $1-\max\left\{2e^{-t^2/2\smax^2}, 2e^{-\gamma t/2} \right\} $. A union bound argument then yields
$$\norm{\Sigma}[\infty]\leq t \quad \text{w. p. at least } 1-\max\left\{2m_1m_2e^{-t^2/2\smax^2}, 2m_1m_2e^{-\gamma t/2} \right\},$$
where $\gamma$ and $\smax$ are defined in \textbf{H}\ref{ass:cvx}. Using $\log(m_1m_2)\leq 2\log d$, where $d = m_1 + m_2$ and setting $t = 6\max \left\{\smax\sqrt{\log d},  \gamma^{-1}\log d\right\},$ we obtain that with probability at least $1-d^{-1}$,
$$\norm{\Sigma}[\infty]\leq 6\max \left\{\smax\sqrt{\log d},  \gamma^{-1} \log d\right\},$$
which proves the first inequality.
Now we prove the second inequality using the following result obtained by extension of Theorem 4 in \cite{Tropp2012} to rectangular matrices.
\begin{proposition}
\label{prop:klopp2014}
Let $W_{1},\ldots , W_n$ be independent random matrices with dimensions $m_1\times m_2$ that satisfy $\PE[{W_i}]=0$. Suppose that
\begin{equation}
\delta_* = \sup_{i\in\nint{n}}\inf_{\delta>0} \left\lbrace\PE[\exp\left(\norm{W_i}/\delta\right)]\leq \e\right\rbrace<+\infty.
\end{equation}
 Then, there exists an absolute constant $c^*$ such that, for all $t>0$ and with probability at least $1-\e^{-t}$ we have
$$\left\|\frac{1}{n}\sum_{i=1}^nW_i\right\|\leq c^*\max\left\lbrace \sigma_W\sqrt{\frac{t+\log d}{n}}, \delta_*\left(\log\frac{\delta_*}{\sigma_W}\right)\frac{t+\log d}{n} \right\rbrace,$$
where
$$\sigma_W = \max\left\lbrace \left\|\frac{1}{n}\sum_{i=1}^{n}\PE[W_{i}W_{i}^{\top}]\right\|^{1/2},\left\|\frac{1}{n}\sum_{i=1}^{n}\PE[W_{i}^{\top}W_{i}]\right\|^{1/2} \right \rbrace.$$
\end{proposition}
For all $(i,j)\in\nint{m_1}\times\nint{m_2}$ define $Z_{ij} = -\Omega_{ij}\left(Y_{ij}-g_j^{'}(\TX_{ij})\right)E_{ij}.$ The sub-exponentiality of the variables $\Omega_{ij}\left(Y_{ij}-g_j^{'}(\TX_{ij})\right)$ implies that for all $i,j\in\nint{m_1}\times \nint{m_2}$
$$\delta_{ij}={\operatorname{inf}}_{\delta>0}\quad \left\lbrace\PE[\exp\left(\left|\Omega_{ij}\left(Y_{ij}-g_j^{'}(\TX_{ij})\right)\right|/\delta\right)]\leq \e\right\rbrace\leq \frac{1}{\gamma}.$$
We can therefore apply \Cref{prop:klopp2014} to the matrices $Z_{ij}$ defined above, with the quantity
\begin{equation}
\label{eq:sigma-def-tail}
\sigma_Z = \max\Bigg\lbrace \left\|\frac{1}{m_1m_2}\sum_{i=1}^{m_1}\sum_{j=1}^{m_2}\PE[{Z_{ij}Z_{ij}^{\top}}]\right\|^{1/2},
\left\|\frac{1}{m_1m_2}\sum_{i=1}^{m_1}\sum_{j=1}^{m_2}\PE[{Z_{ij}^{\top}Z_{ij}}]\right\|^{1/2} \Bigg \rbrace.
\end{equation}
We obtain that for all $t\geq 0$ and with probability at least $1-\e^{-t}$,
\begin{equation*}
\norm{\Sigma } \leq c^*\max\Bigg\lbrace \sigma_Z\sqrt{m_1m_2(t+\log d)},
 \left(\log\frac{1}{\gamma\sigma_Z}\right)\frac{t+\log d}{\gamma } \Bigg\rbrace.
\end{equation*}
We bound $\sigma_Z$ from above and below as follows.
\begin{equation*}
\sum_{i=1}^{m_1}\sum_{j=1}^{m_2}\PE[{Z_{ij}Z_{ij}^{\top}}] = \sum_{i=1}^{m_1}\left\lbrace\sum_{j=1}^{m_2}\PE[{\Omega_{ij}^2}]\PE[{\left(Y_{ij}-g_j^{'}(\TX_{ij})\right)^2}]\right\rbrace E_{ii}(m_1),
\end{equation*}
where $E_{ii}(n)$, $i,n\geq 1$ denotes the $n\times n$ square matrix with $1$ in the $(i,i)$-th entry and zero everywhere else. Therefore
\begin{equation*}
\left\|\frac{1}{m_1m_2}\sum_{i=1}^{m_1}\sum_{j=1}^{m_2}\PE[{Z_{ij}Z_{ij}^{\top}}]\right\|^{1/2} = \sqrt{\frac{1}{m_1m_2}\max_i\sum_{j=1}^{m_2}\PE[{\Omega_{ij}^2}]\PE[{\left(Y_{ij}-g_j^{'}(\TX_{ij})\right)^2}]}.
\end{equation*}
Then, assumption \textbf{H}\ref{ass:cvx} gives
\begin{equation*}
\left\|\frac{1}{m_1m_2}\sum_{i=1}^{m_1}\sum_{j=1}^{m_2}\PE[{Z_{ij}Z_{ij}^{\top}}]\right\|^{1/2} \leq \smax\sqrt{\frac{1}{m_1m_2}\left(\max_i\sum_{j=1}^{m_2}\PE[{\Omega_{ij}^2}]\right)},
\end{equation*}
and
\begin{equation*}
\left\|\frac{1}{m_1m_2}\sum_{i=1}^{m_1}\sum_{j=1}^{m_2}\PE[{Z_{ij}Z_{ij}^{\top}}]\right\|^{1/2} \geq \smin\sqrt{\frac{1}{m_1m_2}\left(\max_i\sum_{j=1}^{m_2}\PE[{\Omega_{ij}^2}]\right)}.
\end{equation*}
Similarly, we obtain
\begin{equation*}
\left\|\frac{1}{m_1m_2}\sum_{i=1}^{m_1}\sum_{j=1}^{m_2}\PE[{Z_{ij}^{\top}Z_{ij}}]\right\|^{1/2} \leq \smax\sqrt{\frac{1}{m_1m_2}\left(\max_j\sum_{i=1}^{m_1}\PE[{\Omega_{ij}^2}]\right)},
\end{equation*}
and
\begin{equation*}
\left\|\frac{1}{m_1m_2}\sum_{i=1}^{m_1}\sum_{j=1}^{m_2}\PE[{Z_{ij}^{\top}Z_{ij}}]\right\|^{1/2} \geq \smin\sqrt{\frac{1}{m_1m_2}\left(\max_j\sum_{i=1}^{m_1}\PE[{\Omega_{ij}^2}]\right)}.
\end{equation*}
Combining the last four inequalities, we obtain
$$\smin\sqrt{\frac{\beta}{m_1m_2}}\leq \sigma_Z \leq \smax\sqrt{\frac{\beta}{m_1m_2}},$$
and setting $t = \log d$, we further obtain for all $t\geq 0$ and with probability at least $1-d^{-1}$:
$$\norm{\Sigma } \leq c^*\max\left\lbrace \smax\sqrt{2\beta\log d}, \frac{2\log d}{\gamma }\log\left(\frac{1}{\smin}\sqrt{\frac{m_1m_2}{\beta}}\right) \right\rbrace,$$
which proves the result.

\end{document}